\documentclass[11pt,3p,review,authoryear]{elsarticle}

\journal{International Journal of Forecasting}

\usepackage[]{graphicx}\usepackage[]{color}
\makeatletter
\def\maxwidth{ %
  \ifdim\Gin@nat@width>\linewidth
    \linewidth
  \else
    \Gin@nat@width
  \fi
}

\makeatother

\definecolor{fgcolor}{rgb}{0.345, 0.345, 0.345}

\usepackage{framed}
\makeatletter
 {\par\unskip\endMakeFramed%
 \at@end@of@kframe}
\makeatother

\definecolor{shadecolor}{rgb}{.97, .97, .97}
\definecolor{messagecolor}{rgb}{0, 0, 0}
\definecolor{warningcolor}{rgb}{1, 0, 1}
\definecolor{errorcolor}{rgb}{1, 0, 0}
\newenvironment{knitrout}{}{} 

\usepackage{rotating}
\usepackage{array,multirow,graphicx}
\usepackage{tabularx}
\usepackage{alltt}
\usepackage[T1]{fontenc}
\usepackage[latin1]{inputenc}
\usepackage{longtable}

\usepackage{wrapfig}
\usepackage{tikz}
\usetikzlibrary{calc}
\usetikzlibrary{matrix}
\usepackage{charter}
\usepackage{url}
\usepackage{t1enc}
\usepackage{amsmath}
\usepackage{amsfonts}
\usepackage{amsthm}
\usepackage[authoryear]{natbib}
\theoremstyle{definition}

\usepackage{rotating}
\usepackage{amssymb}

\usepackage{graphicx}
\usepackage{enumerate}
\usepackage{pdflscape}

\newtheorem{theorem}{Theorem}
\numberwithin{theorem}{section} 
\newtheorem{lemma}{Lemma}
\numberwithin{lemma}{section}

\newtheorem{definition}{Definition}
\numberwithin{definition}{section} 

\newtheorem{remark}{Corollary}
\numberwithin{remark}{section} 

\newtheorem{remark2}{Remark}
\numberwithin{remark2}{section}

\newcommand{\vect}{\boldsymbol} 
\newcommand{\vectoriz}{\texttt{vec}} 
\newcommand{\Sigr}{\vect{\Sigma}_r}
\newcommand{\SigrHat}{\vect{\hat{\Sigma}}_r}

\newcommand{\SigrHatML}{\vect{\hat{\Sigma}}_{r,\text{ML}}}
\newcommand{\SigrHatREML}{\vect{\hat{\Sigma}}_{r,\text{REML}}}
\newcommand{\SigrHatMAP}{\vect{\hat{\Sigma}}_{r,\text{MAP}}}
\newcommand{\SigrHatsREML}{\vect{\hat{\Sigma}}_{r,\text{sREML}}}
\newcommand{\SigrHatPar}{\vect{\hat{\Sigma}}_{r,\text{par}}}
\newcommand{\SigrHatShrink}{\vect{\hat{\Sigma}}_{r,\text{shrink}}}
\newcommand{\Sigh}{\vect{\Sigma}_h}
\newcommand{\SighHat}{\vect{\hat{\Sigma}}_h}
\newcommand{\SigsHat}{\vect{\hat{\Sigma}}_s}
\newcommand{\SighHatd}{\vect{\hat{\Sigma}}_h^{\text{d}}}
\newcommand{\SighHatdT}{\vect{\hat{\Sigma}}_{h,\text{top}}^{\text{d}}}
\newcommand{\SighHatdB}{\vect{\hat{\Sigma}}_{h,\text{bot}}^{\text{d}}}
\newcommand{\yhatTt}{\vect{\hat{y}}_{\text{top},t}}
\newcommand{\yhatBt}{\vect{\hat{y}}_{\text{bot},t}}
\newcommand{\yhatT}{\vect{\hat{y}}_{\text{top}}}
\newcommand{\ytildeFt}{\vect{\tilde{y}}_{\text{full},t}}
\newcommand{\yhatB}{\vect{\hat{y}}_{\text{bot}}}
\newcommand{\YhatB}{\vect{\hat{Y}}_{\text{bot}}}
\newcommand{\YhatT}{\vect{\hat{Y}}_{\text{top}}}
\newcommand{\YT}{\vect{Y}_{\text{top}}}
\newcommand{\ST}{\vect{S}_{\text{top}}}
\newcommand{\PT}{\vect{P}_{\text{top}}}
\newcommand{\PB}{\vect{P}_{\text{bot}}}
\newcommand{\betaT}{\vect{\beta}_{\text{top}}}
\newcommand{\betaTi}{\vect{\beta}_{\text{top},i}}
\newcommand{\betaB}{\vect{\beta}_{\text{bot}}}
\newcommand{\betaOT}{\vect{\beta}_{0,\text{top}}}
\newcommand{\betaHatT}{\vect{\hat{\beta}}_{\text{top}}}
\newcommand{\betaHatB}{\vect{\hat{\beta}}_{\text{bot}}}
\newcommand{\betaHatTML}{\vect{\hat{\beta}}_{\text{top},\text{ML}}}
\newcommand{\betaHatTi}{\vect{\hat{\beta}}_{\text{top},i}}
\newcommand{\eB}{\vect{e}_{\text{bot}}}

\newcommand{\update}{\textcolor{black}}
\newcommand{\rev}{\textcolor{black}}
\newcommand{\revTwo}{\textcolor{black}}
\newcommand{\revThree}{\textcolor{black}}
\DeclareMathOperator*{\argmax}{arg\,max}
\renewcommand{\qed}{$\hfill\blacksquare$}

\IfFileExists{upquote.sty}{\usepackage{upquote}}{}

\usepackage{color}
\usepackage{adjustbox}
\newcommand{\kibitz}[2]{\ifnum\Comments=1\textcolor{#1}{#2}\fi}

\begin{document}

\begin{frontmatter}

\title{Forecast Reconciliation with Uncertainty Quantification}

   \author[dtu]{Jan Kloppenborg M\o ller\corref{cor}}
  \ead{jkmo@dtu.dk}
 
  \author[dtu,quent]{Peter Nystrup}
 
  \author[dtu]{Poul G. Hjorth}
 
  \author[dtu]{Henrik Madsen}
 
  \address[dtu]{Department of Applied Mathematics and Computer Science, Technical University of Denmark, Denmark}
   \address[quent]{QUENT ApS, Denmark}
 
  \cortext[cor]{Corresponding author}






\begin{abstract}
  We propose a new method for estimating the weight matrix used in forecast reconciliation by treating it as parameters in a general linear model. This enables quantification of weight-matrix uncertainty, which propagates to the variance of the reconciled forecasts. We show that forecast reconciliation can be formulated as an orthogonal projection from the space of base-forecast errors onto a coherent linear subspace. Using variance decomposition, we derive a central estimator for the forecast-error covariance matrix and prove that distance-reducing properties \revTwo{hold for} reconciled forecasts at all levels of the hierarchy\revTwo{, as well as for} the forecast-error covariance. We \revTwo{further} show how shrinkage can be introduced \revTwo{through} priors on the weight matrix and the forecast-error covariance matrix. A simulation study demonstrates consistent improvements in log-score. Finally, \revTwo{we illustrate} standard errors for the weight matrix and the variance decomposition formula in a case study on electricity load forecasting in Sweden.
 
\end{abstract}

\begin{keyword}
Forecast reconciliation\sep general linear model\sep orthogonal projection\sep restricted maximum likelihood\sep
maximum a posteriori estimation\sep shrinkage estimation
\end{keyword}

\end{frontmatter}

\thispagestyle{empty}

\section{Introduction}

Forecast reconciliation is an effective \revTwo{approach for} ensuring coherence across hierarchies defined by linear constraints. Forecasts are coherent when they fulfill these constraints. Classical examples include aggregated sales forecasts for individual stores that must align with regional and national sales forecasts \citep[structural reconciliation, see, e.g.,][]{athanasopoulos2009hierarchical} and quarterly forecasts that must align with biannual and annual forecasts \citep[temporal reconciliation, see, e.g.,][]{athanasopoulos2017forecasting}. In addition to ensuring coherency, reconciliation often improves forecast accuracy at all levels of a hierarchy.

\cite{pritularga2021stochastic} argued that the effect of uncertainty in forecast reconciliation has been overlooked. By decomposing the variance of the reconciled forecasts into \revTwo{its} different sources, they showed that uncertainties propagate from \revTwo{estimation of the} variance--covariance matrix to the reconciliation weights, thereby increasing the uncertainty of the reconciled forecasts. Multiple studies have found that more complete approximations of the variance--covariance matrix improve the accuracy of reconciled forecasts, at the cost of increased variance \citep{nystrup2021dimensionality,panagiotelis2021forecast,pritularga2021stochastic,moller2023}. 
To the best of our knowledge, we are the first to quantify the uncertainty of the estimated forecast reconciliation weights and to use this to obtain improved estimates of the forecast-error covariance. \revTwo{In addition to improved quantification of forecast uncertainty, this provides a foundation for model reduction.}

\revTwo{Since the introduction of forecast reconciliation by \cite{athanasopoulos2009hierarchical} in the context of forecasting Australian tourism, and its more general treatment by \cite{hyndman2011optimal}, numerous studies have advanced its} theoretical understanding. \cite{wickramasuriya2019optimal} showed the relation between generalized least squares and the minimum trace (minT) solution and derived the variance of the reconciled forecast errors. \cite{panagiotelis2021forecast} presented a geometric interpretation of forecast reconciliation and proved that the minT solution is optimal in the sense that it minimises the expected loss. \cite{di2024forecast} extended the work by \cite{hollyman2021understanding} on the connection between forecast combination and reconciliation to include linear constraints.

\cite{wickramasuriya2019optimal} provided theoretical justification for using the variance--covariance matrix of the base forecast errors as an estimate of the unknown and unidentifiable variance--covariance matrix of the coherency errors. However, \revTwo{this matrix is generally high-dimensional and often ill-conditioned}. Therefore, shrinkage or simplification \citep{athanasopoulos2017forecasting} is usually applied. \cite{nystrup2021dimensionality} used eigenvalue decomposition, while \cite{moller2023} proposed a parameterised likelihood approach to reduce the dimension of the problem. \cite{eckert2021forecasting} used Bayesian methods to find the posterior distribution of the reconciled forecasts and bias shrinkage to shrink some weights towards zero. By now, the minT approach, including optimal shrinkage \citep{ledoit2003improved}, is established as the standard for forecast reconciliation.

\rev{\cite{wickramasuriya2021prop} showed a number of projection results related to forecast reconciliation, in particular in-sample distance reducing properties of the reconciled forecasts. \cite{Wickramasuriya2024} considered reconciliation in a Gaussian setting, which is also the assumption made throughout this article.}

This article introduces a number of new results and insights on the statistical properties of forecast reconciliation: 1) we show how the reconciliation weights can be estimated as parameters in a general linear model (LM) independently of the forecast variance; 2) the LM formulation implies an orthogonal projection that is shown to generalise to all levels of the hierarchy through the coherency constraints \rev{(similar to \cite{Wickramasuriya2024})}; 3) from the orthogonal projection and the resulting analysis-of-variance (ANOVA) type separation of variation, we derive a central estimator for the forecast error variance--covariance matrix; 4) we prove that distance-reducing properties apply to the reconciled forecasts at all levels of the hierarchy as well as to the forecast-error covariance; 5) we formulate forecast reconciliation with shrinkage using maximum a posteriori (MAP) estimation, which highlights the choice of priors; and 6) the LM formulation allows us to estimate the parameter (weight matrix) variance--covariance matrix \rev{(and standard errors)} and obtain a better estimate of the forecast-error variance--covariance.

\rev{The main contribution is the quantification of the uncertainty of the weight matrix in points 5 and 6, which relies heavily on the contributions in point 1--4.}
The introduced formulations \revTwo{and} the analysis of the statistical properties of forecast reconciliation pave the way for future work on statistical testing and modelling of the parameters in the reconciliation weight matrix.

The outline of this article is as follows. In Section \ref{sec:ForeRecon}, we formulate the general framework, show its equivalence with forecast reconciliation, and derive a number of in-sample results related to orthogonal projection. We show how shrinkage can be formulated as a MAP estimation problem and derive the mean-value parameters and forecast variances in Section \ref{sec:shrinkandmap}. Results on parameter and forecast variance are presented in Section \ref{sec:Par_Fore_Var} along with some test statistics. \rev{Section \ref{sec:scores} briefly outlines the probabilistic scores used in the simulation and case study.} Section \ref{sec:simStu} analyses improvements in forecast distributions in a simulation study. Section \ref{sec:case} illustrates the variance separation and standard errors for the high-dimensional weight matrix in a case study on electricity load forecasting in Sweden. \rev{In Section \ref{sec:persp}, we outline some future research directions and discuss the results.} Finally, Section \ref{sec:con} concludes.

A full list of symbols used in this article is given in Appendix \ref{sec:momen}. As many results in this work involve Kronecker products and vectorisation, we have listed the most important relations in Appendix \ref{sec:useful}. Multivariate distributions and their usage in objective functions are briefly presented in Appendix \ref{sec:logLike}.

\section{Forecast reconciliation as a linear model}\label{sec:ForeRecon}

\cite{hyndman2011optimal} proposed \rev{forecast reconciliation as a minimisation problem. The solution to that problem is equivalent to the solution of the general linear model formulation with the base forecasts as response variable}
\begin{align}
  \vect{\hat{y}}_t=\vect{S}\vect{\tilde{y}}_t+\vect{\epsilon}^h_t;\quad \vect{\epsilon}^h_t\sim N(\vect{0},\vect{\Sigma}_{h})\label{eq:initMod},
\end{align}
where \update{$\vect{\hat{y}}_t\in\mathbb{R}^{n}$ is the base forecast on all levels at time $t$,} 
$\vect{S}\in\mathbb{R}^{n\times{m}}$ is a summation matrix ensuring coherency, $\vect{\tilde{y}}_t\in\mathbb{R}^m$ is the \rev{bottom-level} reconciled  forecast \rev{(note the difference in the dimension of $\vect{\hat{y}}_t$ and $\vect{\tilde{y}}_t$)}, $\vect{\epsilon}^h_t\in\mathbb{R}^n$ is the reconciliation error, $\vect{\Sigma}_{h}\in\mathbb{R}^{n\times n}$ is the variance--covariance of the reconciliation error, $n$ is the dimension of the hierarchy, and $m$ is the dimension of the bottom level of the hierarchy. \rev{The specific solution to the problem depends on the choice/estimate of $\vect{\Sigma}_h$ but is, of course, equivalent to both GLS (generalised least squares) and ML (maximum likelihood) estimates.}

As a simple example, consider forecasts of quarterly, biannual, and annual observations $\vect{\hat{y}}_t=[\hat{y}_{\text{A},t},\hat{y}_{\text{B}_1,t},\hat{y}_{\text{B}_2,t},\hat{y}_{\text{Q}_1,t},...,,\hat{y}_{\text{Q}_4,t}]^\top$. In this case, $m=4$, $n=7$, and the summation matrix is
\begin{align}
\vect{S}=\left[\begin{matrix}
    1 & 1 & 1& 1\\
    1 & 1 & 0& 0\\
    0 & 0 & 1 & 1\\
    & & \vect{I}_4
\end{matrix}\right]=  \left[\begin{matrix}
\ST\\ \vect{I}_4
\end{matrix}\right].\label{eq:SumExample}
\end{align}
We refer to $\ST$ as the top-level summation matrix.

\update{If $\vect{\Sigma}_{h}$ were known,} the formulation \eqref{eq:initMod} would lead to estimated reconciled forecasts
\begin{align}
    \vect{\tilde{y}}_t=
    (\vect{S}^\top \Sigh^{-1}\vect{S})^{-1}\vect{S}^\top\Sigh^{-1}
    \vect{\hat{y}}_t
    = \vect{P}\vect{\hat{y}}_t\label{eq:hier}.
\end{align}
\update{However, the variance--covariance matrix $\vect{\Sigma}_{h}$ is unknown in real applications and it is unidentifiable \citep{wickramasuriya2019optimal}. Usually, some version of the \revTwo{estimated} base-forecast-error variance--covariance matrix is used as proxy.}

\revTwo{\cite{wickramasuriya2021prop} showed that choosing $\Sigh$ as the base-forecast-error variance--covariance matrix makes the solution \eqref{eq:hier} equivalent to minimising the $l_2$ norm of the reconciled error under coherency constraints.}

As pointed out by \cite{panagiotelis2021forecast}, the matrix $\vect{S}\vect{P}$ is a projection matrix from the $n$-dimensional base forecast ($\vect{\hat{y}}_t\in\mathbb{R}^n$) into a coherent $m$-dimensional linear subspace. In the following, we will view the weight matrix, $\vect{P}$, as parameters in a general \revTwo{linear model (}LM). It follows directly from the construction of $\vect{P}$ that \revTwo{a} set of linear \revTwo{constraints is} built into the estimation  \citep{wickramasuriya2019optimal} 
\begin{align}
  \vect{P}\vect{S}=\left[\begin{matrix} \PT & \PB     
      \end{matrix}\right]\left[\begin{matrix} \ST \\ \vect{I}     
      \end{matrix}\right]=\vect{I}_m,\label{eq:linCon}
\end{align}
where $\PT\in\mathbb{R}^{m\times(n-m)}$ and $\PB\in\mathbb{R}^{m\times m}$ are the weight matrices for the top and bottom level, respectively. Equation \eqref{eq:linCon} gives $m^2$ linear constraints, and with $\vect{P}\in \mathbb{R}^{m\times n}$ we have a total of $nm-m^2=m(n-m)$ free parameters to \revTwo{estimate}. Below, we will show how this can be formulated as a linear regression problem with linear constraints, thereby allowing estimation of the parameters and reconciled variance at the bottom level using maximum likelihood (ML) or restricted maximum likelihood (REML) estimation.

\subsection{Maximum likelihood and restricted maximum likelihood}\label{sec:MLE}
\rev{To make the development in the following clear, we start with some notation. The bottom-level observations and base forecasts at time $t$ are denoted $\vect{y}_t\in\mathbb{R}^m$ and $\yhatBt\in\mathbb{R}^m$. The collection of all bottom-level observations ($\vect{y}\in\mathbb{R}^{mT}$) and bottom-level base forecasts ($\yhatB\in\mathbb{R}^{mT}$) is
\begin{align}
\vect{y}=\left[\begin{matrix}\vect{y}_T\\ \vect{y}_{T-1}\\ \vdots \\ \vect{y}_1\end{matrix}\right];\quad
\yhatB=\left[\begin{matrix}\vect{\hat{y}}_{\text{bot},T}\\ \vect{\hat{y}}_{\text{bot},T-1}\\ \vdots \\ \vect{\hat{y}}_{\text{bot},1}\end{matrix}\right],
\end{align}
\revTwo{where $T$ is the number of bottom-level observations. }}

\rev{Further, the base forecasts are divided into bottom-level base forecasts $\yhatBt\in\mathbb{R}^m$ and $\yhatTt\in\mathbb{R}^{n-m}$, i.e., the "top" level is everything that is not the bottom level. The collection of all base forecasts at time $t$ is 
\begin{align}
    \vect{\hat{y}}_t=\left[\begin{matrix}\yhatTt\\ \yhatBt\end{matrix}\right].
\end{align}
}

The main contribution of this article is to show that the forecast reconciliation problem can be written as a \update{LM of the bottom-level base-forecast errors}
\begin{align}
  \vect{y}-\yhatB=\vect{X}\vect{\beta}+\vect{\epsilon};\quad \vect{\epsilon}\sim N(\vect{0},\vect{\Sigma}),\label{eq:GLMformFull}
\end{align}
where $\vect{y}\in\mathbb{R}^{mT}$ are the \update{collection of all} bottom-level observations, $\yhatB\in\mathbb{R}^{mT}$ is the \update{collection of all} bottom-level base forecasts, \update{$T$ is the number of observations in the training set, $\vect{\epsilon}$ is the residuals from the LM, and the variance--covariance matrix} $\vect{\Sigma}\in\mathbb{R}^{mT\times mT}$ and \update{design matrix} $\vect{X}\in\mathbb{R}^{mT\times p}$ are chosen in appropriate ways. \rev{In general, we leave $p$ to be defined, as many of the results apply more generally than the specific result stated \rev{in \eqref{eq:ConstrReg1}} below. For appropriate choices of $\vect{X}$, the parameters $\vect{\beta}$ relate directly to the weight matrix $\vect{P}$.}

\update{Specifically, we show in Theorem \ref{the:opt.proj.hier} (in Section \ref{sec:resLM}) that the reconciliation problem can be formulated as the LM
\begin{align}
  \vect{y}_t-\yhatBt=\left[\vect{I}_m\otimes (\yhatTt^\top-\yhatBt^\top\ST^\top)\right]\betaT+\vect{\epsilon}_t;\quad \vect{\epsilon}_t\sim N(\vect{0},\Sigr),\label{eq:ConstrReg1}
\end{align}
for a specific time point $t$, with $\vect{y}_t\in\mathbb{R}^m$, $\yhatBt\in\mathbb{R}^m$, $\yhatTt\in\mathbb{R}^{n-m}$,  $\betaT\in\mathbb{R}^{m(n-m)}$, $\vect{\epsilon}_t\in\mathbb{R}^m$, and $\vect{\Sigma}_{r}\in\mathbb{R}^{m\times m}$\revTwo{, which is assumed to be positive definite}. 
\update{Section \ref{sec:resLM} provides some general results needed to prove the result given in  \eqref{eq:ConstrReg1}.}
In Section \ref{sec:equiva}, we show that when $\Sigh$ is chosen as the \revTwo{estimated} base-forecast-error variance--covariance, the ML estimate for $\betaT$ is equivalent to the weight matrix $\vect{P}$ in \eqref{eq:hier}.}

\rev{Formulation \eqref{eq:ConstrReg1} is simpler than formulation \eqref{eq:GLMformFull} as it assumes independence between $\vect{\epsilon}_t$ and $\vect{\epsilon}_s$ for $t\neq s$. We will discuss how this can be relaxed in Section \ref{sec:persp}. We note that the response variable is chosen as the bottom-level base forecast errors. The base forecast errors at other (appropriately chosen $m$-dimensional) levels could have been chosen and would lead to equivalent results. This, more general, formulation would increase the complexity of the notation and is omitted here.}

\update{Note that \eqref{eq:initMod} is a model for the coherency error, while \eqref{eq:ConstrReg1} focuses on the bottom-level forecast errors and, as such, the observed values are included on the left-hand side of \eqref{eq:GLMformFull}. Unlike the variance--covariance of the reconciliation error $\Sigh$, \revTwo{an estimate of the} variance--covariance of the regression error $\Sigr$ can be \revTwo{obtained} from the observed data (provided that sufficient data is available). The formulation as a LM allows us to prove a number of in-sample distance-reducing properties in Section \ref{sec:distReduc}}.
Later, in Section \ref{sec:shrinkandmap}, we will show how shrinkage can be introduced by imposing specific priors on $\vect{\beta}$ and $\vect{\Sigma}_r$ using MAP estimation.

\subsection{Results for a class of general linear models}\label{sec:resLM}
\update{First, we consider a class of LMs that is more general than the specific models needed to show the results related to forecast reconciliation.}
 We begin by considering one particular time of prediction and formulate the following linear regression model 
\begin{align}
  \vect{y}_{t}=\vect{X}_{\cdot,t}\vect{\beta}+\vect{\epsilon}_t;\quad \vect{\epsilon}_t\sim N(\vect{0},\Sigr) \textrm{ and iid.}, \label{eq:glmForm}
\end{align}
where \update{$\vect{X}_{\cdot,t}\in\mathbb{R}^{m\times p}$ is the rows of the design matrix $\vect{X}\in\mathbb{R}^{mT\times p}$ corresponding to the observations $\vect{y}_{t}$, and $\vect{\epsilon}_t\in\mathbb{R}^m$ is a sequence of iid.~normal random vectors with mean zero and variance--covariance matrix $\Sigr$}. \rev{Note that $\Sigr$ is a completely different object than $\vect{\Sigma}_h$. The former is the variance--covariance of the reconciled bottom-level forecast errors, while the latter is the variance--covariance of the base forecasts for all levels. The aim in this section is to estimate $\vect{\beta}$ and $\Sigr$ from the LM formulation and show the equivalence between \eqref{eq:initMod} and \eqref{eq:glmForm} for specific structures of the design matrix.}

Reconciled forecasts for other levels are constructed by aggregating the forecasts for the bottom level ($\vect{\tilde{y}}_t=\vect{X}_{\cdot,t}\vect{\hat{\beta}}$), i.e., 
\begin{align}
\ytildeFt=\vect{S}\vect{\tilde{y}}_{t}=\vect{S}\vect{X}_{\cdot,t}\vect{\hat{\beta}},\label{eq:y.tildeF}
\end{align}
where $\ytildeFt\in\mathbb{R}^n$ is the collection of reconciled forecasts on all top and bottom levels. This implies that the forecasts are coherent for all choices of $\vect{X}_{\cdot,t}$.

\update{With $\vect{y}$ as defined in \eqref{eq:GLMformFull} and}
\begin{align}
  \vect{X}=\left[\begin{matrix}
      \vect{X}_{\cdot,T}\\
      \vect{X}_{\cdot,T-1}\\
      \vdots\\
      \vect{X}_{\cdot,1}
      \end{matrix}\right],\label{eq:fullDesign}
\end{align}
the full model for all bottom-level observations, assuming independence between \update{the $T$} ($m$-dimensional) observations, is
\begin{align}
    \vect{y}=\vect{X}\vect{\beta}+\vect{\epsilon};\quad \vect{\epsilon}\sim N(0,\vect{I}_T\otimes\Sigr).\label{eq:glmFormFull}
\end{align}

We consider design matrices of the following special form
\begin{align}
  \vect{X}_{\cdot,t}=&
  \left[\begin{matrix}
      \vect{x}_{1,t}^\top & \vect{0}&\hdots&\vect{0}\\ 
       \vect{0}& \vect{x}_{2,t}^\top & \ddots & \vdots\\
       \vdots &\ddots & \ddots & \vect{0}\\
       \vect{0}&\hdots&\vect{0}& \vect{x}_{m,t}^\top &
    \end{matrix}\right]\label{eq:desingGen},
\end{align}
with $\vect{x}_{i,t}\in\mathbb{R}^{p_i}$, meaning that, in general, we allow $\vect{x}_{i,t}\neq \vect{x}_{j,t}$. This will allow us to remove columns \revTwo{from} the design matrix as part of a test strategy. \update{We will sometimes need the design matrix $\vect{X}_{i,\cdot}$ corresponding to a specific bottom-level observation
\begin{align}
  \vect{X}_{i,\cdot}=&
  \left[\begin{matrix}
      \vect{x}_{i,T}^\top \\ 
       \vect{x}_{i,T-1}^\top \\
       \vdots \\
        \vect{x}_{i,1}^\top &
    \end{matrix}\right]\label{eq:desingGen2}.
\end{align}
}

\rev{Assuming full rank of $\vect{X}$, }the general solution to the model defined by \eqref{eq:fullDesign} is well known, and the ML or generalised least squares (GLS) estimate for the mean-value parameter estimates $\vect{\hat{\beta}}$ in model \eqref{eq:glmFormFull} is 
  \begin{align}
  \vect{\hat{\beta}}=& (\vect{X}^\top(\vect{I}_T\otimes\vect{\hat{\Sigma}}^{-1}_r)\vect{X})^{-1} \vect{X}^\top(\vect{I}_T\otimes\SigrHat^{-1})\vect{y}.
  \label{eq:normalEq}
  \end{align}
  Further, with $\vect{e}_t= \vect{y}_t-\vect{X}_{\cdot,t}\vect{\hat{\beta}}$, the ML estimate of $\Sigr$, for a given $\vect{\hat{\beta}}$, is
  \begin{align}
  \SigrHatML=\frac{1}{T}\sum_{t=1}^T\vect{e}_t \vect{e}_t^\top.\label{eq:Sigma.rML}
  \end{align}
  Hence, in general, the mean-value parameters $\vect{\hat{\beta}}$ depend on the estimate of $\Sigr$, which depends on the estimate of $\vect{\beta}$\rev{. Note that $\vect{\Sigma}_r$ and $\vect{\Sigma}_h$ play very different roles. In \eqref{eq:hier}, $\vect{\Sigma}_h$ ensures that the weights are optimal, while $\vect{\Sigma}_r$ is simply a residual variation that, in the main part of this article, does not affect the weights. The ML estimator for the variance--covariance matrix is usually biased and, hence, the REML estimator is usually preferred when estimating covariance structures}. 
  In situations similar to those usually considered in forecast reconciliation, the solution simplifies considerably with the precise result given in the theorem below\rev{, which may also be found in text books on multivariate regression model \citep[see, e.g.,][]{christensen_1991}}.

\begin{theorem}[LM/GLS solution]\label{the:main}
In the special case where, for all $(i,j)$, $\vect{x}_{i,t}=\vect{x}_{j,t}\in\mathbb{R}^{\bar{p}}$, \revTwo{with $\bar{p}=\frac{1}{m}\sum_{i=1}^mp_i$}\rev{, and $\vect{X}$ has full rank}, the parameter estimates
  \begin{align}
      \vect{\hat{\beta}}= (\vect{X}^\top\vect{X})^{-1} \vect{X}^\top\vect{y},\label{eq:beta.hatSimp}
    \end{align}
  are independent of $\Sigr$ and the REML estimate of $\Sigr$ is
    \begin{align}
      \SigrHatREML=\frac{1}{T-\bar{p}}\sum_{t=1}^T\vect{e}_t \vect{e}_t^\top\label{eq:Sig.rSimp},
      \end{align}  
      where $\vect{e}_t= \vect{y}_t-\vect{X}_{\cdot,t}\vect{\hat{\beta}}$.
\end{theorem}
\begin{proof}
See Appendix \ref{proof:main} \rev{for a proof tailored to our application or \cite{christensen_1991} for a more general approach}.    
\end{proof}

\update{The main takeaway from Theorem \ref{the:main} is that the parameter estimates \eqref{eq:beta.hatSimp}, which will be the forecast reconciliation weights, are independent of the choice of $\Sigr$.}

\update{Note that in the case where $\vect{x}_{i,t}\neq\vect{x}_{j,t}$ for some $(i,j)$, then the independence between the mean-value and variance parameters does not hold, and the REML estimate is given as the solution to 
\begin{align}
      \SigrHatREML =
      \frac{1}{T}\left(
      \sum_{t=1}^T\vect{e_t} \vect{e}_t^\top+
      \frac{
        \partial\log|\vect{X}^\top\left(\vect{I}_T\otimes\vect{\Sigma}^{-1}_{r}\right)\vect{X}|
      }{
        \partial\vect{\Sigma}^{-1}_{r}} \Bigg|_{\vect{\Sigma}^{-1}_{r}= \vect{\hat{\Sigma}}^{-1}_{r,\text{REML}}}\right).\label{eq:Sigma.rREML}
\end{align}
The details for arriving at \eqref{eq:Sigma.rREML} are given in Appendix \ref{proof:main}. In the general case, where the parameters $\vect{\hat{\beta}}$ depend on the estimate of $\Sigr$, the relaxation algorithm \citep[see e.g.][]{Madsen07} can be used to iterate between \eqref{eq:normalEq} and \eqref{eq:Sigma.rML} or \eqref{eq:Sigma.rREML} until convergence.
}

In the considered setting, both $\vect{X}$ and $\vect{\Sigma}$ are \revTwo{sparse} high-dimensional matrices, although their products can be represented as lower-dimensional dense matrices (see Appendix \ref{proof:main}). 
In the special case of $\vect{x}_{i,t}=\vect{x}_{j,t}$, the solution is given directly by \eqref{eq:beta.hatSimp} and \eqref{eq:Sig.rSimp} (or the ML version of the latter). 

It follows directly from (the proof of) Theorem \ref{the:main} that in the special case where $\forall$ $(i,j)$ $\vect{x}_{i,t}=\vect{x}_{j,t}\in\mathbb{R}^{\bar{p}}$\update{, implying that $\vect{X}_{i,\cdot}=\vect{X}_{j,\cdot}=\vect{X}_{1,\cdot}$ $\forall$ $(i,j)$,} then the estimate of $\vect{\beta}$ \rev{is given by the, computationally much more efficient, formulation}
\begin{align}
  \vect{\hat{\beta}}^m= (\vect{X}_{1,\cdot}^\top\vect{X}_{1,\cdot})^{-1}\vect{X}_{1,\cdot}^\top\vect{Y},\label{eq:betaHat.i}
\end{align}
  where $\vect{Y}=\vectoriz^{-1}(\vect{y})\in\mathbb{R}^{T\times m}$ and $ \vect{\hat{\beta}}^m=\left[\vect{\hat{\beta}}_1\quad ...\quad \vect{\hat{\beta}}_m\right]^\top\in\mathbb{R}^{m\times \bar{p}}$ does not depend on $\Sigr$. 
\update{The result in \eqref{eq:beta.hatSimp} is given in terms products of sparse matrices (multiplied by a vector), while \eqref{eq:betaHat.i} is given in terms of products of lower-dimensional dense matrices multiplied by a matrix.}

\subsection{Equivalence to forecast reconciliation}\label{sec:equiva}
The weight matrix $\vect{P}$ depends on the choice/estimate of proxy for the variance--covariance matrix $\Sigh$. Consequently, the equivalence between the regression model \eqref{eq:GLMformFull} and forecast reconciliation \eqref{eq:initMod} and \eqref{eq:hier} depends on this choice. The next theorem states the equivalence if $\Sigh$ were chosen as the \revTwo{estimated} base-forecast-error variance--covariance matrix \rev{(equivalent to the MinT reconciliation)}. \revTwo{The main contribution of the theorem is the explicit formulation of the hierarchical reconciliation problem as an LM, thereby enabling direct application of established results, including the REML estimator. We note that results concerning the ML estimator have previously been presented by \cite{wickramasuriya2019optimal,Wickramasuriya2024,Ando_Narita2024}, albeit without making the connection to the LM formulation explicit. }
 
\begin{theorem}[Equivalence to forecast reconciliation]\label{the:opt.proj.hier}
The model \eqref{eq:GLMformFull}, with $\vect{X}_{\cdot,t}=\vect{I}\otimes\vect{\hat{y}}_t^\top$\rev{, $\vect{\epsilon}_t$ and $\vect{\epsilon}_s$ independent for $s\neq t$}, and linear coherency constraints
$\vectoriz^{-1}\left(\vect{\beta}\right)^\top\vect{S}=\vect{I}_m$, is equivalent to forecast reconciliation. Formally
\begin{align}
 \vect{P}=\vectoriz^{-1}(\vect{\hat{\beta}})^\top,\label{eq:18a}
\end{align}
where $\vect{P}$ is the usual forecast reconciliation weight matrix \eqref{eq:hier} using 
\begin{align}
\SighHat=\frac{1}{T}\left(\vect{Y}\vect{S}^\top-\vect{\hat{Y}}\right)^\top\left(\vect{Y}\vect{S}^\top-\vect{\hat{Y}}\right)=\update{
\frac{1}{T}\sum_{i=1}^T(\vect{S}\vect{y}_t-\vect{\hat{y}}_t)(\vect{S}\vect{y}_t-\vect{\hat{y}}_t)^\top}\label{eq:SighHat},
\end{align}
\update{where $\vect{Y}=\vectoriz^{-1}(\vect{y})$, $\vect{\hat{Y}}=[\YhatT\quad\YhatB]=[\vectoriz^{-1}(\yhatT)\quad\vectoriz^{-1}(\yhatB)]$}
\rev{, the resulting design matrix is assumed to have full rank, and  $\vect{\Sigma_r}$ is assumed positive definite}. The model, including coherency constraints, can be formulated as
\begin{align}
  \vect{y}_t-\yhatBt=\left[\vect{I}_m\otimes (\yhatTt^\top-\yhatBt^\top\ST^\top)\right]\betaT+\vect{\epsilon}_t;\quad \vect{\epsilon}_t\sim N(\vect{0},\Sigr),\label{eq:ConstrReg}
\end{align}
with $\vect{S}=\left[\ST^\top\quad \vect{I}\right]^\top$.

\update{Further, with 
\begin{align}
     \vect{e}_t=\vect{y}_t-\yhatBt-\left[\vect{I}_m\otimes (\yhatTt^\top-\yhatBt^\top\ST^\top)\right]\betaHatT=\vect{y}_t-\vect{\tilde{y}}_t,
     \end{align}
}
the ML and REML estimates of $\Sigr$ are
\begin{align}
            \SigrHatML=& \frac{1}{T}\sum_{i=1}^{T}\vect{e}_t\vect{e}_t^\top;\quad 
                  \SigrHatREML= \frac{1}{T-(n-m)}\sum_{i=1}^{T}\vect{e}_t\vect{e}_t^\top,\label{eq:18}
\end{align}
\update{and the estimate $\betaHatT$ is independent of $\Sigr$.}
\end{theorem}
\begin{proof}
    \revTwo{The proof relies heavily on the results presented in \cite{wickramasuriya2019optimal}. See also} Appendix \ref{proof:opt.proj.hier}.
\end{proof}
\rev{
  \begin{remark2}[Assumptions]
    LMs assume that the design matrix is known and not a function of the observations. In the result above, the design matrix depends on forecasts that are usually a function of the observations. \revTwo{In this article, we use} the assumption that the base forecast models are trained on another period than the weight matrix $\vect{\beta}$.
  \end{remark2}
}

\rev{The independence assumption implies that $Cov[\epsilon_{t,i},\epsilon_{s,j}]=0$, for $i\neq j$ and all $(i,j)$; however, there are no assumptions about $Cov[\epsilon_{t,i},\epsilon_{t,j}]$. We \revTwo{discuss} ways of including correlation between observations in the discussion (Section \ref{sec:persp}). The full rank assumption about the design matrix is often fulfilled for reasonable choices (and number of observations) of the design matrix (i.e., base forecast models), but can be violated if some base forecasts are constructed as linear combinations of others (either directly or indirectly though the construction of base forecast models). The positive definite assumption about $\vect{\Sigma}_r$ is fulfilled for reasonable systems (again it should not be possible to write some observed bottom-level base forecast errors as linear combinations of other observed base-forecast errors). Coherency is ensured  by the linear constraints $\vectoriz^{-1}\left(\vect{\beta}\right)^\top\vect{S}=\vect{I}_m$. Even though only observations at the bottom level is used, forecasts from all levels are included in the estimation through the design matrix, as defined in \eqref{eq:ConstrReg}. } 

\rev{In general, $\Sigh$ is not known and unidentifiable \citep{wickramasuriya2019optimal}. Instead,} different proxies for $\Sigh$ can be used. In this section we will use the \revTwo{estimated} variance--covariance of the base-forecast errors \eqref{eq:SighHat}.
Equation \eqref{eq:ConstrReg} is formulated in terms of the LM given in \eqref{eq:GLMformFull} or \eqref{eq:glmForm} with $\vect{X}_{\cdot,t}=\vect{I}_m\otimes (\yhatTt^\top-\yhatBt^\top\ST^\top)$, where the full design matrix $\vect{X}$ is constructed by \eqref{eq:fullDesign}. The model \eqref{eq:ConstrReg} defines the reconciliation problem as a linear regression problem and, consequently, as a linear projection of the bottom-level base-forecast errors. The dimension of the model, which is the dimension of $\betaT$, is $m\cdot (n-m)$. 

The residual error $\vect{\epsilon}_t$ is (in the case of a perfectly specified model) referred to as the irreducible error by \cite{pritularga2021stochastic}, while variance components related to estimation of the base-forecast variance--covariance matrix appear as the variance of the estimator $\betaHatT$. A bias correction can be included in \eqref{eq:ConstrReg} as $\vect{X}_{\cdot,t}=\vect{I}_m\otimes \left[1,\yhatTt^\top-\yhatBt^\top\ST^\top\right]$ at the expense of $m$ extra parameters. 

\cite{wickramasuriya2019optimal} showed that if $\Sigh$ were known, then the variance of the (bottom-level) reconciled forecasts would be given by 
\begin{align}
  V[\vect{y}_t-\vect{\tilde{y}}_t]=\vect{P} \Sigh \vect{P}^\top.\label{eq:19}
\end{align}
\revTwo{Given the equivalence established in Theorem \ref{the:opt.proj.hier} and the references around it, it is not surprising that the variance estimate \eqref{eq:19} and the ML estimate in \eqref{eq:18} are equal, but we formulate it in the following corollary. 
}
\begin{remark}[Reconciled variance] \label{remark:reconvar}
When $\Sigh$ is chosen as $\SighHat=\frac{1}{T}(\vect{Y}\vect{S}^\top-\vect{\hat{Y}})^\top(\vect{Y}\vect{S}^\top-\vect{\hat{Y}})$, the ML estimate of the residual variance in model \eqref{eq:ConstrReg} is equal to the variance of the reconciled forecast \revTwo{errors} \eqref{eq:19}: 
\begin{align}
  \vect{P} \SighHat \vect{P}^\top =\SigrHatML.
\end{align}
\end{remark}
\begin{proof}
\revTwo{Follows from the equivalence established in Theorem \ref{the:opt.proj.hier}.} See Appendix \ref{proof:reconvar} \revTwo{for explicit calculations}.    
\end{proof}

\update{Theorem \ref{the:opt.proj.hier} and Corollary \ref{remark:reconvar} show that the forecast reconciliation problem using the full forecast-error variance--covariance matrix as proxy for $\Sigh$ can be formulated as a LM for the bottom-level base-forecast errors. One of the key benefits is that it provides access to the well-known statistical methods associated with the general linear model, including the in-sample distance-reducing properties discussed in the next section.}

\subsection{Distance-reducing properties}\label{sec:distReduc}
\update{A key component in LMs is the separation of variation and the resulting (ANOVA-type) tests for model reduction that are based in orthogonal projections. In this section we show how reconciliation can be formulated as orthogonal projections and prove distance-reducing properties. In Theorem \ref{col:proj}, we show that
 \begin{align}
||\vect{y}_I-\vect{\tilde{y}}_I||^2\leq||\vect{y}_I-\vect{\hat{y}}_I||^2,\label{eq:DistReduc1a}
\end{align}
where $I$ corresponds to the collection of all observations/forecasts at a particular level.
}

When $\vect{\Sigma}=\sigma^2\vect{I}$, the LM is often formulated in terms of orthogonal projections \citep[see e.g.][]{madsen_thyregod_2011}. For more general variance--covariance structures, the LM defines orthogonal projections from a transformed space. \update{
The precise definitions are \update{well known from statistical theory, but} given in the remark below \update{for completeness}.
\begin{remark2}[LM projection]\label{remark:glmproj}
With $\vect{z}=\vect{y}-\yhatB\in\mathcal{Y}\subseteq\mathbb{R}^{mT}$, the estimate ($\vect{\hat{\beta}}$) based on the model \eqref{eq:GLMformFull} defines an orthogonal projection from the transformed space $\mathcal{Y}^{*}$, with $\vect{z}^{*}\in\mathcal{Y}^{*}\subseteq\mathbb{R}^{mT}$ and $\vect{z}^{*}=\vect{\Sigma}^{-1/2}\vect{z}$, $\vect{\Sigma}=\vect{\Sigma}^{1/2}\left(\vect{\Sigma}^{1/2}\right)^\top$, into a $p$-dimensional linear subspace.  
If $\vect{\Sigma}=\sigma^2\vect{I}$, then the estimate $\vect{\hat{\beta}}$ defines an orthogonal projection from $\mathcal{Y}$ into a $p$-dimensional linear subspace.
\end{remark2}
}

Using Remark \ref{remark:glmproj}, the model \eqref{eq:glmFormFull}--\eqref{eq:desingGen} defines an orthogonal projection from $\mathbb{R}^{Tm}$ into a $\bar{p}m$ dimensional linear subspace. 
The next lemma states the LM as a projection of the base-forecast errors and shows that, due to the coherency constraints, the projection matrix is the same for base-forecast errors at all levels of the hierarchy.

\begin{lemma}[Projection and separation of variation]\label{the:proj}
Assume \update{that $\vect{X}_{\cdot,t}$ is chosen as in Theorem \ref{the:opt.proj.hier} and that} $\texttt{rank}(\vect{X}_{1,\cdot})=n-m$. If $\vect{H}=\vect{X}_{1,\cdot}(\vect{X}_{1,\cdot}^\top\vect{X}_{1,\cdot})^{-1}\vect{X}_{1,\cdot}^\top$, then $\vect{H}$ defines an orthogonal projection 
 for any selection of columns $I\subseteq \{1,2,...,n\}$ of $\vect{\hat{Y}}=\vectoriz^{-1}(\vect{\hat{y}})\in\mathbb{R}^{T\times n}$ and corresponding selection of rows of $\vect{S}$. In particular, using the notation $\vect{\hat{Y}}_I:=\vect{\hat{Y}}_{\cdot,I}$ and $\vect{S}_I:=\vect{S}_{I,\cdot}$, then $\vect{X}_{1,\cdot}$ defines an orthogonal projection and
 \begin{subequations}
  \begin{align}
  \begin{split}
        \left(\vect{Y}\vect{S}^\top_{I}-\vect{\hat{Y}}_I\right)^\top\left(\vect{Y}\vect{S}^\top_I-\vect{\hat{Y}}_I\right)     =&
 \vect{S}_I\left(\vect{Y}-\vect{\tilde{Y}}\right)^\top\left(\vect{Y}-\vect{\tilde{Y}}\right)\vect{S}_I^\top+ \\ & \left(\vect{\tilde{Y}}\vect{S}^\top-\vect{\hat{Y}}_I\right)^\top\left(\vect{\tilde{Y}}\vect{S}^\top-\vect{\hat{Y}}_I\right)\label{eq:projBa}
 \end{split}\\
 \begin{split}
 =& \left(\vect{Y}\vect{S}^\top_I-\vect{\hat{Y}}_I\right)^\top\left(\vect{I}-\vect{H}\right)
    \left(\vect{Y}\vect{S}^\top_I-\vect{\hat{Y}}_I\right)+  \\
    &\left(\vect{Y}\vect{S}^\top_I-\vect{\hat{Y}}_I\right)^\top\vect{H}\left(\vect{Y}\vect{S}^\top_I-\vect{\hat{Y}}_I\right),\label{eq:projB}
    \end{split}
    \end{align}
 \end{subequations}
 \update{where $\vect{\tilde{Y}}=\vectoriz^{-1}(\vect{\tilde{y}})\in \mathbb{R}^{T\times m}$ is the collection of all bottom-level reconciled forecasts.}
\end{lemma}

\begin{proof}
See Appendix \ref{proof:proj}.
\end{proof}

The sum-of-squares decomposition \eqref{eq:projBa} can be written in terms of the variance--covariance matrices discussed so far as
\begin{align}
     \SighHat= \vect{S}\SigrHatML\vect{S}^\top+ \frac{1}{T}\left(\vect{\tilde{Y}}\vect{S}^\top-\vect{\hat{Y}}\right)^\top\left(\vect{\tilde{Y}}\vect{S}^\top-\vect{\hat{Y}}\right),
\end{align}
where the first term is the residual sum of squares (of the reconciled forecast errors) and the second term is the squared difference between the base and reconciled forecast. Meanwhile, \eqref{eq:projB} stresses separation of variation by a series of orthogonal projections of the base-forecast error similar to Cochran's theorem \citep[see, e.g.,][]{madsen_thyregod_2011}. 

The next theorem \rev{was first reported in  \cite{wickramasuriya2021prop} (Theorem 1), and} shows that as a consequence of the coherent projection, distance-reducing properties hold across any arbitrary splitting of the column space of the base-forecast errors. 

\begin{theorem}[Distance-reducing properties]\label{col:proj}
A consequence of Lemma \ref{the:proj} is that for any column (and row) selection, $I\subseteq\{1,2,...,n\}$, of $\vect{\hat{Y}}$ (and $\vect{S}$) then 
  \begin{align}
      \begin{split}
 \vect{S}_{I}\left(\vect{Y}-\vect{\tilde{Y}}\right)^\top\left(\vect{Y}-\vect{\tilde{Y}}\right)\vect{S}_{I}^\top
 \leq     \left(\vect{Y}\vect{S}_{I}^\top-\vect{\hat{Y}}_{I}\right)^\top\left(\vect{Y}\vect{S}_{I}^\top-\vect{\hat{Y}}_{I}\right)\label{eqVarSep},
 \end{split}
  \end{align}
  where "$\leq$'' should be understood as "$rhs - lhs$" being positive semi-definite. Furthermore, with $\vect{Y}\vect{S}_{I}^\top-\vect{\hat{Y}}_{I}\in \mathbb{R}^{T\times q_I}$, $\vect{y}_I=\vectoriz(\vect{S}_{I}\vect{Y}^\top)$, $\vect{\hat{y}}_I=\vectoriz(\vect{S}_{I}\vect{\hat{Y}}^\top)$,  $\vect{\tilde{y}}_I=\vectoriz(\vect{S}_{I}\vect{\tilde{Y}}^\top)$, and
  $\vect{y}_I-\vect{\hat{y}}_I\in \mathcal{E}_I\subseteq\mathbb{R}^{T q_I}$, model \eqref{eq:ConstrReg} defines an orthogonal projection from $\mathcal{E}_I$ into a $q_I(n-m)$-dimensional subspace where 
\begin{align}
||\vect{y}_I-\vect{\hat{y}}_I||^2=||\vect{y}_I-\vect{\tilde{y}}_I||^2+||\vect{\tilde{y}}_I-\vect{\hat{y}}_I||^2,\label{eq:proj2}
\end{align}
and, therefore,
  \begin{align}
||\vect{y}_I-\vect{\tilde{y}}_I||^2\leq||\vect{y}_I-\vect{\hat{y}}_I||^2.\label{eq:DistReduc1}
\end{align}
\end{theorem}
\begin{proof}
    See Appendix \ref{proof:colProj}.
\end{proof}

The result \eqref{eq:DistReduc1} is similar to the results presented in \citet[][Theorems 3.1 and 3.2]{panagiotelis2021forecast}, but formulated for all observations  in the training set rather than for an individual forecast. Notice that \eqref{eqVarSep} generalises the result to variance--covariance matrices and shows that the distance reduction applies to any columns selection of $\vect{Y}\vect{S}^{\top}-\vect{\hat{Y}}$. \update{Theorem \ref{col:proj} shows that the distance between the observations and the reconciled forecasts is less than the distance between the observations and the bottom-level base forecast, while Lemma \ref{the:proj} provides the separation of variation into orthogonal sums that is the basis for the central variance--covariance estimator and statistical tests.}

\subsection{Central covariance estimator}
In the LM, series of orthogonal projections are used to construct hypothesis tests based on Cochran's theorem, which leads to the F-test based on $\chi^2$-distributions and the central estimator for the variance. Since we are working with variance--covariance matrices, we need results related to the Whishart distribution (multivariate generalisation of the $\chi^2$-distribution).

\update{Using the general definition of the Wishart distribution (see Lemma \ref{lemma:Whishart}, in Appendix \ref{sec:useful}),} we can derive the distribution of the sum of squares of each term in \eqref{eq:projB}, \update{as given in the next corollary. In addition to defining the REML estimator as the central estimator for $\Sigr$, the corollary also points to a test for total homogeneity (see Section \ref{sec:statTest}).}

\begin{remark}[Central estimator for $\SigrHat$]\label{col:CentralSig}
    Assuming model \eqref{eq:ConstrReg} and defining $\vect{H}$ as in Lemma \ref{the:proj}, then 
    \begin{align}
    \left(\vect{Y}-\vect{\tilde{Y}}\right)^\top\left(\vect{Y}-\vect{\tilde{Y}}\right)\sim W_m(\Sigr,T-(n-m)),
    \end{align}
    implying that $\SigrHatREML$ \eqref{eq:18} is the central estimator for $\Sigr$. If, in addition $\betaT=\vect{0}$, then 
    \begin{align}
    \left(\vect{\tilde{Y}}\vect{S}^\top-\vect{\hat{Y}}\right)^\top\left(\vect{\tilde{Y}}\vect{S}^\top-\vect{\hat{Y}}\right)\sim W_m(\boldsymbol{S}\Sigr\boldsymbol{S}^{\top},n-m).
    \end{align}
\end{remark}
\begin{proof}
See  Appendix \ref{proof:CentralSig}
\end{proof}

\rev{\subsection{Orthogonal or oblique projection?}
  The main result in this section is to demonstrate the equivalence between the oblique projection
  \begin{align}
        \ytildeFt=
    \vect{S}(\vect{S}^\top \Sigh^{-1}\vect{S})^{-1}\vect{S}^\top\Sigh^{-1}
    \vect{\hat{y}}_t
    = \vect{S}\vect{P}\vect{\hat{y}}_t,
  \end{align}
  where $\vect{S}\vect{P}$ is an oblique projection matrix, which defines an oblique projection from the base forecast at time $t$ into the space of reconciled forecasts at time $t$. The projection
  \begin{align}
    \vect{\tilde{e}}=\vect{X}(\vect{X}^\top\vect{X})^{-1} \vect{X}^\top\vect{e}
  \end{align}
  from the space of all bottom-level base forecast errors $\vect{e}=\vect{y}-\vect{\hat{y}}$ into the space of reconciled forecast errors $\vect{\tilde{e}}=\vect{y}-\vect{\tilde{y}}$ is orthogonal as long as the assumptions in Theorem \ref{the:main} are fulfilled (i.e., $\vect{x}_{i,t}=\vect{x}_{j,t}$) . If this is not the case, the projection will be oblique. Similarly, if the independence assumption is dropped, then the projection will be oblique. The comparison highlights that the role of $\vect{\Sigma}_h$ is  moved to the design matrix $\vect{X}$ (or into the estimation of $\vect{\beta}$).
  We briefly discuss the independence assumption and how serial correlation could be imposed in Section \ref{sec:persp}.
}

\section{Shrinkage and MAP estimation}\label{sec:shrinkandmap}

In most realistic settings, there is a large number of parameters in the base-forecast variance--covariance matrix $\SighHat$. The results so far use the estimated variance--covariance matrix directly. This is a disadvantage as it is well known that shrinkage is needed in order to stabilise the reconciliation weights \citep[see, e.g.,][]{nystrup2020temporal}.  When applying shrinkage, the projections presented in Lemma \ref{the:proj} are no longer orthogonal. In this section, we show the equivalence between MAP estimation and forecast reconciliation when shrinkage is applied to the variance--covariance matrix $\Sigh$ in \eqref{eq:hier}. 

\update{In the following some result are formulated in terms of the marginal base-forecast variances divided into top and bottom levels, i.e.,
\begin{align}
    \SighHatd=\left[\begin{matrix}
        \SighHatdT & \vect{0}\\
        \vect{0} & \SighHatdB\\
    \end{matrix}\right],
\end{align}
where both $\SighHatdT\in\mathbb{R}^{(n-m)\times (n-m)}$ and $\SighHatdB\in\mathbb{R}^{m\times m}$ are diagonal matrices. 
}

The usual shrinkage estimator for the variance--covariance matrix is 
\begin{align}
  \begin{split}
    \SigsHat
    =& (1-\lambda)\SighHat+\lambda\SighHatd,\label{eq:shrinkage}
  \end{split}
\end{align}
with $\SighHat= \frac{1}{T} (\vect{Y}\vect{S}^\top-\vect{\hat{Y}})^\top (\vect{Y}\vect{S}^\top-\vect{\hat{Y}})$ and $\SighHatd=\textrm{diag}(\SighHat)$ \citep{wickramasuriya2019optimal}. The shrinkage parameter $\lambda$ \update{ensures that the matrix $\SigsHat$ is not ill-conditioned, but} introduces bias in the estimation. Thus, \eqref{eq:shrinkage} expresses a bias--variance tradeoff. The weight matrix calculated using $ \SigsHat$ will be denoted 
\begin{align}
\vect{P}(\lambda)=(\vect{S}^\top\SigsHat^{-1}\vect{S})^{-1}\vect{S}^\top\SigsHat^{-1}. \label{eq:P.lambda}
\end{align}

In this section, we use MAP estimation \rev{to show that the shrinkage estimator is equivalent to MAP estimation for specific choices of prior distributions}. We assume that the conditional density for the observations is $f(\vect{y}|\vect{\theta})$ and that the parameters $\vect{\theta}$ follow some prior distribution $g(\vect{\theta})$. The MAP estimate is obtained as
\begin{align}
  \vect{\hat{\theta}}_{\text{MAP}}= \argmax_{\vect{\theta}}f(\vect{y}|\vect{\theta}) g(\vect{\theta}).\label{eq:MAP}
\end{align}
The purpose is to show equivalence between the shrinkage estimator and an appropriate choice of prior distribution. The MAP formulation excludes estimation of parameters related to the prior distribution, which in our case would be $\lambda$. For fixed $\lambda$, the following theorem applies.

\begin{theorem}[MAP and shrinkage]\label{the:shrink}
  Assume the linear regression model 
  \begin{align}
  \vect{y}_t-\yhatBt=\left[\vect{I}_m\otimes\left(\yhatTt^\top- \yhatBt^\top\ST^\top\right)\right]\betaT+ \vect{\epsilon}_t;\quad \vect{\epsilon}_t\sim N(\vect{0},\Sigr),\label{eq:shrinkDesign}
\end{align}
  with prior distribution $\betaT\sim N(\betaOT,\Sigr\otimes\vect{\Sigma}_{\beta,0})$ and
  \begin{align}
    \begin{split}
      \betaOT=&\vectoriz\left[\left(\SighHatdT + \ST\SighHatdB \ST^\top\right)^{-1} \ST\SighHatdB\right]\\
 \vect{\Sigma}_{\beta,0}=&\frac{1-\lambda}{\lambda T}\left(\SighHatdT+ \ST \SighHatdB\ST^\top\right)^{-1}.\label{eq:priorBeta}
 \end{split}
  \end{align}
  With the choice of design matrix given in \eqref{eq:shrinkDesign}, the MAP estimate of $\vect{\beta}$ is independent of $\Sigr$ and can be written as (with $\betaHatT=\vectoriz(\betaHatT^m)$)
  \begin{align}
    \begin{split}
    \betaHatT^m=& \left((1-\lambda)\left(\YhatT-\YhatB\ST^\top\right)^\top\left(\YhatT-\YhatB\ST^\top\right) +   \lambda T \left(\SighHatdT+\ST\SighHatdB\ST^\top\right)\right)^{-1}\times\\
&\left( (1-\lambda) \left(\YhatT-\YhatB\ST^\top\right)^\top\left(\vect{Y}-\YhatB\right) +
  \lambda T \ST\SighHatdB\right)\label{eq:betaTm},
  \end{split}
\end{align}
\update{where $\YhatT=\vectoriz^{-1}(\yhatT)\in \mathbb{R}^{T\times (n-m)}$ and $\YhatB=\vectoriz^{-1}(\yhatB)\in \mathbb{R}^{T\times m}$ are the collection of top- and bottom-level forecasts, respectively.} 
The solution is equivalent to $\vect{P}(\lambda)$ (see \eqref{eq:P.lambda}) in the sense that
\begin{align}
\vect{P}(\lambda)^\top=\left[\begin{matrix}
\vect{0}\\ \vect{I}
\end{matrix}\right] + 
\left[\begin{matrix}
\vect{I}\\ -\ST^\top
\end{matrix}\right]\betaHatT^m.\label{eq:35}
\end{align}
\end{theorem}
\begin{proof}
    Relies on results presented in \citet[Theorem 1]{wickramasuriya2019optimal}. See Appendix \ref{proof:shrink} for details.
\end{proof}

The problem defined by \eqref{eq:shrinkDesign}--\eqref{eq:priorBeta} is equivalent to the generalised Tikhonov regularisation \citep[see, e.g.,][]{Kalivas2012}; that is, the solution to the minimisation problem 
\begin{align}
  \arg\min_{\betaT}\left(||\vect{X}\betaT-(\vect{y}-\yhatB)||_{\vect{\Sigma}^{-1}}^2 + 
  ||\betaT-\betaOT||^2_{\vect{\Sigma}_{\beta}^{-1}}\right)\label{eq:TikReg},
\end{align}
with $\vect{X}$ defined by \eqref{eq:shrinkDesign}, $\vect{\Sigma}=\vect{I}_T\otimes\Sigr$ and
the shrinkage target defined by \eqref{eq:priorBeta}. Other targets could be chosen. For example, the usual Ridge regression ($\betaOT=\vect{0}$ and $\vect{\Sigma}_{\beta}=\frac{1}{\lambda}\vect{I}$, ($\lambda>0$)) would define the bottom-level base-forecast-error variance--covariance as the shrinkage target.

The next corollary and theorem state the MAP estimate of the variance for a given $\lambda$ and a prior for $\Sigr$ such that it is equivalent to the result given in \citet[Lemma 1]{wickramasuriya2019optimal}.
\begin{remark}[MAP estimate of $\Sigr$]\label{remark:MapSig}
  Using the model in Theorem \ref{the:shrink}, the MAP estimate of $\Sigr$ is 
  \begin{align}
    \begin{split}    \SigrHatMAP
    =& \frac{T/(1-\lambda)}{T+n-m}\left(\vect{P}(\lambda)\SigsHat\vect{P}^\top(\lambda)-     \lambda\left(\left(\SighHatdB\right)^{-1}+ \ST^\top\left(\SighHatdT\right)^{-1}\ST
    \right)^{-1}\right).\label{eq:Sig.rMAP} 
    \end{split}
  \end{align}
  \end{remark}
  \begin{proof}
      See Appendix \ref{collo:MapSig}.
  \end{proof}

  The result in Corollary \ref{remark:MapSig} shows that the difference in variance depends on $\lambda$ through a multiplicative factor and on the shrinkage target (prior for the weight matrix).

  The following theorem introduces a prior for the variance that will give a variance estimate equivalent to the usual forecast reconciliation with shrinkage. 
  
  \begin{theorem}[Prior for $\Sigr$]\label{the:SigMap}
    If, in addition to the prior given in Theorem \ref{the:shrink}, the prior for the variance--covariance matrix $\Sigr$ is chosen as the inverse Wishart $\Sigr\sim \mathcal{W}^{-1}(\vect{\Psi},v)$ ($v>m-1$) with 
  \begin{align}
    \vect{\Psi}=\frac{\lambda T}{1-\lambda}\left(\left(\SighHatdB\right)^{-1}+ \ST^\top\left(\SighHatdT\right)^{-1}\ST
    \right)^{-1}\label{eq:Psi},
  \end{align}
  and $v=\frac{\lambda T}{1-\lambda}-(n+1)$ (implying that $\lambda>\frac{m+n}{T+m+n}$), then the MAP estimate of $\SigrHat$ is
  \begin{align}
    \SigrHatShrink=&\vect{P}(\lambda)\SighHat\vect{P}^\top(\lambda).\label{eq:Sigma.rMap}
  \end{align}
  \end{theorem}
  \begin{proof}
 See Appendix \ref{proof:MAP-var}.
  \end{proof}

Note that if $\lambda\leq\frac{m+n}{T+m+n}$ ($v\leq m-1$), then the probability density function of the inverse Wishart is not defined and the MAP interpretation of the above is not valid. However, for any fixed $\lambda>0$, implying that $v>-(n+1)$, \eqref{eq:proofMapVar} in  Appendix \ref{proof:MAP-var} defines the solution \eqref{eq:Sigma.rMap} in a regularisation setting. 

In Theorem \ref{the:SigMap}, the prior for the forecast variance is chosen independently of the prior for the mean. It is possible to choose other priors for $\Sigr$ depending on the specific application, e.g., independence between observations, AR(1) correlation structure, or not to have any prior for the residual variance. The choice above simply shows the equivalence with the formulation of the forecast variance given by \cite{wickramasuriya2019optimal}. In the formulation of Theorems \ref{the:shrink} and \ref{the:SigMap}, $\lambda$ is the only hyperparameter. 

In the previous derivation we have included REML estimates for each of the models. A similar result for the model presented in Theorems \ref{the:shrink} and \ref{the:SigMap} is given in the corollary below.
\begin{remark}[REML and MAP]\label{remark:RemlMap}
    Using the model in Theorems \ref{the:shrink}--\ref{the:SigMap}, the REML correction term (defined as the log-determinant of the Hessian of the log-likelihood wrt.~$\vect{\beta}$) is 
    \begin{align}
     \log |\vect{\vect{X}}^\top\vect{\Sigma}^{-1}\vect{\vect{X}}+\vect{\Sigma}_{\beta}^{-1}|.
    \end{align}
    If $\vect{x}_{i,t}=\vect{x}_{j,t}\in \mathbb{R}^{n-m}$, then the MAP--REML estimate of $\Sigr$ is 
    \begin{align}
  \begin{split}
  \SigrHatsREML =&\frac{T}{T-(n-m)(1-\lambda)}\vect{P}(\lambda)\SigsHat\vect{P}^\top(\lambda).
    \end{split}\label{eq:MapReml}
    \end{align}
\end{remark}
\begin{proof}
    See Appendix \ref{proof:RemlMap}. 
\end{proof}

In the general case where $\vect{x}_{i,t}\neq\vect{x}_{j,t}$, the MAP--REML estimate of $\Sigr$ is more complicated since the priors should be recalculated. The obvious application would be that some columns of $\vect{X}$ are removed due to non-significance, in this case the prior needs to be recalculated and the derivative of the determinant wrt.~$\Sigr$ should be recalculated in line with the result in \eqref{eq:Sigma.rREML}.
We will not explore this further here. 

The methods used in this section imply that parameter and forecast variance can be written explicitly on the form that we discuss in the next section. 

\section{Parameter and forecast variance}\label{sec:Par_Fore_Var}
Using the general theory for LMs (univariate and iid.~residuals), the variance of a forecast is given by
\begin{align}
  \hat{V}[y_{t+h}-\hat{y}_{t+h}]=\vect{x}_{t+h}^\top\hat{V}[\vect{\hat{\beta}}]\vect{x}_{t+h} +\hat{\sigma}^2.\label{eq:39}
\end{align}

In the cases covered in this article, \eqref{eq:39} can be formulated as 
\begin{align}
  \hat{V}[\vect{y}_{t+h}-\vect{\tilde{y}}_{t+h}]=\hat{V}[\vect{X}_{\cdot,t+h}\vect{\hat{\beta}}+\vect{\epsilon}_{t+h}]=
  \vect{X}_{\cdot,t+h}\hat{V}[\vect{\hat{\beta}}] \vect{X}_{\cdot,t+h}^\top +\SigrHat.\label{eq:forcastCov}
\end{align}
The first term relates to the uncertainty of the estimated parameters (referred to as reconciliation-matrix estimation error by \cite{pritularga2021stochastic}). The second term relates to the usual stochastic uncertainty (referred to as the irreducible error in \cite{panagiotelis2021forecast}).

The exact form of $\hat{V}[\vect{\hat{\beta}}]$ depends on the model formulation, but in the general case it is 
\begin{align}
  \hat{V}[\vect{\hat{\beta}}] = \left(
  \vect{X}^\top\left(\vect{I}_T\otimes
    \SigrHat^{-1}\right)\vect{X}
    \right)^{-1}  .\label{eq:41}
\end{align}
In case $\vect{X}_{\cdot,t}=\vect{I}_m\otimes\vect{x}_t^T$, this simplifies to
\begin{align}
  \hat{V}[\vect{\hat{\beta}}] = 
  \SigrHat\otimes 
  \left(\vect{X}^\top_{1,\cdot} \vect{X}_{1,\cdot}\right)^{-1}.\label{eq:42}
\end{align}

Any of the  estimates of $\SigrHat$ can be applied (usually the REML estimate would be preferred). When \eqref{eq:forcastCov} is used with $\SigrHatsREML$\rev{, i.e., taking the parameter uncertainty into account}, it will be referred to as
\begin{align}
\SigrHatPar(\vect{X}_{\cdot,t+h})=\vect{X}_{\cdot,t+h}\hat{V}[\vect{\hat{\beta}}] \vect{X}_{\cdot,t+h}^\top +\SigrHatsREML,\label{eq:ForeVar}
\end{align}
\rev{where the subscript {\tt par} is used to stress the dependence on the parameter variance--covariance matrix. Further,} we will use the short-hand notation $\SigrHatPar:=\SigrHatPar(\vect{X}_{\cdot,t+h})$.

In the presence of prior distributions, we can write the MAP estimator \eqref{eq:betaTm} as an affine transformation  of the ML estimate
\begin{align}
  \betaHatT=
\vect{W}_{1,\lambda}\vect{W}_2\betaHatTML + \vect{W}_{1,\lambda}\vect{\Sigma}_{\beta}^{-1}\betaOT,\label{eq:W1}
\end{align}
where $\vect{W}_{1,\lambda}=\left(\vect{X}^\top(\vect{I}_T\otimes\Sigr)^{-1}\vect{X} +   \vect{\Sigma}_{\beta}^{-1}\right)^{-1}$,
  $\vect{W}_2=\vect{X}^\top (\vect{I}_T\otimes\Sigr)^{-1}\vect{X}$, and $\betaHatTML=(\vect{X}^\top(\vect{I}_T\otimes\Sigr)^{-1}\vect{X})^{-1}\vect{X}^\top (\vect{I}_T\otimes\Sigr)^{-1}(\vect{y}-\yhatB)$ is the ML estimator of $\betaHatT$ \citep{vanWieringen2023}. Ignoring the variance of the priors (i.e., $\SighHatd$), the variance of $\betaHatT$ is 
\begin{align}
 \begin{split} \hat{V}[\betaHatT]=& \vect{W}_{1,\lambda}\vect{W}_2V[\betaHatTML]\vect{W}_2\vect{W}_{1,\lambda}\\
 =&\SigrHat\otimes\left((\vect{X}_{1,\cdot}^\top\vect{X}_{1,\cdot} + \vect{\Sigma}_{\beta,0})^{-1} \vect{X}_{1,\cdot}^\top\vect{X}_{1,\cdot} (\vect{X}_{1,\cdot}^\top\vect{X}_{1,\cdot}+\vect{\Sigma}_{\beta,0})^{-1}\right)\label{eq:RidgeVar}.
   \end{split}
\end{align}

Using \eqref{eq:35} and \eqref{eq:vect1} we can find the variance--covariance of the weight matrix as 
\begin{align}
\begin{split}
\hat{V}\left[\vectoriz\left(\vect{P}^\top\right)\right] =&
\hat{V}\left[\vectoriz\left(\left[\begin{matrix}
\vect{0}\\ \vect{I}
\end{matrix}\right] - 
\left[\begin{matrix}
\vect{I}\\ \ST^\top
\end{matrix}\right]\betaHatT^m\right)\right]\\
=&
\left(\vect{I}_m\otimes\left[\begin{matrix}
\vect{I}\\ \ST^\top
\end{matrix}\right]\right)
\hat{V}\left[\betaHatT\right]
\left(\vect{I}_m\otimes\left[\begin{matrix}
\vect{I}\\ \ST^\top
\end{matrix}\right]\right)^\top.
\end{split}\label{eq:VarP}
\end{align}
This can be used for calculating standard errors of the weight matrix and pairwise correlations. Further, the standard errors are the basis for the usual Wald test.  

\subsection{Model reduction}\label{sec:statTest}
The theory introduced in the previous sections supports general test strategies, such as ANOVA-type (likelihood-ratio) tests for specific hypothesis and partial (Wald) tests for specific parameters. The latter supports an exploratory approach where the least significant parameters are removed one by one. 

The obvious null hypothesis is that some parameters could be equal to zero. Given the parameter variance, it is straight forward to construct the Wald test--statistics as 
\begin{align}
  z_{\text{obs},i}=\frac{\betaHatTi}{\sqrt{\hat{V}[\betaHatT]_{ii}}}.\label{eq:WaldTest}
\end{align}
This can be compared to a Student $t$-distribution with appropriate number of degrees of freedom, which can be approximated by a standard normal in most realistic examples.

In light of Lemma \ref{the:proj}, Lemma \ref{lemma:Whishart}, and Corollary \ref{col:CentralSig}, a reasonable test statistic for the null hypothesis $\vect{\beta}=\vect{0}$ against the alternative would be 
\begin{align}
F_I = \frac{\vect{1}^\top \left(\vect{\tilde{Y}}\vect{S}_I^\top-\vect{\hat{Y}}_I\right)^\top\left(\vect{\tilde{Y}}\vect{S}_I^\top-\vect{\hat{Y}}_I\right)\vect{1}/(n-m)}{\vect{1}^\top \vect{S}_I\left(\vect{Y}-\vect{\tilde{Y}}\right)^\top\left(\vect{Y}-\vect{\tilde{Y}}\right)\vect{S}_I^\top\vect{1}/(T-n+m)}=\frac{\sum_{i,j\in I}\vect{SS}_{\text{mod},ij}/(n-m)}{\sum_{i,j\in I} \vect{SSE}_{ij}/(T-n+m)},\label{eq:F.test}
\end{align}
where the second equality is notation and $F_I$ should be compared to an $F$-distribution with $n-m$ and $T-n+m$ degrees of freedom. When shrinkage is applied, the parameter variance will depend on the shrinkage parameter. We will not explore this point further.

\rev{\section{Probabilistic scores}\label{sec:scores}}

In the following two sections we will use a number of scoring metrics to compare the presented methods. Each of these scoring metrics are briefly presented below, for more details, see, e.g., \cite{gneiting2007} or \cite{bjerregaard_et_al2021}. In all cases the score depends on a distribution assumption. In our case, we assume the normal distribution and depending on the score it might be evaluated in a multi- or uni-variate setting. The scores below are presented for one (multivariate) observation, and the reported scores are averages over the entire data set used for testing.  

\subsubsection*{Coverage rate}
The coverage rate (covR) measures the relative number of observations inside an interval at some level $1-\alpha$.  We use $\alpha =0.05$. The coverage rate should be compared to the nominal level (in our case 0.95). The coverage rate score is very intuitive, but not a proper score. 

\subsubsection*{Interval score}
The interval score compares data to a prediction interval at some level ($1-\alpha$) and is given by
\begin{align}
  \text{IntS}_{1-\alpha}(F,y)= U_{\alpha}-L_{\alpha} + \frac{2}{\alpha}\left(L_{\alpha}-y\right)I(y<L_{\alpha}) +
  \frac{2}{\alpha}\left(y-U_{\alpha}\right)I(y>U_{\alpha}),
\end{align}
\revTwo{where $U_{\alpha}$ and $L_{\alpha}$ are the upper and lower endpoint of the $1-\alpha$ prediction interval for the probabilistic forecast $F$ and $y$ is a scalar}. Hence, it is the width of the interval plus a penalty for observations outside the interval.
\revTwo{The average interval score it calculated as
\begin{align}
 \overline{\text{IntS}}_{1-\alpha}(F,\boldsymbol{y})=\frac{1}{Tm}\sum_{i=1}^T\sum_{j=1}^m\text{IntS}_{1-\alpha}(F_{i,j},\boldsymbol{y}_{i,j}),
\end{align}
where $\boldsymbol{y}_{i,j}$ is element $j$ of $\boldsymbol{y}_i$.} Again, we use $\alpha=0.05$ in all cases.

\subsubsection*{Continuous Ranked Probability Score}
The Continuous Ranked Probability Score (CRPS) compares the observations to the estimated distribution function
\begin{align}
  \text{CRPS}(F,y)=\int_{-\infty}^yF(u)^2du+ \int_{y}^{-\infty}(1-F(u))^2du,
\end{align}
\revTwo{where $F(u)$ is the cumulative distribution function}. Using CRPS point-wise (i.e, uni-variate) and assuming the normal distribution, the score allows an explicit distribution \citep{gneiting2005} and \revTwo{therefore a} fast evaluation. \revTwo{The average CRPS is calculated as
  \begin{align}
\overline{\text{CRPS}}(F,\boldsymbol{y})=\frac{1}{Tm}\sum_{i=1}^T\sum_{j=1}^m\text{CRPS}(F_{i,j},\boldsymbol{y}_{i,j}).  
  \end{align}
}
\subsubsection*{Log score}
The log-score can be given as a uni- or multivariate score. We used the multivariate version in the simulation study and both in the case study. In our setting, the multivariate log-score is defined as
\begin{align}
\text{LogS}(\vect{\mu}_i,\vect{\Sigma}_r;\vect{y}_i)=-\log \phi(\vect{y}_i;\vect{\mu}_i,\vect{\Sigma}_r)\label{eq:logS},
\end{align}
where $\phi(\cdot)$ is the multivariate normal density, $\vect{\mu}_i$ is the forecast mean ($\vect{\hat{y}}_B$ or $\vect{\tilde{y}}$ depending on the setting), and $\vect{\Sigma}_r$ is one of the proposed variance--covariance estimators.
\revTwo{The average log-score is calculated as
  \begin{align}
\overline{\text{LogS}}(\vect{\mu},\vect{\Sigma}_r;\vect{y})=\frac{1}{T}\sum_{i=1}^T\text{LogS}(\vect{\mu}_i,\vect{\Sigma}_r;\vect{y}_i).
\end{align}
}
The marginal log-score is defined as
\begin{align}
\text{LogS}_{\text{marg}}(\vect{\mu}_{i,j},\vect{\Sigma}_{jj};\vect{y}_{i,j})=-\log \phi(\vect{y}_{i,j};\vect{\mu}_{i,j},\vect{\Sigma}_{jj})\label{eq:logSmarg}.
\end{align}
\revTwo{where $\vect{\Sigma}_{jj}$ is the $j$'th diagonal element of $\vect{\Sigma}_r$. The average marginal log-score is defined as
\begin{align}
  \overline{\text{LogS}}_{\text{marg}}(\vect{\mu}_{i,j},\vect{\Sigma}_{jj};\vect{y}_{i,j})=
  \frac{1}{Tm}\sum_{i=1}^T\sum_{j=1}^m\text{LogS}_{\text{marg}}(\vect{\mu}_{i,j}, \vect{\Sigma}_{jj};\vect{y}_{i,j})  .
\end{align}}
The similarity with the log-likelihood function makes the multivariate log-score the natural score in our setting. 

\subsubsection*{Variogram score}
The variaogram score is another multivariate score. It compares the expected variation with the observed variation. It is defined as
\begin{align}
  Vs_p(F_i,\vect{y}_i)=& \sum_{l>k}w_{ij}(|y_{i,k}-y_{i,l}|^p-E_F[|Y_{i,k}-Y_{i,l}|^p])^2,\label{eq:VarS}
  \end{align}
where $F_i$ is the forecast distribution, $\vect{y}_i$ is a set of observations, $w_{ij}$ (we use $w_{ij}=1$) is a weight function and $p$ is the order of the variogram score. We follow the recommendation of \cite{Scheuerer_Hamill2015} and choose $p=1/2$.  \revTwo{The average variagram score is defined as
\begin{align}
  \overline{Vs}_p(F,\vect{y})=&\frac{1}{T}\sum_{i=1}^TVs_p(F_i,\vect{y}_i).
\end{align}
}

\subsubsection*{Average and relative scores}
In all cases, we report the average score and relative scores given as
\begin{align}
  \text{relSc}(F_{benc},F_{model},\vect{y})= \frac{Sc(F_{model},\vect{y})-Sc(F_{benc},\vect{y})}{Sc(F_{benc},\vect{y})}\cdot 100\%\label{eq:relVs}.
\end{align}
\revThree{Here, $F_{benc}$ refers to a base line benchmark model. }
Thus, negative values imply that the model is better than the base model. The relative score is used for the interval score, CRPS, and the variogram scores. For coverage rate, we simply report the numerical value of the score. As the log-score is not bounded below zero, we report the average difference instead of a relative log-score.

\rev{\section{Simulation study}\label{sec:simStu}}
The purpose of the simulation study is to compare different estimators for the forecast-error variance--covariance matrix and discuss the impact of parameter uncertainty.  We simulate the data-generating process at the bottom level ($\vect{y}_{t}\in\mathbb{R}^4$) using a multivariate AR(1) process 
\begin{align}
  \vect{y}_{t}=\vect{A}\vect{y}_{t-1}+\vect{\epsilon}_t;\quad \vect{\epsilon}_t\sim N(\vect{0},\vect{\Sigma}_{\epsilon})\quad\textrm{and iid.,}\label{eq:simY}
\end{align}
with $\vect{A}_{ii}=0.7$, $\vect{A}_{ij}=-0.25$ ($i\neq j$),
and $(\vect{\Sigma}_{\epsilon})_{ij}=\rho^{|i-j|}$. We present simulation studies using $\rho=0$ and $\rho=0.75$. For each choice of training length $T$, the following steps are done:
\begin{enumerate}
\item Simulate three \revTwo{multivariate} time series of length $T$ \revTwo{using \eqref{eq:simY}}.
\item \revTwo{Create univaritate time series for all aggregation levels ($y^i_t=\vect{S}_i\vect{y}_t$). }
\item\label{step:3} Estimate base forecast models using the first series\revTwo{. These are AR(1) models for all time series.}
\item\label{step:4} Estimate the weight matrix using the second series, \revTwo{i.e., the one-step-ahead errors from the model estimated in step \ref{step:3}}.  
  \item Make one-step-ahead reconciled forecasts \revTwo{(using the model from step \ref{step:3} and the weights estimated in step \ref{step:4}) for} the third series and calculate scores.
\end{enumerate}

The simulated bottom level is aggregated using the summation matrix
\begin{align}
  \vect{S}=\left[\begin{matrix}
    1 &   1 &   1 &   1\\
    1 &   1 &   1 &   0\\
    1 &   1 &   0 &   1\\
    1 &   0 &   1 &   1\\
    0 &   1 &   1 &   1\\
    1 &   1 &   0 &   0\\
    1 &   0 &   1 &   0\\
    1 &   0 &   0 &   1\\
    0 &   1 &   1 &   0\\
    0 &   1 &   0 &   1\\
    0 &   0 &   1 &   1\\
    & \vect{I}
    \end{matrix}\label{eq:STsim}
      \right],
\end{align}
i.e., the top levels are all possible sums of  4, 3, and 2 bottom-level observations.

\begin{figure}
\includegraphics[width=\linewidth]{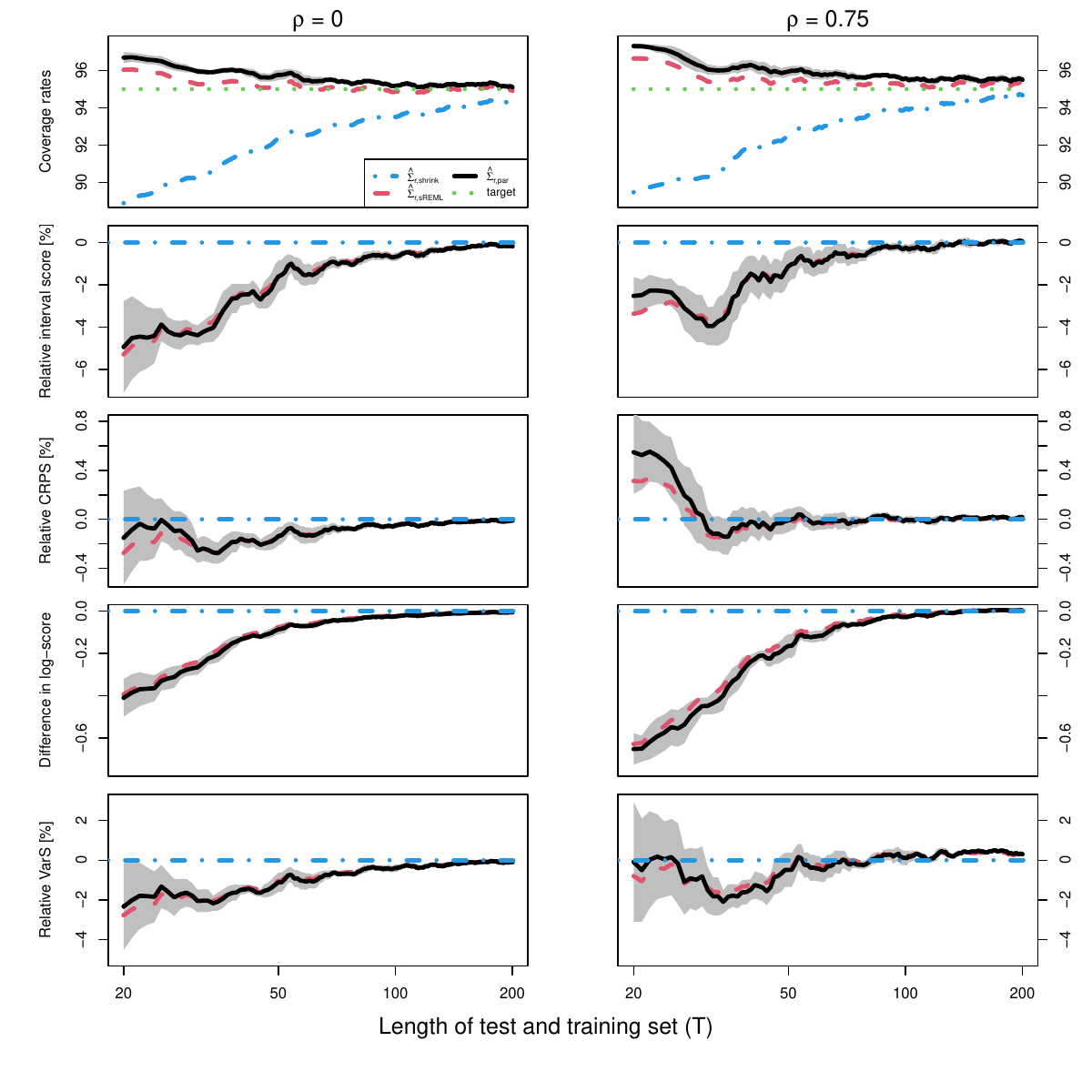} 
\caption{
Different scores in the simulation study \revTwo{using 50 realizations for each training/test length ($T$)}. Top panel: coverage rate of the 95\% prediction interval. 
Second through fifth panel: \revThree{relative interval score for the 95\% prediction interval,} relative CRPS, difference in log score, and  relative  variogram score \revTwo{(in all cases relative to models using $\SigrHatShrink$)}. \revThree{To improve readability of the figure and focus on average behavior,} lines \revThree{are smoothed by} averaging scores over $\{T-4,...,T+4\}$ \revThree{(and the 50 realizations for each $T$)} \revTwo{, where $T$ is the length of the used time series}. Confidence intervals (grey areas) are calculated based on variation over those scores (only given for $\SigrHatPar$). Note the log-scale of the $x$-axis.
}\label{fig:LogSsim}
\end{figure}

Independent, uni-variate AR(1) models are estimated at each level, and the forecast errors are calculated based on the estimated models.
The construction \revThree{of the base forecast models and the resulting summation matrix} implies that \revThree{$\vect{\hat{\Sigma}}_h$} is singular and that the design matrix does not have full rank \revThree{(see Appendix \ref{app:singularSig})}. \revThree{The choice of data-generating process also implies that $\vect{\Sigma}_h$ is singular (see Appendix \ref{app:singularSig}).}
Hence, some kind of shrinkage is needed. This construction illustrates how a highly overparameterised system behaves.

We test three different estimates of the forecast variance:

\begin{enumerate}
  \item The usual MAP shrinkage estimate $\hat{V}[\vect{y}_{t}-\vect{\tilde{y}}_t]=\SigrHatShrink$ (see Theorem \ref{the:SigMap}). 
    \item The MAP--REML shrinkage estimate $\hat{V}[\vect{y}_{t}-\vect{\tilde{y}}_t]=\SigrHatsREML$ (see Corollary \ref{remark:RemlMap}). 
    \item The forecast variance $\SigrHatPar$ defined by \eqref{eq:ForeVar} using the MAP--REML shrinkage estimate of $\Sigr$.
\end{enumerate}

In all three cases we use the optimal value for $\lambda_{\text{opt}}$ given by \cite{schafer2005shrinkage} and implemented in the R-package {\tt corpcor} \citep{corpcor}. We note that the point forecast for the reconciled forecast is the same in all cases, such that the only difference is the estimation of the forecast variance--covariance matrix.

We simulate 50 times for each length of training and test period using $\rho=0$ and $\rho=0.75$. The result in terms of relative scores is shown in Figure \ref{fig:LogSsim}.
For small sample sizes there is a large improvement in the log-score when using one of the REML estimators for both choices of $\rho$, while for larger sample sizes the difference between the estimators becomes small. There is also a small gain from including the parameter uncertainty ($\SigrHatMAP$). It is not surprising that the quality of the forecasts is similar when $T$ is large, as the effect of the REML term will be small in that case.

 The log-score is the natural choice in a simulation study where the data generating process is known, but the figure also includes the other score presented in Section \ref{sec:scores}. The figure shows that in our case the prediction intervals (coverage rates) are generally too wide when using $\SigrHatsREML$ or $\SigrHatPar$ and too narrow when using $\SigrHatShrink$ (Figure \ref{fig:LogSsim}). The (95\%) interval score is generally better for the REML methods. The differences in CRPS are small for all choices of variance--covariance estimators (relative CRPS within 1\%). For $T$ small and $\rho=0.75$, $\SigrHatsREML$ and $\SigrHatPar$ produce slightly higher CRPS than the shrinkage estimator. The variogram score is generally better for the REML estimators when $\rho=0$, while the relative performance between the shrinkage and the REML estimators varies with sample size when $\rho=0.75$.
  
\begin{figure}
  \begin{center}
\includegraphics[width=0.8\linewidth,height=0.5\textheight]{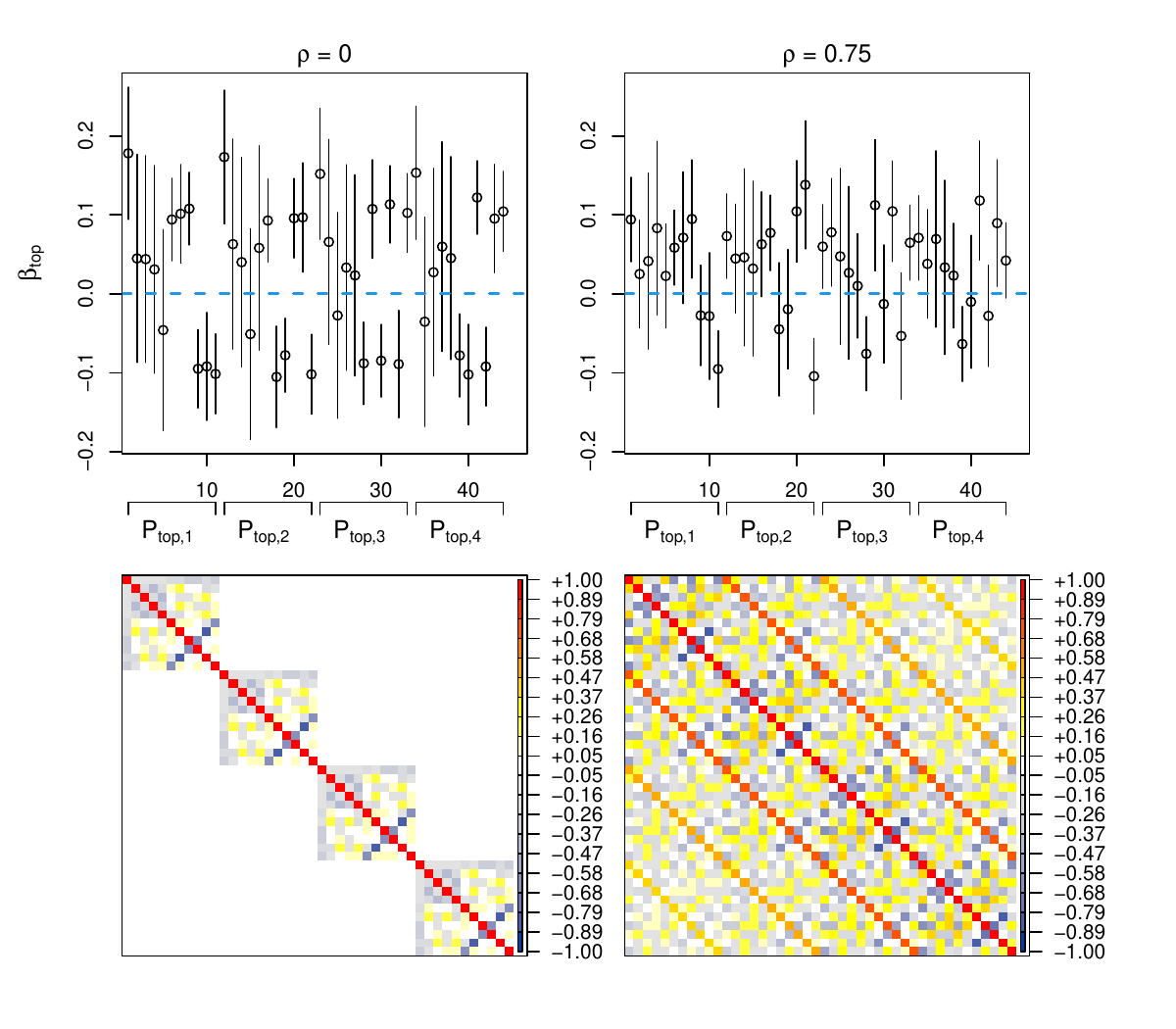} \end{center}
\caption{Top panels: the average estimated weights for the simulations using $T=200$, and $\rho=0$ and $\rho=0.75$ respectively, the plot include approximate 95\% confidence intervals ($\pm 2$ standard errors) for each of the weights. Standard errors are calculated as the average standard errors from the simulated series. $P_{\text{top},i}$ refers to row $i$ of the weights for the top levels (see eq.~\eqref{eq:STsim}). Bottom panel: average correlation matrix corresponding to the average variance--covariance matrix for the estimated parameters (see eq.~\eqref{eq:RidgeVar}).}\label{fig:SimParamPlot}
\end{figure}

Figure \ref{fig:SimParamPlot} presents the average top level weights for the simulation study (using $T=200$). It is clear that the patterns of the estimated weights are quite different between the two cases. In both cases, the standard errors of the weights are quite large (compared to the magnitude of the weights) and many confidence intervals overlap with zero. In the case of $\rho=0$, the weights seem to have a more systematic behavior\revTwo{. In all cases the weight related to the top level is well defined and nonzero. In contrast, the confidence intervals for weights corresponding to sums of three bottom levels overlap with zero, while all weights associated with sums of two bottom levels remain significantly different from zero}. In addition, the separation between the rows is clear from the correlation plot when $\rho=0$. In particular, there is no (or a very small) correlation between the rows of the weight matrix. \revTwo{When $\rho=0.75$, there are no patterns in the weights, and further there is a strong correlation between rows of the weight matrix}. When forecasting, the parameter uncertainty presented in Figure \ref{fig:SimParamPlot} will propagate to the uncertainty of the forecasts, resulting in the improvement of the log-score seen in Figure \ref{fig:LogSsim}.

The presented simulation study showed improvements in the log-score and in most of the other scores, even though the conclusion was less clear. The evaluation of the parameter uncertainty is presented to illustrate how correlation in the data generating process propagated to parameter correlations.

\rev{\section{Forecasting electricity load}\label{sec:case}}
As a case study, we consider forecasts of electricity load in Sweden. The data considered span a period of five years from 2016 to 2020. \revTwo{In order to illustrate the presented methods we perform a five-fold cross validation.} For each of the five years, we use the remaining four years to estimate forecasting models and train a hierarchy \revTwo{(i.e., estimate $\betaHatT$ and the presented versions of $\SigrHat$)}. Three years are used for training a mean-value model, including AR-structures, and one year is used for training the hierarchy. The years used for test and training are given in Table \ref{tab:summCase}. 

\begin{table}[]
\centering
\begingroup\small
\begin{tabular}{l|r|rr|r|rrr}
  & $\frac{||\vect{y}_I-\vect{\hat{y}}_I||^2}{T}$ & $\frac{||\vect{y}_I-\vect{\tilde{y}}_I||^2}{T}$ & $\frac{||\vect{\tilde{y}}_I-\vect{\hat{y}}_I||^2}{T}$ & $\frac{ \frac{\sum\vect{SS}_{mod}}{n-m}}{\frac{\sum \vect{SSE}}{T-n+m}}$ & $\frac{||\vect{y}_I-\vect{\tilde{y}}_I^{\lambda}||^2}{T}$ & $\frac{||\vect{\tilde{y}}_I^{\lambda}-\vect{\hat{y}}_I||^2}{T}$ & $\frac{2(\vect{\tilde{e}}_I^{\lambda})^T\vect{{\tilde{\hat{e}}}}_I^{\lambda}}{T}$ \\ 
  \hline
24h & 209.5 & 89.04 & 120.44 & 12.36 & 102.97 & 95.62 & 10.89 \\ 
  12h & 93.3 & 55.43 & 37.85 & 7.19 & 63.66 & 25.97 & 3.65 \\ 
  8h & 59.6 & 40.54 & 19.06 & 4.78 & 46.90 & 10.31 & 2.38 \\ 
  6h & 47.4 & 32.90 & 14.54 & 4.56 & 38.03 & 7.43 & 1.98 \\ 
  4h & 32.8 & 22.97 & 9.80 & 4.58 & 26.53 & 4.86 & 1.38 \\ 
  3h & 25.3 & 17.67 & 7.60 & 4.6 & 20.42 & 3.68 & 1.16 \\ 
  2h & 17.3 & 12.14 & 5.15 & 4.49 & 14.03 & 2.45 & 0.80 \\ 
  1h & 9.3 & 6.28 & 3.03 & 4.89 & 7.27 & 1.65 & 0.41 \\ 
   \hline
Total & 494.4 & 276.98 & 217.47 & 5.01$^{a)}$ & 319.81 & 151.97 & 22.66 \\ 
  \end{tabular}
\endgroup
\caption{\label{tab:projTabSE} Variance separation \eqref{eq:proj2} on the training set based on the forecasts for 2018 (i.e., training period is 2017). Column one is the sum of squared base-forecast errors on the training set. Columns two and three illustrate the orthogonal projection by adding up to column one. Column four is the test statistic \revTwo{($F_{I}$, see equation \eqref{eq:F.test})}, \revTwo{$^{a)}$ indicates that the number is not the column sum, but rather the test statistic related to all levels of the hierarchy}. Columns five to seven  show the results using shrinkage ($\vect{\tilde{e}}_I^{\lambda}=\vect{y}_I-\vect{\tilde{y}}_I^{\lambda}$ and $\vect{{\tilde{\hat{e}}}}_I^{\lambda}=\vect{\tilde{y}}_I^{\lambda}-\vect{\hat{y}}_I$).} 
\end{table}

The mean-value model includes seasonal and diurnal variation, and a double-seasonal AR-model with weekly and daily variation is fitted to the residuals. \revTwo{ Hence the data used for this case study is
\begin{align}
  \vect{y}=\vect{c}-\vect{\mu}
\end{align}
where $\vect{c}$ is the electricity load and $\vect{\mu}$ is a mean value structure including seasonal and diurnal variation. Further time-series models (double seasonal AR-models) are estimated subsequently, using the usual one-step-ahead residuals, on all aggregation levels (see below). Both models are trained on the training set as indicated in Table \ref{tab:summCase} (Train ARX models). Using the training sets as indicated in Table \ref{tab:summCase} (Train Hierarchy), one day (24 hours) forecasts are made every day at midnight. These forecasts are used for training the hierarchy (estimate $\betaHatT$ and $\Sigr$). Finally, the estimated parameters (related to mean value, AR, and the weight matrix) are used for hierarchical forecast reconciliation on the test set. } 
More details on the model and data can be found in \cite{moller2023}.

Base forecasts are made using the AR models described above to forecast consumption from midnight to midnight. 
The levels of the hierarchy are 1, 2, 3, 4, 6, 8, 12, and 24 hours, implying that $n=60$ and $m=24$. The purpose is to illustrate the estimation methods for the \revTwo{bottom-level reconciled }forecast \revTwo{error} variance--covariance matrix, the parameter variance--covariance ($\hat{V}[\betaHatT]$ and $\hat{V}[\vect{P}]$, see Section \ref{sec:Par_Fore_Var}), and the orthogonal projections. We focus on the bottom level (hourly forecast 1--24 hours ahead) in the accuracy evaluation. Hence, the base forecast model used for the evaluation is the hourly base forecast, with a variance--covariance matrix estimated from the training set of the hierarchy. For example, when evaluating forecasts for 2018, the years 2016, 2019, and 2020 are used for training the mean-value models and 2017 is used to train the hierarchy and estimate the forecast variance--covariance of the 24 hour forecast (see Table \ref{tab:summCase}). 

\revTwo{In four of the five folds, future observations are used when training the hierarchy and/or the forecasting models. This approach is based on treating different years as independent realisations. While this assumption is not strictly valid due to temporal dependence, it is considered reasonable because the forecasting horizon is limited to one day ahead, whereas the observations used across folds are separated by years and are therefore expected to exhibit only weak temporal dependence. The fifth fold is different in that it does not use any future information during training. It represents a realistic forecasting situation where only past data are available. In this sense, the first four folds can be interpreted as a form of validation for the forecast for the year 2020.}

\revTwo{The skill scores is quite stable across all folds (see Table \ref{tab:summCase}). This validation is particularly important in our test case because 2020 was affected by the SARS CoV-2 pandemic, which may have altered electricity demand patterns. In a real operational setting, the model would ideally be updated continuously to capture time-varying dynamics in the data. However, adaptive estimation of the model components is not considered further in this study.}

\newcommand{\STAB}[1]{\begin{tabular}{@{}c@{}}#1\end{tabular}}
\begin{table}[]
\centering
\begingroup\small
\begin{tabular}{|l|l|l|rrrrr|r|}
  \hline
&  \multirow{2}{*}{\STAB{\rotatebox[origin=c]{90}{Period}}}
  &  Test & 2016 & 2017 & 2018 & 2019 & 2020 & Average \\ 
  &  &Train Hierarchy ($\vect{\beta_{\text{top}}}$) & 2017 & 2016 & 2017 & 2018 & 2019 &  \\
  &  & Train (ARX models)  &\small{18-20} & \small{18-20}& \small{16,19,20} & \small{16-17,20}  &\small{16-19} &\\
   \hline
\multirow{26}{*}{\STAB{\rotatebox[origin=c]{90}{Probabilictic skill score}}} & \multirow{2}{*}{\STAB{\rotatebox[origin=c]{90}{RMSE}}} &RMSE(base) & 0.62 & 0.63 & 0.63 & 0.67 & 0.59 & 0.63 \\ 
&&  RRMSE(reconcilied) & -9.00 & -8.83 & -6.98 & -6.16 & -4.34 & -7.06 \\ 
\cline{2-9}
& \multirow{4}{*}{\STAB{\rotatebox[origin=c]{90}{Coverage rate}}} &covR(base) & 95.18 & 94.02 & 94.14 & 93.30 & 95.72 & 94.47 \\ 
&&  covR($\vect{\hat{\Sigma}}_{r,\text{shrink}}$) & $\vect{ 94.75 }$ & 93.52 & 93.55 & 92.89 & $\vect{ 95.09 }$ & 93.96 \\ 
&&  covR($\vect{\hat{\Sigma}}_{r,\text{sREML}}$) & 95.58 & 94.38 & 94.29 & 93.85 & 95.88 & 94.80 \\ 
&&  covR($\vect{\hat{\Sigma}}_{r,\text{par}}$) & 95.72 & $\vect{ 94.60 }$ & $\vect{ 94.59 }$ & $\vect{ 94.19 }$ & 96.04 & $\vect{ 95.03 }$ \\ 
\cline{2-9}
&\multirow{4}{*}{\STAB{\rotatebox[origin=c]{90}{Marginal Log.}}} &$\overline{\text{LogS}}_{\text{marg}}$(base) & 0.829 & 0.871 & 0.868 & 0.917 & 0.840 & 0.865 \\ 
&&  $\Delta$LogS$_{\text{marg}}$($\vect{\hat{\Sigma}}_{r,\text{shrink}}$) & $\vect{ -0.084 }$ & -0.082 & -0.062 & -0.053 & $\vect{ -0.041 }$ & -0.064 \\ 
&&  $\Delta$LogS$_{\text{marg}}$($\vect{\Sigma}_{r,\text{sREML}}$) & -0.080 & -0.086 & -0.065 & -0.058 & -0.035 & -0.065 \\ 
&&  $\Delta$LogS$_{\text{marg}}$($\vect{\Sigma}_{r,\text{par}}$) & -0.080 & $\vect{ -0.090 }$ & $\vect{ -0.069 }$ & $\vect{ -0.063 }$ & -0.034 & $\vect{ -0.067 }$ \\ 
\cline{2-9}
&\multirow{4}{*}{\STAB{\rotatebox[origin=c]{90}{Interval}}} & $\overline{\text{IntS}}$(base) & 3.09 & 3.35 & 3.42 & 3.56 & 2.93 & 3.27 \\ 
&&  relIntS($\vect{\hat{\Sigma}}_{r,\text{shrink}}$) & -6.59 & -7.01 & -3.19 & -2.32 & $\vect{ -3.56 }$ & -4.49 \\ 
&&  relIntS($\vect{\hat{\Sigma}}_{r,\text{sREML}}$) & -6.58 & -7.60 & -3.88 & -3.46 & -2.70 & -4.85 \\ 
&&  relIntS($\vect{\hat{\Sigma}}_{r,\text{par}}$) & $\vect{ -6.73 }$ & $\vect{ -8.15 }$ & $\vect{ -4.74 }$ & $\vect{ -4.64 }$ & -2.64 & $\vect{ -5.42 }$ \\ 
 \cline{2-9}
&\multirow{4}{*}{\STAB{\rotatebox[origin=c]{90}{CRPS}}} &$\overline{\text{CRPS}}$(base) & 0.327 & 0.330 & 0.331 & 0.350 & 0.322 & 0.332 \\ 
&&  relCRPS($\vect{\hat{\Sigma}}_{r,\text{shrink}}$) & $\vect{ -9.43 }$ & $\vect{ -8.85 }$ & $\vect{ -7.77 }$ & $\vect{ -7.03 }$ & $\vect{ -4.49 }$ & $\vect{ -7.52 }$ \\ 
&&  relCRPS($\vect{\hat{\Sigma}}_{r,\text{sREML}}$) & -9.17 & -8.63 & -7.51 & -6.84 & -4.13 & -7.26 \\ 
&&  relCRPS($\vect{\hat{\Sigma}}_{r,\text{par}}$) & -9.14 & -8.63 & -7.55 & -6.90 & -4.08 & -7.27 \\ 
\cline{2-9}
&\multirow{4}{*}{\STAB{\rotatebox[origin=c]{90}{Logarithmic}}} &$\overline{\text{LogS}}$(base) & -5.65 & -4.63 & -4.00 & -3.52 & -1.51 & -3.86 \\ 
&&  $\Delta$LogS($\vect{\hat{\Sigma}}_{r,\text{shrink}}$) & $\vect{ 0.322 }$ & $\vect{ -0.337 }$ & $\vect{ -0.020 }$ & $\vect{ 0.218 }$ & $\vect{ -0.328 }$ & $\vect{ -0.029 }$ \\ 
&&  $\Delta$LogS($\vect{\Sigma}_{r,\text{sREML}}$) & 0.589 & -0.216 & 0.115 & 0.359 & -0.294 & 0.111 \\ 
&&  $\Delta$LogS($\vect{\Sigma}_{r,\text{par}}$) & 0.624 & -0.200 & 0.075 & 0.375 & -0.290 & 0.117 \\ 
 \cline{2-9}
&\multirow{4}{*}{\STAB{\rotatebox[origin=c]{90}{Variorgram}}} &$\overline{\text{Vs}}$(base) & 18.6 & 19.4 & 19.6 & 20.4 & 19.5 & 19.5 \\ 
&&  relVs($\vect{\hat{\Sigma}}_{r,\text{shrink}}$) & $\vect{ -3.48 }$ & $\vect{ -4.85 }$ & $\vect{ -2.62 }$ & $\vect{ -3.69 }$ & $\vect{ -2.96 }$ & $\vect{ -3.52 }$ \\ 
&&  relVs($\vect{\hat{\Sigma}}_{r,\text{sREML}}$) & -2.12 & -3.88 & -1.40 & -2.83 & -1.65 & -2.38 \\ 
 &&  relVs($\vect{\hat{\Sigma}}_{r,\text{par}}$) & -2.09 & -3.84 & -1.36 & -2.80 & -1.71 & -2.36 \\
 \hline
  \end{tabular}
\endgroup
\caption{\label{tab:summCase}  Summary statistics for out-of-sample load forecasts at the bottom level (one hour resolution). \revTwo{Row one is the year used as test set, row two indicates the year used for training the hierarchy, and row three is the years (short format) used for training the forecasting models}. Rows four and five are the accuracy of the base forecasts and the relative improvement of RMSE using the reconciled forecast. The remaining rows present the scores of the base forecast and the relative improvements (see Section \ref{sec:scores}).} 
\end{table}

Table \ref{tab:projTabSE} illustrates variance separation \eqref{eq:proj2} when using 2018 as test set. Column one is the sum of squared base-forecast errors on the training set. Columns two and three illustrate the orthogonal projection by adding up to column one. Comparing column four of Table \ref{tab:projTabSE} to an F-distribution with $n-m=36$ and $T-n+m=365-36=329$ degrees of freedom (95\%-quantile is 1.45) shows that the improvement in sum of squared errors is significant on all levels of the hierarchy, with the largest improvement occurring at the top level (24h). When shrinkage is applied, the sums of squared \revThree{reconciled} errors increase \revTwo{(compare columns two and five)} and the projection is no longer orthogonal (last column is not equal to zero). Similar tables for all other years are given in Appendix \ref{sec:projTabs}.

\begin{knitrout}
\definecolor{shadecolor}{rgb}{0.969, 0.969, 0.969}\color{fgcolor}
\begin{figure}
{\centering
\includegraphics[width=\maxwidth]{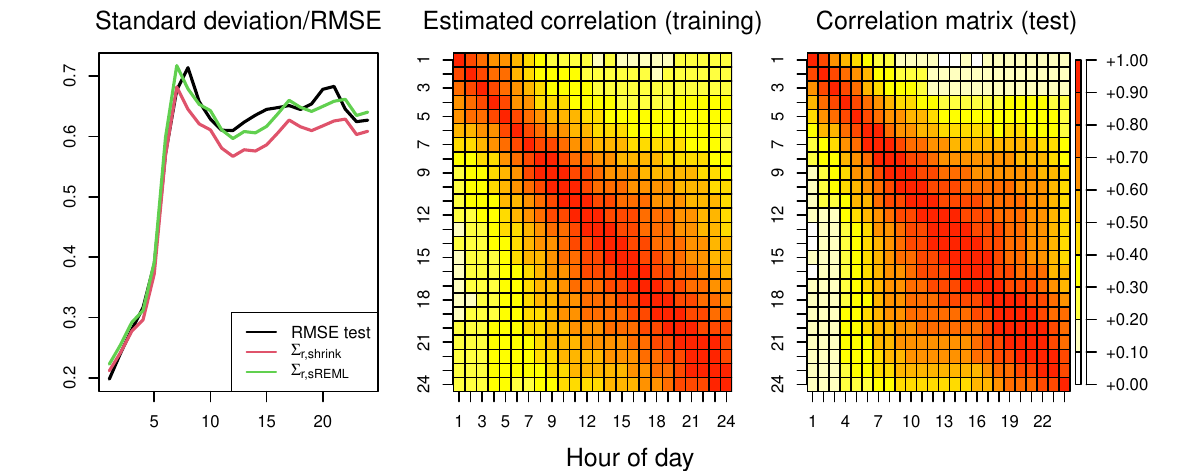} 
}
\vspace{-1cm}
\caption{Left: Out-of-sample RMSE and standard deviation for different estimates of the variance--covariance matrix. Middle: Correlation matrix corresponding to $\SigrHatsREML$. Right: the observed correlation in the test set. The results correspond to using the year 2018 as test set. \label{fig:SweSigr}}
\end{figure}
\end{knitrout}

In addition to the scores presented previously, we consider the relative root mean square error defined by 
\begin{align}
  RRMSE =\frac{RMSE_{\text{Reconciled}}-RMSE_{\text{Base}}}{RMSE_{\text{Base}}} \cdot 100\%.\label{eq:RRMSE}
\end{align}

All evaluated models have the same mean-value prediction and, hence, also the same RRMSE. Using the directly estimated bottom-level base-forecast-error variance--covariance matrix as reference, we evaluate different estimators for the forecast variance--covariance of the reconciled forecast error. A summary of the results is shown in Table \ref{tab:summCase}. RRMSE is improved for all years and the coverage rate (of the 95\% prediction interval) is generally better for the two presented methods, although it differs between \revTwo{folds}.

The marginal log-score improves relative to the base forecast and the two REML based methods perform slightly better than the shrinkage estimator, although there is variation between the years. The interval score improves for all methods and is slightly better using $\vect{\hat{\Sigma}}_{\text{sREML}}$ and $\vect{\hat{\Sigma}}_{\text{par}}$, with the exception of 2020 and 2016. There is a large improvement in CRPS using reconciliation, but the difference between the methods is small. $\vect{\hat{\Sigma}}_{\text{shrink}}$ performs slightly better for all years.

The multivariate scores generally favor $\vect{\hat{\Sigma}}_{\text{shrink}}$. The log-score shows huge variability, and in some years all reconciled forecasts are worse than the base forecasts. The variogram score improves for all variance--covariance estimates and is much more stable across the years than the log-score.  \revTwo{The log-score is sensitive to deviations from the model assumptions \citep{bjerregaard_et_al2021}. In particular, the model assumes that the residual variance ($\vect{\Sigr}$) is constant over time. This assumption may not hold in practice, and residual analysis (not shown here) suggests the presence of seasonal patterns} \revThree{(in all years)}\revTwo{, with higher variance during winter than in summer. Such temporal variation is not currently accounted for in the model, and developing a framework to incorporate it is left for future work.} \revThree{Deviations from the assumed distribution, such as extreme (multivariate) observations, may also occur, and these observations can have a substantial impact on the log-score.}

\begin{knitrout}
\definecolor{shadecolor}{rgb}{0.969, 0.969, 0.969}\color{fgcolor}\begin{figure}
{\centering
\includegraphics[width=1\linewidth]{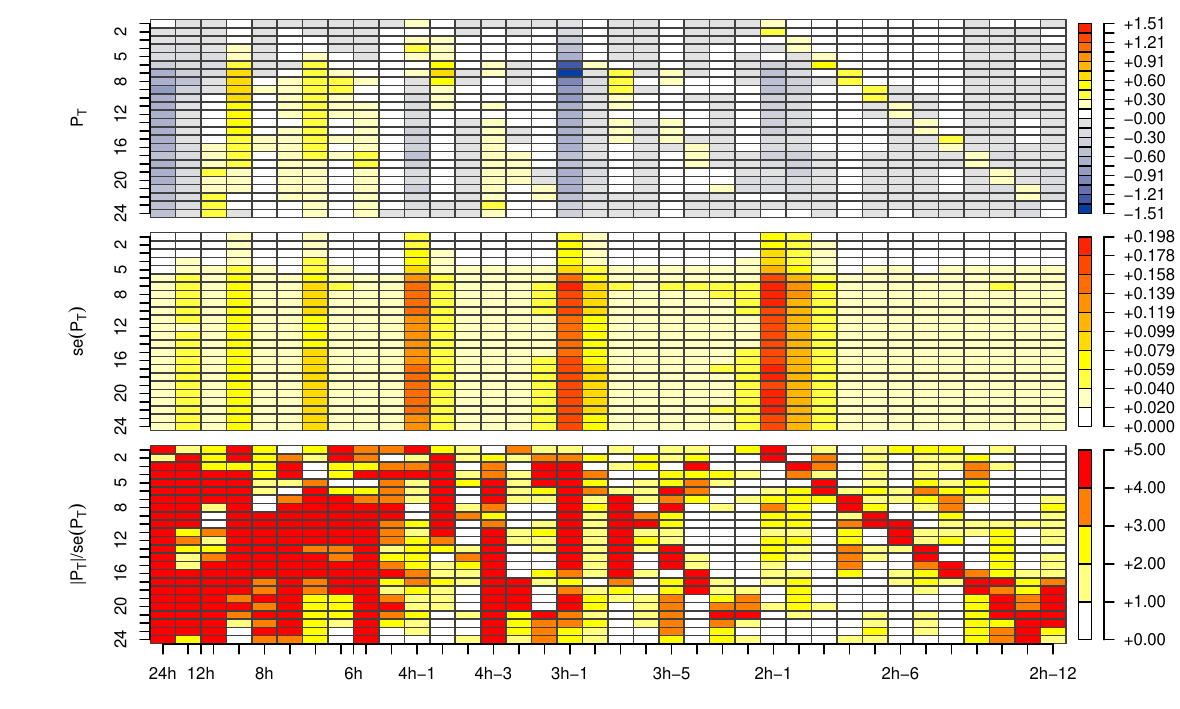}
}
\vspace{-1cm}
\caption[]{Top: volume-weighted weight matrix. Middle: corresponding standard errors. Bottom: absolute value of the weights divided by the standard errors \revTwo{(i.e., $|z_{\text{obs},i}|$, see equation \eqref{eq:WaldTest})}. "4h-1" refers to the first 4 hour forecast and similarly for "3h-5" and so on. The results are for the model using 2018 as test set.}\label{fig:PlotWeights}
\end{figure}
\end{knitrout}

Figure \ref{fig:SweSigr} \revTwo{shows the estimated correlation matrix and standard deviations when 2018 is used as test set. A clear autocorrelation structure is visible, along with an increasing variance in the residuals as the forecast horizon extends from one to seven hours, which is qualitatively consistent with an AR process. The standard deviation is not entirely stable, showing some fluctuations and tending to be higher in the morning and evening hours. In this case, the estimated variance on the test set closely matches the estimated standard deviations. Although the overall correlation structure is similar between estimates from the test set, there are still noticeable differences that are likely to influence both the variogram and the log-score.  
Similar plots for the other years (Figure \ref{fig:SweSigrSE1} in Appendix \ref{app:case}) reveal clear differences between the years. In particular, 2020 stands out as being different from the others, with a lower RMSE on the test set compared to the other years.}

\begin{figure}
\begin{center}
\vspace{-1cm}
\includegraphics[width=0.35\textheight,height=0.35\textheight]{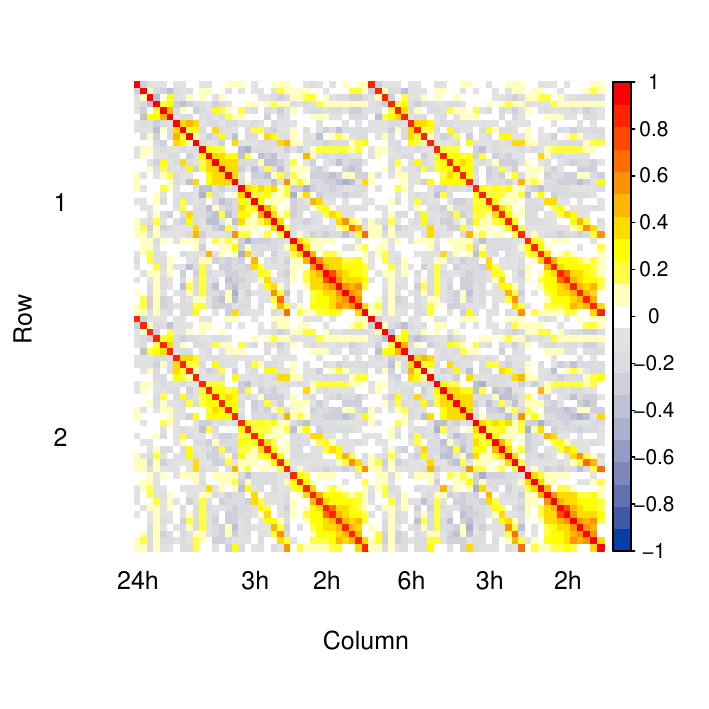} 
\vspace{-1cm}
\end{center}
\caption{Parameter correlation for the first two rows in the estimated weight matrix when using 2018 as test set.\label{fig:SweCorrPlot}}
\end{figure}

The weights, their standard errors, and the ratio between the absolute values of the weights and their standard errors (absolute value of the Wald test statistic for the hypothesis that a weight is zero) are shown in Figure \ref{fig:PlotWeights} for the case where 2018 is used as test set (see Appendix \ref{app:case}, Figure \ref{fig:PlotWeightsSE1_4} for other years). \revTwo{No clear overall pattern is evident in the weights or their standard errors. However, the Wald test statistics are notably large for the two-hour forecasts when used as explanatory variables for the corresponding one-hour forecasts. From a model selection perspective, several weights appear to be strong candidates for removal.}

The correlation matrix for the first two rows and the first 36 columns in the weight matrix (i.e., the first $72$ of 864 elements of $\betaHatT$) is shown in Figure \ref{fig:SweCorrPlot}\revTwo{, when using 2018 as the test set (see Figure \ref{fig:SweCorrPlotSE1Row1}, for the other folds)}. It is clear that there is correlation between some weights from the same row and that the correlation is decreasing with the distance in time. The correlation between weights from different levels (e.g., weights related to two- and three-hour forecasts) is generally low both within and between rows. Further, the correlation between row one and two is strong for the same predictor but decreases for rows further away (see Figure \ref{fig:SweCorrPlotFull} in \ref{app:case} for the full correlation matrix). This decreasing correlation is a consequence of the decreasing correlation in the residual correlation matrix shown in Figure \ref{fig:SweSigr}. 

\subsection{Discussion of case study}
The scoring on the different test sets showed mixed results, but the marginal log-score generally improved using $\SigrHatsREML$ and $\SigrHatPar$ compared to $\SigrHatShrink$, while the opposite was the case for the multivariate log-score. Hence, in the presented case the marginal variance is generally better using $\SigrHatsREML$ and $\SigrHatPar$, while covariances are generally too large compared to $\SigrHatShrink$. In the case study, the residual variation has a strong seasonal variation (not shown), with large residual variance in the winter and smaller variation in summer. Therefore, the models might benefit from a time-dependent variance--covariance structure that could take this variation into account.

\rev{
  \section{Perspectives, discussion, and future work}\label{sec:persp}}
  The main contribution of this article is the LM formulation of the forecast reconciliation problem, which \revTwo{makes it possible to leverage} the well-known statistical results associated with LMs. In particular, it enables dimensional reduction based on established statistical tests, such as Wald and likelihood ratio (ANOVA-type) tests for model reduction. \revTwo{However, this approach also introduces certain challenges. For example, closed-form parameter solutions, such as those provided in Theorem \ref{the:main}, are no longer available, making iterative estimation procedures necessary. Moreover,} the sparsity result in \eqref{eq:betaHat.i} does not hold when columns of $\vect{X}$ are removed\revTwo{, which calls for the development of an efficient algorithmic solution.}

A major assumption maintained throughout the article is that the reconciled forecast errors ($\vect{\epsilon}_t$ and $\vect{\epsilon}_s$, $t\neq s$) are independent. \revTwo{Although we have not investigated the validity of this assumption in detail, it appears plausible when the reconciled forecasts adequately capture the underlying data-generating process. Nevertheless, the proposed framework allows this assumption to be relaxed. The most straightforward way to illustrate this extension is by formulating the problem within a matrix normal distribution framework.}

The models \revTwo{considered in this article} can be written as
  \begin{align}
    \vect{Y}-\vect{\hat{Y}}_B\sim N_{T\times m}(\vect{\tilde{Y}},\vect{\Sigma}_T,\vect{\Sigma}_r),\label{eq:matNorm}
  \end{align}
  where $\vect{\Sigma}_T$ represents the correlation between rows. We chose $\vect{\Sigma}_T=\vect{I}_T$, although this may be replaced by a parameterised correlation structure, such as an AR(1) structure. One key challenge is that, while the dimension of $\vect{\Sigma}_r$ remains fixed regardless of $T$, the dimension of $\vect{\Sigma}_T$ grows with the number of observations. \revTwo{Consequently, the model introduces a dimensionality issue that must be addressed. Within the framework developed in this article, namely the multivariate normal formulation, the model \eqref{eq:matNorm} can equivalently be written as}
  \begin{align}
    \vect{y}-\vect{\hat{y}}_B\sim N_{T\cdot m}(\vect{\tilde{y}},\vect{\Sigma}_T\otimes\vect{\Sigma}_r).
  \end{align}

Finally, in most realistic applications, the number of parameters in $\vect{\Sigma}_r$ \revTwo{is likely to be large. The formulations presented here nevertheless accommodate} parameterised variance--covariance structures, such as AR(1) structures or structures based on the spatial distance between the bottom-level observation. 

\section{Conclusion}\label{sec:con}
We have formulated the forecast reconciliation problem as a LM in order to enable uncertainty quantification, and demonstrated the equivalence between standard forecast reconciliation and MAP estimation for specific choices of priors. This formulation allowed us to build on established results from regression analysis and derive the REML estimate of the forecast-error variance--covariance matrix. The LM formulation implied that forecast reconciliation can be viewed as an orthogonal projection at every level of the hierarchy, which we exploited to prove in-sample distance-reducing properties.

A further advantage of the proposed LM formulation was that it enables estimation of the parameter (weight matrix) variance--covariance matrix, as well as the computation of parameter standard errors. The simulation \rev{and case} study demonstrated that REML estimation is important for \revTwo{obtaining accurate} estimates of the uncertainty of the reconciled forecasts. Moreover, including the parameter uncertainty in the forecast variance--covariance matrix led to improvements in forecast skill. \revTwo{The case study on electricity load forecasting in Sweden illustrated orthogonality,} the effect of shrinkage,  parameter uncertainty, and parameter correlation in the context of high-dimensional parameter vectors.  

The proposed framework also provided a clear avenue for model reduction. In the case study, the results suggested that many weights could be set to zero. \revTwo{An important direction for future research is to investigate how such reductions affect the overall accuracy of the reconciled forecasts. In particular, further work is needed to determine appropriate procedures for conducting these tests, including the development of guidelines regarding the order in which tests should be performed. The effect of shrinkage on the statistical testing procedures also warrants further investigation.}

In this article, we have focused on establishing the equivalence between forecast reconciliation and the LM formulation, while treating the estimation of the forecast variance--covariance as variants of the observed squared deviations between the reconciled forecasts and observations. Within the proposed framework, parameterised forms of the bottom-level reconciled variance--covariance matrix can be estimated using ML or REML methods.

In summary, we believe that the formulation proposed in this article contributes to the understanding of forecast reconciliation. Specifically, it addresses the question of the uncertainty of the estimated reconciliation weights, at least under the Gaussian assumption. The proposed LM formulation also opens up several avenues for future research.

\bibliographystyle{PNY.bst}
\bibliography{references.bib}

\clearpage
\newpage

\appendix
\counterwithin{figure}{section}
\counterwithin{table}{section}

\renewcommand{\thesection}{\Alph{section}}

\section{Nomenclature}\label{sec:momen}
\LTcapwidth=0.9\textwidth
\begin{longtable}{| p{.34\textwidth} | p{.59\textwidth} |}
  \caption{Symbols used in the article. The first usage is marked in parenthesis \rev{and refers to an equation number}. If the first usage is not in an equation, the first appearance will be around the referenced equation.}\\
  \hline
  Symbol & Explanation  \\
  \hline
  ML & Maximum likelihood \eqref{eq:linCon}\\
  LM & Linear model \eqref{eq:linCon}\\
  REML & Restricted maximum likelihood \eqref{eq:linCon}\\
  MAP & Maximum a posteriori \eqref{eq:ConstrReg1}\\
  GLS & Generalised least squares \eqref{eq:normalEq}\\
  LogS & Log-score \eqref{eq:logS}\\
  $Vs_p$ & Variogram score \eqref{eq:VarS}\\
  RRMSE & Relative root mean square error \eqref{eq:RRMSE}\\
  relVs & Relative variogram score \eqref{eq:relVs}\\
  $n \in\mathbb{N}$  & Dimension of the base forecast \eqref{eq:initMod}\\
  $m\in \mathbb{N}$ & Dimension of the bottom-level observations \eqref{eq:initMod}\\
  $T\in \mathbb{N}$ & Number of observations in the training set \eqref{eq:GLMformFull}\\
  $p_i\in\mathbb{N}$ & Number of parameters for modelling bottom-level observation $i$ \eqref{eq:desingGen}\\
  $p=\sum_{i=1}^mp_i\in\mathbb{N}$ & Total number of parameters for modelling bottom-level observations \eqref{eq:GLMformFull}\\
  $\bar{p}=\frac{1}{m}\sum_{i=1}^mp_i\in\mathbb{N}$ & Average number or regression parameters \revTwo{Theorem \ref{the:main}}\\
  $\lambda\in [0,1]$ & Shrinkage parameter \eqref{eq:shrinkage}\\
  $v\in\mathbb{R}_+$ & Degrees of freedom for prior distribution of $\Sigr$ \eqref{eq:Psi}\\
  $q_I\in\mathbb{N}$ & Number of element in an index set $I$ \eqref{eqVarSep}\\
  $\vect{\epsilon}_t^h\in\mathbb{R}^n$ & Reconciliation error at time $t$ \eqref{eq:initMod}\\
  $\vect{\epsilon}_t$, $\vect{e}_t\in\mathbb{R}^m$ & Collection of regression errors at time $t$ in \eqref{eq:glmForm} and Theorem \ref{the:main}\\   
  $\vect{\epsilon}\in\mathbb{R}^{Tm}$ & Collection of all regression errors \eqref{eq:GLMformFull}\\   
  $\vect{\tilde{y}}_t\in\mathbb{R}^{m}$ & Reconciled forecast on the bottom level at time $t$ \eqref{eq:initMod}. \\
  $\vect{\hat{y}}_t\in\mathbb{R}^{n}$ & Base forecast on all levels at time $t$  \eqref{eq:initMod}\\ 
  $\yhatB=\vectoriz(\YhatB^\top) \in\mathbb{R}^{T m}$ & Collection of all bottom-level base forecasts 
  \eqref{eq:GLMformFull}\\
  $\vect{y}=\vectoriz(\vect{Y}^\top)\in \mathbb{R}^{m\cdot T}$ & Collection of all bottom-level observations \eqref{eq:GLMformFull}\\
  $\vect{y}_{t}\in\mathbb{R}^m$ & Bottom-level observation at time $t$ \eqref{eq:GLMformFull}\\
  $\ytildeFt=\vect{S}\vect{\tilde{y}}_t\in\mathbb{R}^{n}$ & Reconciled forecast for all levels at time $t$ \eqref{eq:y.tildeF}. \\
  $\vect{y}_{i,\cdot}\in\mathbb{R}^T$ & Collection of all observations at bottom level $i$ \eqref{eq:D8}\\
  $\yhatBt$ & Bottom-level base forecast at time $t$ \eqref{eq:ConstrReg}. \\
  $\yhatTt$ & Top-level (i.e., excluding bottom-level) base forecast at time $t$ \eqref{eq:ConstrReg}. \\
  $\vect{x}_{i,t}\in \mathbb{R}^{p_i}$ & One row in the design matrix $\vect{X}_{i,\cdot}$ \eqref{eq:desingGen}\\
  $\vect{\beta},\vect{\hat{\beta}}\in \mathbb{R}^{p}$ & Coefficients in the regression model \eqref{eq:GLMformFull}, \eqref{eq:y.tildeF}\\
  $\vect{\hat{\beta}}_i\in \mathbb{R}^{p_i}$ & Coefficients in the regression model for bottom-level observation $i$ \eqref{eq:betaHat.i} \\
  $\betaT,\betaHatT\in \mathbb{R}^{m(n-m)}$ & Coefficients in the regression model corresponding to top-level weights \eqref{eq:ConstrReg}\\
  $\betaOT$ & Prior mean for $\betaT$ \eqref{eq:priorBeta}.\\     
  $\vect{\beta}^m=\vectoriz^{-1}(\vect{\beta})$ & Matrix version of the regression parameters (may appear with a subscript) \eqref{eq:betaHat.i} \\
  $\vect{Y}=\vectoriz^{-1}(\vect{y})\in \mathbb{R}^{T\times m}$ & Collection of all bottom-level observations \eqref{eq:betaHat.i}\\
  $\vect{\hat{Y}}=[\YhatT\quad \YhatB]\in\mathbb{R}^{T\times n}$ & Collection of all base 
  forecasts (Corollary \ref{remark:reconvar})\\
  $\vect{\tilde{Y}} \in\mathbb{R}^{T\times m}$ & Collection of all bottom-level reconciled forecast \eqref{eq:projBa} \\
  $\YhatB=\vectoriz^{-1}(\yhatB) \in\mathbb{R}^{T\times m}$ & Collection of all bottom-level base 
  forecasts \eqref{eq:betaTm}\\    
  $\vect{\hat{Y}}_I\in\mathbb{R}^{T\times q_I}$ & Column selection of $\vect{\hat{Y}}$, Lemma \ref{the:proj}\\    
  $\YhatT \in\mathbb{R}^{T\times (n-m)}$ & Collection of top-level (i.e., excluding the bottom-level) base forecasts    \eqref{eq:betaTm}. \\
  $\vect{SSE}$, $\vect{SS}_{\text{mod}}$ & Sum of square matrices  \eqref{eq:F.test}\\
  $\vect{S}=[\ST^\top\quad \vect{I}]^\top \in\mathbb{R}^{n\times m}$& Summation matrix \eqref{eq:initMod}\\
  $\ST\in \mathbb{R}^{(n-m)\times m}$& Top-level (i.e., excluding the
  bottom-level) summation matrix \eqref{eq:SumExample}\\
  $\vect{S}_I\in\mathbb{R}^{q_I\times m}$ & A row selection of $\vect{S}$ defined by the indexset $I$, Lemma \ref{the:proj}\\
  $\vect{X}\in \mathbb{R}^{Tm\times p}$ & Design matrix \eqref{eq:GLMformFull}\\
  $\vect{X}_{\cdot,t}\in \mathbb{R}^{m\times p}$ & Design matrix at time $t$ \eqref{eq:glmForm}\\
  $\vect{X}_{i,\cdot}\in \mathbb{R}^{T\times p_i}$ & Design matrix  for bottom observation $i$ \eqref{eq:desingGen2}\\
    $\vect{P}=[\PT\quad \PB]\in\mathbb{R}^{m\times n}$ & The weight matrix s.t. $\vect{\tilde{Y}}=\vect{P}\vect{\hat{Y}}$ \eqref{eq:hier}\\
  $\vect{P}(\lambda)$ & Weight matrix based on shrinkage using $\lambda$ \eqref{eq:shrinkage}\\
  $\vect{H}\in\mathbb{R}^{T\times T}$ & Projection matrix, Lemma (\ref{the:proj})\\
  $\vect{\Sigma}_{h} \in\mathbb{R}^{n\times n}$ & Variance--covariance of the reconciliation error \eqref{eq:initMod}\\  
  $\SighHat \in\mathbb{R}^{n\times n}$ & The specific form given in \eqref{eq:SighHat}\\
  $\vect{\Sigma}_{r}\in\mathbb{R}^{m\times m}$ & Variance--covariance of the regression error (at  time $t$) 
  \eqref{eq:glmForm}\\
  $\sigma^{ij}$ & Element $(i,j)$ of $\vect{\Sigma}_{r}^{-1}$ \eqref{eq:sparse}\\
  $\vect{\hat{\Sigma}}_{r,\text{ML}},\vect{\hat{\Sigma}}_{r,\text{REML}}\in\mathbb{R}^{m\times m}$ & ML and REML 
  estimates/estimators for $\Sigr$ \eqref{eq:Sigma.rML} and \eqref{eq:Sigma.rREML}\\
  $\vect{\Sigma}\in\mathbb{R}^{T\cdot m\times T\cdot m}$ & Variance--covariance of the regression error \eqref{eq:GLMformFull}\\   
  $\SigsHat=(1-\lambda)\SighHat+\lambda\SighHatd\in\mathbb{R}^{n\times n}$ & Variance--covariance of the base-forecast error \eqref{eq:shrinkage}\\
  $\SighHatdB\in\mathbb{R}^{m\times m}$ & $\vect{\Sigma}_{h}^d$ for bottom level  \eqref{eq:priorBeta}\\
  $\SighHatdT\in\mathbb{R}^{(n-m)\times (n-m)}$ & $\vect{\Sigma}_{h}^d$ for  top level  \eqref{eq:priorBeta}\\
  $\vect{\hat{\Sigma}}_{r,\text{MAP}}, \SigrHatShrink\vect{\hat{\Sigma}}_{r,sREML},\SigrHatPar$  & Versions of the MAP estimator for $\Sigr$,  
  \eqref{eq:Sig.rMAP},\eqref{eq:Sigma.rMap},\eqref{eq:MapReml}, \eqref{eq:ForeVar}\\ 
  $\vect{\Sigma}_{\beta}=\Sigr\otimes\vect{\Sigma}_{\beta,0}$ & Prior variance for $\vect{\beta}_T$ \eqref{eq:F9}\\
  $\vect{\Sigma}_{\beta,0}\in \mathbb{R}^{(n-m)\times (n-m)}$& Related to prior variance for $\vect{\beta}_T$ \eqref{eq:priorBeta}\\
  $\vect{\Psi}$ & Prior parameter for $\Sigr$ \eqref{eq:Psi}\\
  $\vect{W}_{1,\lambda}=\left(\vect{X}^\top\vect{\Sigma}^{-1}\vect{X} +   \vect{\Sigma}_{\beta}^{-1}\right)^{-1}$ & Used to derive the variance of MAP estimator \eqref{eq:W1}\\
     $\vect{W}_{2}=\vect{X}^\top \vect{\Sigma}^{-1}\vect{X}$ & Used to derive the variance of MAP estimator \eqref{eq:W1}\\
  \hline
\end{longtable}

\section{Some useful relations}\label{sec:useful}

\subsection{Some results from matrix algebra}
As the proofs in this appendix depend heavily on properties of Kronecker products and vectorisation, we have collected the most important ones here. All expression are collected from \cite{petersen_Matrix2006}, where many useful relations used in this article can be found.

The vectorisation of a matrix $\vect{P}\in\mathbb{R}^{m\times n}$ is given by the invertible map $\mathbb{R}^{m\times n}\rightarrow \mathbb{R}^{mn}$
\begin{align}
  \vectoriz(\vect{P})=[P_{11},\hdots,P_{m1},P_{12},\hdots,P_{m2},\hdots,P_{1n},\hdots,P_{mn}]^\top.
\end{align}
The inverse vectorisation of the vector $\vect{y}\in\mathbb{R}^{m n}$ is given by the invertible map $\mathbb{R}^{mn}\rightarrow \mathbb{R}^{m\times n}$
\begin{align}
\vectoriz_{m\times n}^{-1}(\vect{y})=\left[
\begin{matrix}
y_1 & y_{m+1}& \cdots & y_{n\cdot (m-1)+1}\\
\vdots & \vdots & &\vdots\\
y_m &  y_{2m}&\cdots & y_{n\cdot m}
\end{matrix}
\right].
\end{align}
As the target space ($\mathbb{R}^{m\times n}$) should be clear from the context, we will omit the subscript and use the simpler notation $\vectoriz^{-1}(\cdot)$.

The vectorisation of a product of matrices $\vect{A}\in\mathbb{R}^{k\times l}$, $\vect{B}\in\mathbb{R}^{l\times m}$, and $\vect{C}\in\mathbb{R}^{m\times n}$ can be written as 
\begin{align}
  \vectoriz(\vect{A}\vect{B}\vect{C})=&\left(\vect{C}^\top\otimes \vect{A}\right)\vectoriz(\vect{B})\label{eq:vect2}\\
  =&\left(\vect{I}_n\otimes(\vect{A}\vect{B})\right)\vectoriz(\vect{C})\\
  =&\left((\vect{C}^\top\vect{B}^\top)\otimes \vect{I}_k\right)\vectoriz(\vect{A})\\
  \vectoriz(\vect{A}\vect{B})=&\left(\vect{I}_m\otimes\vect{A}\right)\vectoriz(\vect{B})\label{eq:vect1}\\
  =&\left(\vect{B}\otimes\vect{I}_k\right)\vectoriz(\vect{A}).
  \end{align}
Further, the following relations apply to matrices of appropriate dimensions and invertiability
\begin{align}
\left(\vect{A}\otimes\vect{B}\right)^{-1}=&\vect{A}^{-1}\otimes\vect{B}^{-1}\\
  (\vect{A}\otimes \vect{B})(\vect{C}\otimes \vect{D})=&
  (\vect{A}\vect{C})\otimes (\vect{B} \vect{D}).\label{eq:kronecker.matrix}
\end{align}

If $\vect{A}\in\mathbb{R}^{k\times k}$ and $\vect{B}\in\mathbb{R}^{l\times l}$ then \begin{align}
  |\vect{A}\otimes\vect{B}|=|\vect{A}|^l|\vect{B}|^k.
\end{align}

For derivatives wrt.~variance--covariance matrices the following relations hold \citep[][(61), (57), (100)]{petersen_Matrix2006}
\begin{align}
\frac{\partial \vect{a}^\top\vect{\Sigma}^{-1}\vect{b}}{\partial \vect{\Sigma}}=&-\vect{\Sigma}^{-1}\vect{a}\vect{b}^\top\vect{\Sigma}^{-1}\\
\frac{\partial \log|\vect{\Sigma}|}{\partial \vect{\Sigma}}=&\vect{\Sigma}^{-1}\\
\frac{\partial \text{Tr}(\vect{\Psi}\vect{\Sigma})}{\partial \vect{\Sigma}}=&\vect{\Psi}\label{eq:derivTrace}.
\end{align}
Combining \eqref{eq:derivTrace} and \citet[][(59), (124)]{petersen_Matrix2006} while assuming that $\vect{\Sigma}$ and $\vect{\Psi}$ are symmetric yields
\begin{align}
\begin{split}
    \frac{\partial \text{Tr}(\vect{\Psi}\vect{\Sigma}^{-1})}{\partial \vect{\Sigma}}=&-\vect{\Sigma}^{-1}\vect{\Psi}\vect{\Sigma}^{-1}\label{eq:derivInvTrace}\\
    \frac{\partial \text{Tr}(\vect{\Psi}_1\vect{\Sigma}^{-1}\vect{\Psi}_2)}{\partial \vect{\Sigma}}=&-\vect{\Sigma}^{-1}\vect{\Psi}_1^\top\vect{\Psi}_2^\top\vect{\Sigma}^{-1}.
    \end{split}
\end{align}

For products of vectors the following \citep[eq. (521)]{petersen_Matrix2006} holds
\begin{align}
 \vectoriz(\vect{A})^\top \vectoriz(\vect{B}) = \text{Tr}(\vect{A}^\top\vect{B})\label{eq:B14}.
\end{align}

The following relation, which is related to the matrix normal distribution \citep{Iranmanesh_2010}, will also be useful. Using \eqref{eq:vect2} and \eqref{eq:B14}, we can write
\begin{align}
\begin{split}
\vect{y}^\top(\vect{\Sigma}^{-1}\otimes\vect{\Omega}^{-1})\vect{y}=&\vect{y}^\top\vectoriz(\vect{\Omega}^{-1}\vectoriz^{-1}(\vect{y})\vect{\Sigma}^{-1})\\
=&\text{Tr}\left(\vectoriz^{-1}(\vect{y})^\top\vect{\Omega}^{-1}\vectoriz^{-1}(\vect{y})\vect{\Sigma}^{-1}\right)\\
=&\text{Tr}\left(\vect{Y}^\top\vect{\Omega}^{-1}\vect{Y}\vect{\Sigma}^{-1}\right).\label{eq:MatrixNorm}
\end{split}
\end{align}

\subsection{Probability densities, log-likelihood functions and log-priors}\label{sec:logLike}

In the article, we use the following multivariate distributions: the multivariate normal (Gaussian), the inverse Wishart, and the Wishart distribution. Only the first two are used for likelihood/MAP estimation, but for completeness we include all three in the discussion below.   

\subsubsection{Multivariate normal distribution}
If a random variable ($\vect{y}\in\mathbb{R}^n$) follows a multivariate normal distribution with mean $\vect{\mu}$ and variance--covariance $\vect{\Sigma}$, we write $\vect{y}\sim N(\vect{\mu},\vect{\Sigma})$, with $\vect{\mu}\in\mathbb{R}^n$,  $\vect{\Sigma}\in\mathbb{R}^{n\times n}$, and $\vect{\Sigma}$ positive semi--definite. The probability density function (pdf.), $\phi(\cdot)$, of the multivariate normal is defined when $\vect{\Sigma}$ is positive definite and is given by 
\begin{align}
\begin{split}
\phi(\vect{y};\vect{\mu},\vect{\Sigma})=&\frac{1}{(2\pi)^{n/2}\sqrt{|\vect{\Sigma}|}}e^{-\frac{1}{2}(\vect{y}-\vect{\mu})^\top\vect{\Sigma}^{-1}(\vect{y}-\vect{\mu})}\\
=&\frac{1}{(2\pi)^{n/2}\sqrt{|\vect{\Sigma}|}}e^{-\frac{1}{2}\text{Tr}\left((\vect{y}-\vect{\mu})\vect{\Sigma}^{-1}(\vect{y}-\vect{\mu})^\top\right)}.
\end{split}
\end{align}
For likelihood estimation, we use parameterised versions of the mean and variance, i.e., $\vect{\mu}:=\vect{\mu}(\vect{\beta})$ and $\vect{\Sigma}:=\vect{\Sigma}(\vect{\psi})$. The log-likelihood is 
\begin{align}
\begin{split}
l(\vect{\beta},\vect{\psi}) =& \log(\phi(\vect{y};\vect{\mu},\vect{\Sigma}))\\
\propto & -\frac{1}{2} \log(|\vect{\Sigma}(\vect{\psi})|)- \frac{1}{2}(\vect{y}-\vect{\mu}(\vect{\beta}))^\top\vect{\Sigma}^{-1}(\vect{\psi})(\vect{y}-\vect{\mu}(\vect{\beta})).
\end{split}
\end{align}
For fixed $\vect{\psi}$ or $\vect{\beta}$ we will use the notation
\begin{align}
\begin{split}
l(\vect{\beta})\propto& - \frac{1}{2}(\vect{y}-\vect{\mu}(\vect{\beta}))^\top\vect{\Sigma}^{-1}(\vect{\psi})(\vect{y}-\vect{\mu}(\vect{\beta}))\\
l(\vect{\psi})\propto& -\frac{1}{2} \log(|\vect{\Sigma}(\vect{\psi})|)- \frac{1}{2}(\vect{y}-\vect{\mu}(\vect{\beta}))^\top\vect{\Sigma}^{-1}(\vect{\psi})(\vect{y}-\vect{\mu}(\vect{\beta})).
    \end{split}
\end{align}
When used as a prior distribution, i.e., $\vect{y}|\vect{\beta}\sim N(\vect{\mu}(\vect{\beta}),\vect{\Sigma}(\vect{\psi}))$, $\vect{\beta}\sim N(\vect{\mu}_{\beta},\vect{\Sigma}_{\beta}(\vect{\psi}))$,  we will use the notation
\begin{align}
l_{MAP}(\vect{\beta},\vect{\psi})=\log\left(\phi(\vect{y};\vect{\mu}(\vect{\beta}),\vect{\Sigma}(\vect{\psi})\right) +
\log\left(\phi(\vect{\beta};\vect{\mu}_{\beta},\vect{\Sigma}(_{\beta}\vect{\psi})\right).
\end{align}

\subsubsection{The inverse Wishart distribution}
If a random matrix $\Sigr\in \mathbb{R}^{m\times m}$ ($\Sigr$ positive definite) follows an inverse Wishart distribution with scale matrix $\vect{\Psi}>0$ and degrees of freedom $v>m-1$, we write $\Sigr\sim \mathcal{W}^{-1}(\vect{\Psi},v)$ and the density is given by
\begin{align}
  f_{\Sigma_r}(\Sigr)=
  \frac{
  |\vect{\Psi|}^{v/2}}{2^{vm/2}\Gamma_m\left(\frac{v}{2}\right)} |\Sigr|^{-\frac{m+v+1}{2}} e^{-\frac{1}{2}\text{Tr}(\vect{\Psi}\Sigr^{-1})},
\end{align}
where $\Gamma_m(\cdot)$ is the multivariate Gamma function. The inverse Wishart is the conjugate prior for the variance--covariance matrix in the multivariate normal. When used as a prior distribution (i.e., $\vect{\Psi}$ and $v$ fixed) we will use the form 
\begin{align}
    \log(f_{\Sigma_r}(\Sigr))\propto
    -\frac{m+v+1}{2}\log(|\Sigr|) -\frac{1}{2}\text{Tr}(\vect{\Psi}\Sigr^{-1}).
\end{align}

In this work we use the inverse Wishart as a prior in the setting $\vect{y}|\vect{\beta},\Sigr\sim \allowbreak N(\vect{\mu}(\vect{\beta}),\allowbreak\vect{\Sigma}(\Sigr))$, $\vect{\beta}|\Sigr\sim N(\vect{\mu}_{\beta},\vect{\Sigma}_{\beta}(\Sigr))$, $\Sigr\sim\mathcal{W}^{-1}(\vect{\Psi},v)$, and we use the notation 
\begin{align}
l_{MAP}(\vect{\beta},\Sigr)=\log\left(\phi\left(\vect{y};\vect{\mu}(\vect{\beta}),\vect{\Sigma}(\Sigr)\right)\right) +
\log\left(\phi(\vect{\beta};\vect{\mu}_{\beta},\vect{\Sigma}_{\beta}(\Sigr))\right) + \log(f_{\Sigma_r}(\Sigr)).
\end{align}

\subsubsection{The Wishart distribution}
\update{
In general we use the sum of squares definition of the Whishart distribution. The following definition and lemma is adapted from \cite{Rao_1973}. We start with the general definition of the Wishart distribution.
\begin{definition}[Wishart distribution]\label{def:Whishart}
   Let the matrix $\vect{Z}=\left[\vect{z}_1\quad ...\quad\vect{z}_T\right]^\top\in\mathbb{R}^{T\times p}$ be a collection of normal random variables s.t.~$\vect{z}_i\sim N(\vect{0},\vect{\Sigma}_z)$ and iid., then 
   $\vect{Q}=\vect{Z}^\top\vect{Z}$ follows a Whishart distribution with $T$ degrees of freedom and scale parameter $\vect{\Sigma}_z$, $\vect{Q}\sim W_p(\vect{\Sigma}_z,T)$. If $T>p-1$ and $\vect{\Sigma}_z$ is positive definite, then the probability density function exists and $E[\vect{Q}]=T\vect{\Sigma}_z$.
\end{definition}
The following lemma relates projections and the Whishart distribution. 
\begin{lemma}[Some properties of the Wishart distribution]\label{lemma:Whishart}
    Let $\vect{Z}$ be as in Definition \ref{def:Whishart} and $\vect{H}$ a projection matrix with $r=\text{Tr}(\vect{H})$, then a necessary and sufficient condition for $\vect{Z}^\top\vect{H}\vect{Z}$ to follow a Wishart distribution is that $\frac{\vect{v}^\top\vect{Z}^\top\vect{H}\vect{Z}\vect{v}}{\sigma_v^2}\sim \chi^2(r)$, \rev{with $\sigma_v^2=V[\vect{v}^\top\vect{z}_i]$}, for any fixed vector $\vect{v}$, and in that case $\vect{Z}^\top\vect{H}\vect{Z}\sim W_p(\vect{\Sigma}_z,r)$.
\end{lemma}
\begin{proof}
    See \citet[8b]{Rao_1973}
\end{proof}
}
The general definition of the Wishart distribution used in this article is given in Definition \ref{def:Whishart}. For completeness we also state the pdf. If $\vect{Q}\in\mathbb{R}^{m\times m}$ ($\vect{Q}$ positive definite) follows a Wishart distribution with positiv definite scale matrix $\vect{\Psi}$ and degrees of freedom $v$, we write $\vect{Q}\sim \mathcal{W}(\vect{\Psi},v)$. The pdf is defined for $v>m-1$ and is given by 
\begin{align}
  f_{Q}(\vect{Q})=
  \frac{
  |\vect{\Psi|}^{-v/2}}{2^{vm/2}\Gamma_m\left(\frac{v}{2}\right)} |\vect{Q}|^{\frac{m+v+1}{2}} e^{-\frac{1}{2}\text{Tr}(\vect{\Psi}\vect{Q}^{-1})}.
\end{align}

\section{Proof of Theorem \ref{the:main}}\label{proof:main}
\rev{Let $\vect{e}=[\vect{e}_1^\top\hdots,\vect{e}_T^\top]^\top$, then}
the log-likelihood is given by 
\begin{align}
  \begin{split}
  l(\Sigr;\vect{e}) \propto&
  -\frac{1}{2}\log(|\vect{I}_T\otimes\Sigr|) - \frac{1}{2}\vect{e}^\top (\vect{I}_T\otimes\Sigr^{-1})\vect{e}\\
  &=  -\frac{1}{2}\log(|\Sigr|^T) - \frac{1}{2}\sum_{t=1}^T\vect{e}_t^\top \Sigr^{-1}\vect{e}_t,
  \end{split}
\end{align}
with 
\begin{align}
  \vect{e}_t=\vect{y}_t-\vect{X}_{\cdot,t}\vect{\hat{\beta}}.
\end{align}
$\vect{\hat{\beta}}$ is the ML (or weighted least square) estimate of $\vect{\beta}$ is well-known and,  
given by 
\begin{align}
  \vect{\hat{\beta}}=& (\vect{X}^\top(\vect{I}_T\otimes\vect{\Sigma}^{-1}_r)\vect{X})^{-1} \vect{X}^\top(\vect{I}_T\otimes\Sigr^{-1})\vect{y}.
  \end{align}
By direct matrix multiplications, using the diagonal like structure in \eqref{eq:desingGen}, we find the matrix $  \vect{X}^\top\vect{\Sigma}^{-1}\vect{X}$ as
\begin{align}
  \begin{split}
\vect{X}^\top\vect{\Sigma}^{-1}\vect{X}
&=\left[\begin{matrix}
    \sigma^{11}\sum_{t=1}^T\vect{x}_{1,t}\vect{x}_{1,t}^\top & \sigma^{12}\sum_{t=1}^T\vect{x}_{1,t}\vect{x}_{2,t}^\top &\cdots & \sigma^{1m}\sum_{t=1}^T\vect{x}_{1,t}\vect{x}_{m,t}^\top\\
    \sigma^{12}\sum_{t=1}^T\vect{x}_{2,t}\vect{x}_{1,t}^\top & \sigma^{22}\sum_{t=1}^T\vect{x}_{2,t}\vect{x}_{2,t}^\top &\cdots & \sigma^{2m}\sum_{t=1}^T\vect{x}_{2,t}\vect{x}_{m,t}^\top\\
    \vdots &  & \ddots &\\
    \sigma^{1m}\sum_{t=1}^T\vect{x}_{m,t}\vect{x}_{1,t}^\top &\cdots & &\sigma^{mm}\sum_{t=1}^T\vect{x}_{m,t}\vect{x}_{m,t}^\top
  \end{matrix}\right]\\
&=\left[\begin{matrix}
    \sigma^{11}\vect{X}^\top_{1,\cdot}\vect{X}_{1,\cdot} & \sigma^{12}\vect{X}^\top_{1,\cdot}\vect{X}_{2,\cdot} &\cdots & \sigma^{1m}\vect{X}^\top_{1,\cdot}\vect{X}_{m,\cdot}\\
    \sigma^{12}\vect{X}^\top_{2,\cdot}\vect{X}_{1,\cdot} & \sigma^{22}\vect{X}_{2,\cdot}^\top\vect{X}_{2,\cdot} &\cdots & \sigma^{2m}\vect{X}_{2,\cdot}^\top\vect{X}_{m,\cdot}\\
    \vdots &  & \ddots &\\
    \sigma^{1m}\vect{X}_{m,\cdot}^\top\vect{X}_{1,\cdot} &\cdots & &\sigma^{mm}\vect{X}_{m,\cdot}^\top\vect{X}_{m,\cdot}
  \end{matrix}\right],\label{eq:sparse}
 \end{split}
\end{align}
\update{where $\sigma^{ij}=\left(\vect{\Sigma}^{-1}\right)_{ij}$.}

Using \eqref{eq:vect2} we get 
\begin{align}
(\vect{I}_T\otimes\Sigr^{-1})\vect{y}=\vectoriz(\Sigr^{-1}\vect{Y}^\top).
\end{align}

We treat the special case below.

{\bf Special case 1:} $\vect{x}_{i,t}=\vect{x}_{j,t}=\vect{x}_{t}$, $\forall$ $(i,j)$. In this case we get (using \eqref{eq:sparse})
\begin{align}
  \vect{X}^\top\vect{\Sigma}^{-1}\vect{X}=&\Sigr^{-1}\otimes \sum_{i=1}^T\vect{x}_{t}\vect{x}_{t}^\top=\Sigr^{-1}\otimes \left(\vect{X}_{1,\cdot}^\top\vect{X}_{1,\cdot}\right),\label{eq:XSigX}
\end{align}
and therefore
\begin{align}
  (\vect{X}^\top\vect{\Sigma}^{-1}\vect{X})^{-1}=&\Sigr\otimes \left(\sum_{i=1}^T\vect{x}_{t}\vect{x}_{t}^\top\right)^{-1} =\Sigr\otimes \left(\vect{X}_{1,\cdot}^\top\vect{X}_{1,\cdot}\right)^{-1}.\label{eq:C.7}
\end{align}
Setting $\vect{\xi}=\Sigr^{-1}\vect{Y}^\top$ we can write 
\begin{align}
\begin{split}
    \vect{X}^\top\vectoriz(\vect{\xi})
    =&\left[    \begin{matrix}
      \vect{x}_1 & \vect{0} & \vect{0} &\vect{x}_2 & \cdots & \vect{x}_T &  & 
      \vect{0}\\
      \vect{0} &\ddots & \vect{0}& &\ddots & & \ddots\\
      \vect{0} & \vect{0} &\vect{x}_1 & \vect{0}  &\vect{x}_2 & \cdots & &\vect{x}_T
      \end{matrix}
      \right]\left[
      \begin{matrix}
      \xi_{11}\\ \vdots \\ \xi_{1m}\\ \xi_{21}\\ \vdots\\ \xi_{Tm}
      \end{matrix}
      \right]\\
    =& \left[
      \begin{matrix}
      \vect{x}_1\xi_{11} + \vect{x}_2\xi_{21} +\cdots 
      \vect{x}_T\xi_{T1}\\
      \hdots\\
      \vect{x}_1\xi_{1m} + \vect{x}_2\xi_{2m} +\cdots 
      \vect{x}_T\xi_{Tm}\\
      \end{matrix}
      \right].
      \end{split}
 \end{align}
 Hence, the order can be rearranged to
 \begin{align}
 \begin{split}
  \vect{X}^\top\vectoriz(\vect{\xi})    =&    \left[
      \begin{matrix}
      \vect{X}^\top_{1,\cdot} & \vect{0} & \cdots & \vect{0}\\
      \vect{0} &\vect{X}^\top_{1,\cdot} &  \cdots & \vect{0}\\
      \vdots & & \ddots & \vect{0} &\\
      \vect{0} & \vect{0} & & \vect{X}^\top_{1,\cdot} 
      \end{matrix}
      \right]\left[
      \begin{matrix}
      \xi_{11}\\ \vdots \\ \xi_{T1}\\ \xi_{12}\\ \vdots\\ \xi_{Tm}
      \end{matrix}
      \right]\\
   =&\left(\vect{I}_m\otimes
     \vect{X}_{1,\cdot}^\top\right)\vectoriz(\vect{\xi}^\top)\\
     =&\left(\vect{I}_m\otimes 
     \vect{X}_{1,\cdot}^\top\right)\vectoriz(\vect{Y}\Sigr^{-1})\\ 
    =& \vectoriz\left(\vect{X}_{1,\cdot}^\top\vect{Y} 
      \Sigr^{-1}\right),
    \end{split}
\end{align}
and, consequently, (again using \eqref{eq:vect2})
\begin{align}
  \begin{split}
  \vect{\hat{\beta}}=& \left[\Sigr\otimes\left(\vect{X}_{1,\cdot}^\top\vect{X}_{1,\cdot}\right)^{-1}\right]\vect{X}^\top(\vect{I}_T \otimes\Sigr^{-1})\vect{y}\\
    =&\vectoriz\left(\left(\vect{X}_{1,\cdot}^\top\vect{X}_{1,\cdot}\right)^{-1} \vect{X}_{1,\cdot}^\top\vect{Y}\right).\label{eq:C.9}
    \end{split}
\end{align}

This proves that $\vect{\hat{\beta}}$ is independent of $\Sigr$.

\subsubsection*{Maximum likelihood estimate of $\Sigr$}

The log-likelihood wrt.~$\Sigr$ is given by 
\begin{align}
  \begin{split}
  l(\Sigr;\vect{e}) \propto&
  -\frac{1}{2}\log(|\vect{I}_T\otimes\Sigr|) - \frac{1}{2}\vect{e}^\top (\vect{I}_T\otimes\Sigr^{-1})\vect{e}\\
  &=  -\frac{1}{2}\log(|\Sigr|^T) - \frac{1}{2}\sum_{t=1}^T\vect{e}_t^\top \Sigr^{-1}\vect{e}_t.
  \end{split}
\end{align}
In the general case, $\vect{e}_t$ is a function of $\Sigr$; but as shown above, it is independent of $\Sigr$ in the special Case 1.

The derivative wrt.~$\Sigr$ of the log-likelihood is 
\begin{align}
  \frac{\partial l}{\partial \Sigr}=&
    -\frac{T}{2}\Sigr^{-1} + \frac{1}{2}\Sigr^{-1}\sum_{t=1}^T\vect{e_t} \vect{e}_t^\top\Sigr^{-1},
\end{align}
and, hence, for fixed $\vect{\beta}$, the maximum likelihood estimate of $\Sigr$ is 
\begin{align}
  \vect{\hat{\Sigma}}_{ML,r}=\frac{1}{T}\sum_{t=1}^T\vect{e}_t \vect{e}_t^\top.
\end{align}
In the general case where the error $\vect{e}_t$ is a function of $\SigrHat$, iterations are needed in order to find the best estimate. 

\subsubsection*{REML estimate of $\Sigr$}\label{sec:proofREML}
The restricted log-likelihood wrt.~$\Sigr$ is given by 
\begin{align}
    l_{\text{REML}}(\Sigr;\vect{e}) \propto&
    -\frac{1}{2}\log(|\Sigr|^T) - \frac{1}{2}\sum_{t=1}^T\vect{e}_t^\top \Sigr^{-1}\vect{e}_t - \frac{1}{2}\log(|\vect{X}^\top\vect{\Sigma}^{-1}\vect{X}^\top|).
\end{align}

We divide the answer into a special case and the general case:\\
{\bf Special case:} If $\vect{x}_{i,t}=\vect{x}_{j,t}=\vect{x}_{t}\in\mathbb{R}^{\bar{p}}$, then $\vect{\hat{\beta}}$ and therefore $\vect{e}$ are independent from $\Sigr$, which (using \eqref{eq:XSigX}) leads to
\begin{align}
  \begin{split}
    l_{\text{REML}}(\Sigr;\vect{e}) 
     \propto & -\frac{1}{2}\log(|\Sigr|^T) - \frac{1}{2}\sum_{t=1}^T\vect{e}_t^\top \Sigr^{-1}\vect{e}_t -
     \frac{1}{2}\log\Big|\vect{\Sigma}^{-1}_r\otimes\left(\sum_{t=1}^{T}\vect{x}_t\vect{x}_t^\top\right)\Big|\\
          =& -\frac{T}{2}\log(|\Sigr|) - \frac{1}{2}\sum_{t=1}^T\vect{e}_t^\top \Sigr^{-1}\vect{e}_t - \frac{1}{2}\log\Big|\vect{\Sigma}^{-1}_r\Big|^{\bar{p}}
          \label{eq:proof_llRE},
          \end{split}
\end{align}
and
\begin{align}
    \frac{\partial l_{\text{REML}}}{\partial \Sigr}=&
    -\frac{T}{2}\Sigr^{-1} + \frac{1}{2}\Sigr^{-1}\sum_{t=1}^T\vect{e}_t \vect{e}_t^\top\Sigr^{-1}+\frac{\bar{p}}{2}\Sigr^{-1}.
\end{align}
The resulting REML estimate of $\Sigr$ is
\begin{align}
  \vect{\hat{\Sigma}}_{r,\text{REML}}=\frac{1}{T-\bar{p}}\sum_{t=1}^T\vect{e_t} \vect{e}_t^\top.
\end{align}
{\bf General case:}
In the general case we need the derivative of the REML term
\begin{align}
\frac{\partial \log|\vect{X}^\top\vect{\Sigma}^{-1}\vect{X}|}{\partial\vect{\Sigma}_{r}} =& 
\frac{\partial \log|\vect{X}^\top\vect{\Sigma}^{-1}\vect{X}|}{\partial\vect{\Sigma}_{r}^{-1}}
\frac{\partial \vect{\Sigma}^{-1}}{\partial\vect{\Sigma}_{r}}\label{proof:GenReml}.
\end{align}
To that end we calculate \citep[using][(46)]{petersen_Matrix2006}
\begin{align}
\frac{\partial \log|\vect{X}^\top\vect{\Sigma}^{-1}\vect{X}|}{\partial\sigma^{ij}}  &= Tr\left((\vect{X}^\top\vect{\Sigma}^{-1}\vect{X})^{-1}
\frac{\partial \vect{X}^\top\vect{\Sigma}^{-1}\vect{X}}{\partial\sigma^{ij}}\right).\label{eq:C21}
\end{align}
With $\vect{X}^\top\vect{\Sigma}^{-1}\vect{X}$ as in \eqref{eq:sparse} we get
\begin{align}
    \frac{\partial \vect{X}^\top\vect{\Sigma}^{-1}\vect{X}}{\partial\sigma^{ij}} = 
    \left[
    \begin{matrix}
        \vect{0} &\vect{0} & \vect{0}\\
        \vect{0} &\vect{X}_{i,\cdot}^\top\vect{X}_{j,\cdot} & \vect{0}\\
        \vect{0} &\vect{0} & \vect{0}
    \end{matrix}
    \right].
\end{align}
Therefore, \eqref{eq:C21} can be written as
\begin{align}
  \begin{split}
  \frac{\partial \log|\vect{X}^\top\vect{\Sigma}^{-1}\vect{X}|}{\partial\sigma^{ij}} =&\text{Tr}\left((\vect{X}^\top\vect{\Sigma}^{-1}\vect{X})^{-1}_{I_j,I_i}\vect{X}^\top_{i,\cdot}\vect{X}_{j,\cdot}\right),
  \end{split}
\end{align}
with $I_{i}=\{P_{i-1}+1,P_{i}+2,...,P_{i}+p_i\}$ and  $P_i=\sum_{l=0}^ip_l$, with the convention that $p_0=0$. From that we have 
\begin{align}
  \begin{split}
    \left(\frac{\partial \log|\vect{X}^\top\vect{\Sigma}^{-1}\vect{X}|}{\partial\vect{\Sigma}_{r}}\right)_{ij}&= \text{Tr}\left( \frac{\partial \log|\vect{X}^\top\vect{\Sigma}^{-1}\vect{X}|}{\partial\vect{\Sigma}_{r}^{-1}}\Sigr^{-1}\frac{\partial\Sigr}{\sigma_{ij}} \Sigr^{-1}\right)\\
   &= \text{Tr}\left( \Sigr^{-1}\frac{ \partial \log|\vect{X}^\top\Sigr^{-1}\vect{X}|}{\partial\vect{\Sigma}_{r}^{-1}}\Sigr^{-1}\frac{\partial\Sigr}{\sigma_{ij}} \right).\label{eq:partXX}
    \end{split}
\end{align}
Since $\left(\frac{\partial\Sigr}{\sigma_{ij}}\right)_{kl}=1$ for $k=i,l=j$ and zero otherwise we have  
\begin{align}
  \frac{\partial \log|\vect{X}^\top\vect{\Sigma}^{-1}\vect{X}|}{\partial\vect{\Sigma}_{r}}= \Sigr^{-1}\frac{\partial \log|\vect{X}^\top\vect{\Sigma}^{-1}\vect{X}|}{\partial\vect{\Sigma}_{r}^{-1}}\Sigr^{-1},
\end{align}
and the REML estimate of $\Sigr$ is given as the solution to 
\begin{align}
    -\frac{T}{2}\Sigr^{-1} + \frac{1}{2}\Sigr^{-1}\sum_{t=1}^T\vect{e_t} \vect{e}_t^\top\Sigr^{-1}+\frac{1}{2}\frac{\partial \log|\vect{X}^\top\vect{\Sigma}^{-1}\vect{X}|}{\partial\vect{\Sigma}_{r}}=0,
\end{align}
or
\begin{align}
  \begin{split}
      \vect{\hat{\Sigma}}_{\text{REML},r}^{-1} =&
      \frac{1}{T}\left(
      \Sigr^{-1}\sum_{t=1}^T\vect{e_t} \vect{e}_t^\top\Sigr^{-1}+
      \frac{
        \partial\log|\vect{X}^\top(\vect{I}_T\otimes\Sigr^{-1})\vect{X}|
      }{\partial\Sigr}\right)\Bigg|_{\Sigr=\SigrHatREML}\\
      =&      \frac{1}{T} \left(\Sigr^{-1}\left(
     \sum_{t=1}^T\vect{e_t} \vect{e}_t^\top+
      \frac{
        \partial\log|\vect{X}^\top(\vect{I}_T\otimes\Sigr^{-1})\vect{X}|
      }{\partial\Sigr^{-1}}\right)\Sigr^{-1}\right)\Bigg|_{\Sigr=\SigrHatREML}.
      \end{split}
\end{align}
The solution can be written in terms of $\Sigr$ as
\begin{align}
      \vect{\hat{\Sigma}}_{\text{REML},r} =
      \frac{1}{T}\left(
      \sum_{t=1}^T\vect{e_t} \vect{e}_t^\top+
      \frac{
        \partial\log|\vect{X}^\top(\vect{I}_T\otimes\Sigr^{-1})\vect{X}|
      }{
        \partial\Sigr^{-1}}\right)\Bigg|_{\Sigr=\SigrHatREML}.
\end{align}
This is \eqref{eq:Sigma.rREML} and concludes the proof of Theorem \ref{the:main} and Eq. \eqref{eq:Sigma.rREML}.\qed
\section{Proof of Theorem \ref{the:opt.proj.hier}}\label{proof:opt.proj.hier}
Consider the regression problem 
\begin{align}
  \vect{y}=\vect{X}\vect{\beta}+\vect{\epsilon};\quad \vect{\epsilon}_t\sim N(\vect{0},\vect{I}_T\otimes\Sigr)\label{eq:appGLM},
\end{align}
where $\vect{X}_{\cdot,t}=\vect{I}_m\otimes\vect{\hat{y}}_t^\top$, with the linear constraints $\vect{S}^\top\vectoriz^{-1}(\vect{\beta})=\vect{I}_m$. We use the following notation
\begin{align}
  \vect{\hat{Y}}=\left[\YhatT\quad \YhatB\right];\quad\vect{S}=\left[\begin{matrix}\ST\\ \vect{I}_m\end{matrix}\right];\quad \vect{\beta}^m=\vectoriz^{-1}(\vect{\beta})=\left[\begin{matrix}\betaT^m\\ \betaB^m\end{matrix}\right].
\end{align}
The linear constraints can be written as 
\begin{align}
  \vect{S}^\top\vect{\beta}^m=\left[\begin{matrix}\ST^\top & \vect{I}_m\end{matrix}\right]\left[\begin{matrix}\betaT^m\\ \betaB^m\end{matrix}\right]=\ST^\top\betaT^m+\betaB^m = \vect{I}_m,
\end{align}
i.e.,
\begin{align}
  \betaB^m=\vect{I}_m-\ST^\top\betaT^m.
\end{align}
Using the model definition \eqref{eq:appGLM} we get 
\begin{align}
  \begin{split}
  \vect{y}_t=&(\vect{I}_m\otimes\vect{\hat{y}}^\top_t)\vect{\beta}+\vect{\epsilon}_t\\
  =&(\vect{I}_m\otimes\vect{\hat{y}}_t^\top)\vectoriz\left(\left[\begin{matrix}\betaT^m\\ \betaB^m\end{matrix}\right]\right)+\vect{\epsilon}_t\\
  =&(\vect{I}_m\otimes\vect{\hat{y}}_t^\top)\vectoriz\left(\left[\begin{matrix}\vect{0}\\ \vect{I}_m\end{matrix}\right]+\left[\begin{matrix}\vect{I}_{n-m}\\ -\ST^\top\end{matrix}\right]\betaT^m\right)+\vect{\epsilon}_t.
  \end{split}
\end{align}
Using \eqref{eq:vect1} we get 
\begin{align}
  \begin{split}
  \vect{y}_t=&\vectoriz\left(\vect{\hat{y}}_t^\top\left(\left[\begin{matrix}\vect{0}\\ \vect{I}_m\end{matrix}\right]+\left[\begin{matrix}\vect{I}_{n-m}\\ -\ST^\top\end{matrix}\right]\betaT^m\right)\right)+\vect{\epsilon}_t\\
  =&\vectoriz(\yhatBt^\top)+\vectoriz\left((\yhatTt^\top- \yhatBt^\top\ST^\top)\betaT^m\right)+\vect{\epsilon}_t.
  \end{split}
\end{align}
Rearranging, we get 
\begin{align}
  \begin{split}
  \vect{y}_t-\yhatBt =&\left(\vect{I}_m\otimes(\yhatTt^\top- \yhatBt^\top\ST^\top)\right)\betaT+\vect{\epsilon}_t\\
  =&\vect{\tilde{X}}_{\cdot,t}\betaT+\vect{\epsilon}_t,
  \end{split}
\end{align}
which is the postulated model. Using \eqref{eq:betaHat.i} the parameter estimates can be written as
\begin{align}
  \begin{split}
  \betaT^m=&\left[\begin{matrix} 
      \left(\vect{\tilde{X}}_{1,\cdot}^\top\vect{\tilde{X}}_{1,\cdot}^\top\right)^{-1}\vect{\tilde{X}}_{1,\cdot}^\top\vect{y}_{1,\cdot} & \hdots & \left(\vect{\tilde{X}}_{1,\cdot}^\top\vect{\tilde{X}}_{1,\cdot}^\top\right)^{-1}\vect{\tilde{X}}_{1,\cdot}^\top\vect{y}_{m,\cdot}
  \end{matrix}\right]\\
  =&\left(\vect{\tilde{X}}_{1,\cdot}^\top\vect{\tilde{X}}_{1,\cdot}^\top\right)^{-1}\vect{\tilde{X}}_{1,\cdot}^\top\vect{Y},\label{eq:D8}
  \end{split}
\end{align}
\update{where $\vect{y}_{i,\cdot}\in\mathbb{R}^T$ is the collection of observation at bottom level $i$ at all times.}

By inserting $\vect{\tilde{X}}_{1,\cdot}$ we get
\begin{align}
  \begin{split}
  \betaHatT^m=&\left(\left(\YhatT- \YhatB\ST^\top\right)^\top
  \left(\YhatT- \YhatB\ST^\top\right)\right)^{-1}
  \left(\YhatT- \YhatB\ST\top\right)^\top
 \left(\vect{Y}-\YhatB\right).
  \end{split}
\end{align}

The next step is to show that this is equivalent to $\vect{P}$. Using Theorem 1 of \cite{wickramasuriya2019optimal} the usual weight matrix can be written as
\begin{align}
  \begin{split}
  \vect{P}=&\left(\vect{S}^\top\Sigh^{-1}\vect{S}\right)^{-1} \vect{S}^\top\Sigh^{-1}\\
  =& \vect{J}-\vect{J}\Sigh\vect{U}
  \left(\vect{U}^\top\Sigh\vect{U}\right)^{-1} \vect{U}^\top,\label{eq:P}
  \end{split}
\end{align}
where $\vect{J}=\left[\vect{0}_{m\times n-m}\quad \vect{I}_m\right]$, and $\vect{U}=  \left[\begin{matrix}\vect{I}_{n-m}\\ -\ST^\top \end{matrix}\right]$.

Replacing $\Sigh$ by $\SighHat=\frac{1}{T}\left(\vect{YS}^\top-\vect{\hat{Y}}\right)^\top \left(\vect{YS}^\top-\vect{\hat{Y}}\right)$, we calculate each term of \eqref{eq:P}
\begin{align}
  \begin{split}
  \vect{U}^\top\left(\vect{YS}^\top-\vect{\hat{Y}}\right)^\top=&
  \left[\begin{matrix}\vect{I}_{n-m} & -\ST \end{matrix}\right]\left(\vect{S}\vect{Y}^\top-\vect{\hat{Y}}^\top\right) \\
  =& \left[\begin{matrix}\vect{I}_{n-m} & -\ST \end{matrix}\right]\left(\left[\begin{matrix}\ST\\\vect{I}_{m}\end{matrix}\right]\vect{Y}^\top-\vect{\hat{Y}}^\top\right)\\
  =&-\left[\begin{matrix}\vect{I}_{n-m} & -\ST \end{matrix}\right]\vect{\hat{Y}}^\top\\
  =&-\left(\YhatT^\top-\ST\YhatB^\top\right).\label{eq:D.11}
  \end{split}
\end{align}
Further, 
\begin{align}
  \begin{split}
  \vect{J}(\vect{YS}^\top-\vect{\hat{Y}})^\top=&
  \left[\vect{0}_{m\times n-m}\quad \vect{I}_m\right] \left(\left[\begin{matrix}\ST\\\vect{I}_{m}\end{matrix}\right]\vect{Y}-\vect{\hat{Y}}^\top\right)\\
 =& \vect{Y}^\top-\YhatB^\top\label{eq:D.12}
 \end{split}
\end{align}
and 
\begin{align}
  \begin{split}
  \vect{P}=&\vect{J}+
  \left(\vect{Y}-\YhatB\right)^\top \left(\YhatT-\YhatB\ST^\top\right)
  \left(\left(\YhatT-\YhatB\ST^\top\right)^\top
  \left(\YhatT-\YhatB\ST^\top\right)\right)^{-1}
  \vect{U}^\top\\
  =&\vect{J}+\left(\betaHatT^m\right)^\top\vect{U}^\top.
  \end{split}
  \end{align}
Since 
\begin{align}
  \left(\vect{\hat{\beta}}^m\right)^\top=&\vect{J}+\left(\betaHatT^m\right)^\top\vect{U}^\top,
\end{align}
we can conclude that if $\Sigh=\frac{1}{T}\left(\vect{YS}^\top-\vect{\hat{Y}}\right)^\top \left(\vect{YS}^\top-\vect{\hat{Y}}\right)$ then 
\begin{align}
  \vect{P}=\left(\vect{\hat{\beta}}^m\right)^\top,
\end{align}
where $\vect{\hat{\beta}}^m$ is calculated using usual linear regression with linear constraint ($\vect{S}^\top\vect{\hat{\beta}}^m=\vect{I}$). This proves the equivalence between \eqref{eq:18} and \eqref{eq:ConstrReg}. The ML and REML estimators in \eqref{eq:18a} follow from \eqref{eq:Sigma.rML} and \eqref{eq:Sig.rSimp} in Theorem \ref{the:main}. Further that $\betaHatT$ is independent of $\Sigr$ also follow from Theorem \ref{the:main}.\qed 
\section{Proof of Corollary \ref{remark:reconvar}}\label{proof:reconvar}
We need to show that
\begin{align}
  \begin{split}
     \frac{1}{T}\vect{P}\left(\vect{Y}\vect{S}^\top-\vect{\hat{Y}}\right)^\top \left(\vect{Y}\vect{S}^\top-\vect{\hat{Y}}\right) \vect{P}^\top=&
\frac{1}{T} \left(  \vect{Y}-\YhatB-\left(\YT-\YhatB
  \ST^\top\right)\betaHatT^m\right)^\top\times\\ & \left( \vect{Y}-\YhatB-\left(\YT-\YhatB
  \ST^\top\right)\betaHatT^m\right).
  \end{split}
\end{align}
It suffices to show that
\begin{align}
   \left(\vect{Y}\vect{S}^\top-\vect{\hat{Y}}\right) \vect{P}^\top=
  \vect{Y}-\YhatB-\left(\YT-\YhatB
  \ST^\top\right)\betaHatT^m\label{eq:reconvar}.
\end{align}
Using Theorem \ref{the:opt.proj.hier}, the left-hand side can be written as 
\begin{align}
  \begin{split}
  (\vect{Y}\vect{S}^\top-\vect{\hat{Y}}) \vect{P}^\top =&
  \left[ \vect{Y}\ST^\top-\YhatT\quad \vect{Y}-\YhatB\right]
    \left[\begin{matrix}
    \betaHatT^m \\ \betaHatB^m\end{matrix}\right]\\
     =&
  \left[ \vect{Y}\ST^\top-\YhatT\quad \vect{Y}-\YhatB\right]
    \left[\begin{matrix}
    \betaHatT^m \\ \vect{I}-\ST^\top\betaHatT^m\end{matrix}\right]\\
    =&\left(\vect{Y}\ST^\top-\YhatT - 
    (\vect{Y}-\YhatB)\ST^\top\right)\betaHatT^m + 
    \vect{Y}-\YhatB\\
        =&\vect{Y}-\YhatB-\left(\YhatT -
    \YhatB\ST^\top\right)\betaHatT^m,
    \end{split}
\end{align}
which completes the proof.\qed
\section{Proof of Lemma \ref{the:proj}}\label{proof:proj}
We assume that $\vect{X}_{1,\cdot}$ has full rank (i.e., $\texttt{rank}(\vect{X}_{1,\cdot})=n-m$).
In that case  the reconciled forecast is 
\begin{align}
\begin{split}
  \vect{\tilde{Y}}=&\YhatB+\vect{X}_{1,\cdot}\betaHatT^m\\
  =&\YhatB+\vect{X}_{1,\cdot}(\vect{X}_{1,\cdot}^\top\vect{X}_{1,\cdot})^{-1}\vect{X}_{1,\cdot}^\top(\vect{Y}-\YhatB)\\
    =&\YhatB+\vect{H}(\vect{Y}-\YhatB),\label{eq:H2}
\end{split}
\end{align}
with $\vect{H}=\vect{X}_{1,\cdot}(\vect{X}_{1,\cdot}^\top\vect{X}_{1,\cdot})^{-1}\vect{X}_{1,\cdot}^\top$. Both $\vect{H}$ and $\vect{I}-\vect{H}$ are projection matrices,
and 
\begin{align}
    \begin{split}
\vect{\tilde{Y}}-\YhatB=&
\vect{H}(\vect{Y}-\YhatB)\\
\vect{Y}-\vect{\tilde{Y}}=&\vect{Y}-\vect{H}(\vect{Y}-\YhatB)-\YhatB\\
=&(\vect{I}-\vect{H})(\vect{Y}-\YhatB)\label{eq:H3}.
    \end{split}
\end{align}

For the top levels note that $\vect{X}_{1,\cdot}=\YhatT-\YhatB\ST^\top$, and 
\begin{align}
  \begin{split}
    \vect{\tilde{Y}}\ST^\top-\YhatT=&  
    \YhatB\ST^\top+\vect{H}(\vect{Y}-\YhatB)\ST^\top-\YhatT\\
    =& \vect{H}(\vect{Y}-\YhatB)\ST^\top-\vect{X}_{1,\cdot}.
    \end{split}
\end{align}
Using that  $\vect{X}_{1,\cdot}=\vect{H}\vect{X}_{1,\cdot}$, we can write
\begin{align}
  \begin{split}
    \vect{\tilde{Y}}\ST^\top-\YhatT
    =& \vect{H}(\vect{Y}\ST^\top-\YhatB\ST^\top-\vect{X}_{1,\cdot})\\
    =& \vect{H}(\vect{Y}\ST^\top-\YhatB\ST^\top-\YhatT+\YhatB\ST^\top)\\
    =& \vect{H}(\vect{Y}\ST^\top-\YhatT).
    \end{split}
\end{align}
As an immediate consequence we can write 
\begin{align}
  \begin{split}
    \vect{Y}\ST^\top-\vect{\tilde{Y}}\ST^\top=&
    \vect{Y}\ST^\top-\vect{H}(\vect{Y}\ST^\top-\YhatT)-\YhatT\\
    =&(\vect{I}-\vect{H})(\vect{Y}\ST^\top-\YhatT).
    \end{split}
\end{align}
Hence,
\begin{align}
  \begin{split}
    \vect{\tilde{Y}}\vect{S}^\top-\vect{\hat{Y}}=&\left[\vect{H}(\vect{Y}\ST^\top-\YhatT)
      \quad
    \vect{H}(\vect{Y}-\YhatB)\right]\\
    =&
    \vect{H}(\vect{Y}\vect{S}^\top-\vect{\hat{Y}})\\
    (\vect{Y}-\vect{\tilde{Y}})\vect{S}^\top=&
    \left[\vect{Y}\ST^\top-\vect{\tilde{Y}}\ST^\top\quad \vect{Y}-\vect{\tilde{Y}}\right]\\ 
    =& 
    \left[(\vect{I}-\vect{H})(\vect{Y}\ST^\top-\YhatT)
      \quad
    (\vect{I}-\vect{H})(\vect{Y}-\YhatB)\right]\\
    =&    (\vect{I}-\vect{H})(\vect{Y}\vect{S}^\top-\vect{\hat{Y}})\label{eq:projAll},
   \end{split}
\end{align}
and
\begin{align}
(\vect{Y}\vect{S}^\top-\vect{\hat{Y}})= (\vect{I}-\vect{H})(\vect{Y}\vect{S}^\top-\vect{\hat{Y}}) + \vect{H}(\vect{Y}\vect{S}^\top-\vect{\hat{Y}}).\label{eq:proofProj}
\end{align}
Orthogonality follows from $\vect{H}=\vect{H}^\top$ and $\vect{H}^2=\vect{H}$. This is true elementwise as we can choose an arbitrary index set $I$ and get 
\begin{align}
  \begin{split}
    \left( \left(\vect{Y}\vect{S}^\top-\vect{\hat{Y}}\right)^\top\left(\vect{Y}\vect{S}^\top-\vect{\hat{Y}}\right)\right)_{I,I}=\left(\vect{Y}\vect{S}_{I}^\top-\vect{\hat{Y}}_{I}\right)^\top\left(\vect{Y}\vect{S}_{I}^\top-\vect{\hat{Y}}_{I}\right),
  \end{split}
\end{align}
and similarly for the terms on the right-hand side of \eqref{eq:proofProj}. This concludes the proof of Lemma  \ref{the:proj}.\qed
\section{Proof of Theorem \ref{col:proj}}\label{proof:colProj}
Use \eqref{eq:projAll} and \eqref{eq:vect1} to get
\begin{align}
  \begin{split}
    \vect{\tilde{y}}_{I}-\vect{\hat{y}}_I=&\vectoriz(\vect{S}_I\vect{\tilde{Y}}^\top-\vect{\hat{Y}}_I^\top)\\
    =& \vectoriz((\vect{S}_I\vect{Y}^\top-\vect{\hat{Y}}^\top_I)\vect{H})\\
    =&(\vect{H}\otimes\vect{I}_{q_I})\vectoriz(\vect{S}_I\vect{Y}^\top-\vect{\hat{Y}}^\top_I)\\
    =&(\vect{H}\otimes\vect{I}_{q_I})(\vect{y}_{I}-\vect{\hat{y}}_I)\\
    \vect{y}_{I}-\vect{\tilde{y}}_I=&\vectoriz(\vect{S}_I\vect{Y}^\top-\vect{S}_I\vect{\tilde{Y}}_I^\top)\\
    =&\vectoriz((\vect{S}_I\vect{Y}^\top-\vect{\hat{Y}}_I^\top)(\vect{I}-\vect{H}))\\
    =& ((\vect{I}-\vect{H})\otimes \vect{I}_{q_I})\vectoriz(\vect{S}_I\vect{Y}^\top-\vect{\hat{Y}}_I^\top)\\
    =& (\vect{I}_{Tq_I}-\vect{H}\otimes \vect{I}_{q_I})(\vect{y}_{I}-\vect{\hat{y}}_I).
  \end{split}
\end{align}
With $\vect{H}_I=\vect{H}\otimes\vect{I}_{q_I}$ we get
\begin{align}
  \begin{split}
    \vect{y}_{I}-\vect{\hat{y}}_I=&\vect{H}_I(\vect{y}_{I}-\vect{\hat{y}}_I)+ (\vect{I}_{Tq_I}-\vect{H}_I)(\vect{y}_{I}-\vect{\hat{y}}_I).
  \end{split}
\end{align}
Since $\vect{H}$ is a projection matrix we get $\vect{H}_I^2=(\vect{H}\otimes\vect{I}_{q_I})(\vect{H}\otimes\vect{I}_{q_I})=\vect{H}^2\otimes\vect{I}_{q_I}^2=\vect{H}\otimes\vect{I}_{q_I}$ and 
$\vect{H}_I^\top=(\vect{H}\otimes\vect{I}_{q_I})^\top=\vect{H}^\top\otimes\vect{I}_{q_I}^\top=\vect{H}\otimes\vect{I}_{q_I}=\vect{H}_I$. Orhtogonality follows from $(\vect{I}_{Tq_I}-\vect{H}_I)\vect{H}_I=\vect{0}$. Thus, we have 
\begin{align}
 \begin{split} ||\vect{y}_I-\vect{\hat{y}}_I||^2=&(\vect{y}_I-\vect{\hat{y}}_I)^\top\vect{H}_I(\vect{y}_I-\vect{\hat{y}}_I)+(\vect{y}_I-\vect{\hat{y}}_I)^\top(\vect{I}_{Tq_I}-\vect{H}_I)(\vect{y}_I-\vect{\hat{y}}_I)\\
  =&||\vect{y}_I-\vect{\tilde{y}}_I||^2+||\vect{\tilde{y}}_I-\vect{\hat{y}}_I||^2,
  \end{split}
  \end{align}
which is \eqref{eq:proj2}, and \eqref{eq:DistReduc1} follows directly. This concludes the proof of Theorem \ref{col:proj}.\qed
\section{Proof of Corollary \ref{col:CentralSig}}\label{proof:CentralSig}

We consider the model 
\begin{align}
  \vect{y}_t-\yhatBt=\left[\vect{I}_m\otimes (\yhatTt^\top-\yhatBt^\top\ST^\top)\right]\betaT+\vect{\epsilon}_t;\quad \vect{\epsilon}_t\sim N(\vect{0},\Sigr).\label{eq:ConstrRegApp}
\end{align}
Multiplying with a fixed vector $\vect{v}$ and setting $\vect{X}_{1,t}=\yhatTt^\top-\yhatBt^\top\ST^\top$, we get 
\begin{align}
  \vect{v}^\top(\vect{y}_t-\yhatBt)=[v_1\vect{X}_{1,t}\quad v_2\vect{X}_{1,t}\quad ...\quad v_{m}\vect{X}_{1,t}]\betaT+\vect{\epsilon}_t;\quad \vect{v}^\top\vect{\epsilon}_t\sim N(\vect{0},\Sigr).
\end{align}
Setting $\vect{\gamma}=\sum_{i=1}^mv_i\betaTi$, $z_t=\vect{v}^\top(\vect{y}_t-\yhatBt)$, $u_t=\vect{v}^\top\vect{\epsilon}_t$, and $\sigma_v^2=\vect{v}^\top\Sigr\vect{v}$, we get 
\begin{align}
z_t=\vect{X}_{1,t}\vect{\gamma}+u_t;\quad u_t\sim N(0,\sigma_v^2). 
\end{align}
Written in matrix-vector notation this is 
\begin{align}
\vect{z}=\vect{X}_{1,\cdot}\vect{\gamma}+\vect{u};\quad \vect{u}\sim N(0,\sigma_v^2\vect{I}),  
\end{align}
which is the usual general linear model for which variance decomposition and Cochran's Theorem imply the $\chi^2$-distribution \citep[see e.g.][]{madsen_thyregod_2011}. For completeness we write out the chain of projection. The model can be written as
\begin{align}
\begin{split}
\vect{u}=&\vect{z}-\vect{X}_{1,\cdot}\vect{\gamma}\\
=&\vect{z}-\vect{X}_{1,\cdot}\vect{\hat{\gamma}}+\vect{X}_{1,\cdot}\vect{\hat{\gamma}}-\vect{X}_{1,\cdot}\vect{\gamma}\\
=&(\vect{I}-\vect{H})\vect{z}+\vect{H}\vect{z}-\vect{X}_{1,\cdot}\vect{\gamma}.
\end{split}
\end{align}
By definition $\frac{\vect{u}^\top\vect{u}}{\sigma^2_v}\sim\chi^2(T)$. For the second term, centrality of the the estimator $\vect{\hat{\gamma}}$ implies that $E[(\vect{I}-\vect{H})\vect{z}]=\vect{0}$. Since $(\vect{I}-\vect{H})$ is a projection matrix there exists an orthogonal basis that spans the space of $(\vect{I}-\vect{H})\vect{z}$ and, hence, by Lemma \ref{lemma:Whishart},  $\frac{\vect{z}^\top(\vect{I}-\vect{H})\vect{z}}{\sigma_v^2}\sim \chi^2(T-n+m)$.  
Additionally, if $\vect{\gamma}=\vect{0}$,  then $\frac{\vect{z}^\top\vect{H}\vect{z}}{\sigma_v^2}\sim\chi^2(n-m)$. Since $\vect{v}$ is arbitrary, Lemma \ref{lemma:Whishart} implies that  
    \begin{align}
    \left(\vect{Y}-\vect{\tilde{Y}}\right)^\top\left(\vect{Y}-\vect{\tilde{Y}}\right)\sim W_m(\Sigr,T-(n-m)).
    \end{align}
This holds under the model assumptions (i.e., also when $\betaT\neq\vect{0}$). If, in addition, $\vect{\gamma}=\vect{0}$ for all $\vect{v}$, implying that $\betaT=\vect{0}$, then 
    \begin{align}
    \left(\vect{\tilde{Y}}\boldsymbol{S}^{\top}-\vect{\hat{Y}}\right)^\top\left(\vect{\tilde{Y}}\boldsymbol{S}^{\top}-\vect{\hat{Y}}\right)\sim W_m(\boldsymbol{S}\Sigr\boldsymbol{S}^{\top},n-m).
\end{align}
which concludes the proof of Corollary \ref{col:CentralSig}.\qed 

\section{Proof of Theorem \ref{the:shrink}}\label{proof:shrink}

The usual shrinkage estimator for the variance--covariance matrix is 
\begin{align}
  \begin{split}
  \SigsHat  =& (1-\lambda)\SighHat+\lambda\SighHatd,
  \end{split}
\end{align}
with $\SighHat=\frac{1}{T}\left(\vect{Y}\vect{S}^\top-\vect{\hat{Y}}\right)^\top\left(\vect{Y}\vect{S}^\top-\vect{\hat{Y}}\right)$. Using Theorem 1 of \cite{wickramasuriya2019optimal} and \eqref{eq:D.11}--\eqref{eq:D.12} we can write 
\begin{align}
  \vect{J}\SigsHat\vect{U}=& -\frac{1-\lambda}{T}\left(\vect{Y}-\YhatB\right)^\top \left(\vect{\hat{Y}}-\YhatB\ST^\top\right) + \lambda\vect{J}\SighHat^{\text{d}}\vect{U},
\end{align}
where 
\begin{align}
  \begin{split}
  \vect{J}\SighHat^{\text{d}}\vect{U} = &
  \left[\vect{0}_{m\times n-m}\quad \vect{I}_m\right] \left[\begin{matrix}\SighHatdT &\vect{0} \\
      \vect{0} & \SighHatdB
    \end{matrix}\right]
   \left[\begin{matrix}\vect{I}_{n-m}\\ -\ST^\top \end{matrix}\right]\\
   =&-\SighHatdB\ST^\top.
   \end{split}
\end{align}

Also, using \eqref{eq:D.11}, we can write
\begin{align}
  \vect{U}^\top\SigsHat\vect{U}=&
  \frac{1-\lambda}{T} \left(\vect{Y}-\YhatB\right)^\top \left(\vect{Y}-\YhatB\right)+ \lambda\vect{U}^\top\SighHatd\vect{U},
\end{align}
and
\begin{align}
  \begin{split}
  \vect{U}^\top\SighHatd\vect{U}=&
  \left[\begin{matrix}\vect{I}_{n-m}& -\ST \end{matrix}\right]
  \left[\begin{matrix}\SighHatdT &\vect{0} \\
      \vect{0} & \SighHatdB
    \end{matrix}\right]
   \left[\begin{matrix}\vect{I}_{n-m}\\ -\ST^\top \end{matrix}\right]\\
   =&\SighHatdT+\ST\SighHatdB \ST^\top.
   \end{split}
\end{align}
Hence,
\begin{align}
  \begin{split}
  \vect{P}(\lambda)
  =& \vect{J}-\vect{J}\SigsHat\vect{U}
  \left(\vect{U}^\top\SigsHat\vect{U}\right)^{-1} \vect{U}^\top\\
  =&\vect{J}+ \left(\frac{1-\lambda}{T} \left(\vect{Y}-\YhatB\right)^\top\left(\YhatT-\YhatB\ST^\top\right)+\lambda\SighHatdB\ST^\top\right)\times\\
  &\left(\frac{1-\lambda}{T} \left(\YhatT-\YhatB\ST^\top\right)^\top \left(\YhatT-\YhatB\ST^\top\right)+ 
  \lambda\left(\SighHatdT+\ST\SighHatdB \ST^\top\right)\right)^{-1}\vect{U}^\top.\label{eq:F6}
  \end{split}
\end{align}
Now consider the model 
\begin{align}
  \vect{y}_t-\yhatBt=\left[\vect{I}_m\otimes\yhatTt^\top- \vect{I}_m\otimes\left(\yhatBt^\top\ST^\top\right)\right]\betaT+ \vect{\epsilon}_t,
\end{align}
or 
\begin{align}
    \vect{y}-\yhatB=\vect{X}\betaT+ \vect{\epsilon}\label{eq:AppGlm},
  \end{align}
with $\vect{\epsilon}\sim N(\vect{0},\vect{I}_T\otimes\Sigr)$ and $\betaT\sim N_{m(n-m)}(\betaOT,\vect{\Sigma}_{\beta})$.
Its log-posterior density is 
\begin{align}
  \begin{split}
  l_{\text{MAP}}(\betaT,\Sigr)\propto & -\frac{1}{2}\log|\vect{\Sigma}| -
  \frac{1}{2}\left(\vect{y}-\YhatB-\vect{X}\betaT\right)^\top\vect{\Sigma}^{-1} \left(\vect{y}-\YhatB-\vect{X}\betaT\right) -\\
&  \frac{1}{2}\left(\betaT-\vect{\beta}_{0T}\right)^\top\vect{\Sigma}^{-1}_{\beta}\left(\betaT-\betaOT\right)-\frac{1}{2}\log|\vect{\Sigma}_{\beta}|,\label{eq:F9}
  \end{split}
\end{align}
and, consequently, the MAP estimate of $\betaT$ is 
\begin{align}
  \begin{split}
  \betaHatT=&\left(\vect{X}^\top(\vect{I}_T\otimes\Sigr)^{-1}\vect{X} + 
  \vect{\Sigma}_{\beta}^{-1}\right)^{-1}
  (\vect{X}^\top (\vect{I}_T\otimes\Sigr)^{-1}(\vect{y}-\yhatB)+\vect{\Sigma}_{\beta}^{-1}\betaOT).\label{eq:F10}
  \end{split}
\end{align}
Choosing $\vect{\Sigma}_{\beta}=\Sigr\otimes \vect{\Sigma}_{0,\beta}$, we go through each term. Since $\vect{x}_{t,i}=\vect{x}_{t,j}$, using \eqref{eq:C.7}, we get
\begin{align}
  \begin{split}
\left(\vect{X}^\top(\vect{I}_T\otimes\Sigr)^{-1}\vect{X} + 
  \vect{\Sigma}_{\beta}^{-1}\right)^{-1}=& 
  \Sigr\otimes\left(\vect{X}_{1,\cdot}^\top\vect{X}_{1,\cdot} + 
  \vect{\Sigma}_{0,\beta}^{-1}\right)^{-1},
  \end{split}
  \end{align}
  and using, \eqref{eq:C.9}
  \begin{align}
  \begin{split}
    (\vect{X}^\top (\vect{I}_T\otimes\Sigr)^{-1}(\vect{y}-\yhatB)=&
    \left(\vect{X}^\top\vectoriz(\Sigr^{-1} \left(\vect{Y}-\YhatB\right)^\top\right)\\
  =&\left(\vect{I}_m\otimes\vect{X}_{1,\cdot}^\top\right)\vectoriz\left( (\vect{Y}-\YhatB)\Sigr^{-1}\right)\\
    =&\vectoriz\left(\vect{X}_{1,\cdot}^\top (\vect{Y}-\YhatB)\Sigr^{-1}\right).
\end{split}
\end{align}
It follows from the definition of the prior and \eqref{eq:vect2} that
\begin{align}
\begin{split}
    \vect{\Sigma}_{\beta}^{-1}\betaOT=&(\vect{\Sigma}_{r}^{-1}\otimes \vect{\Sigma}_{0,\beta}^{-1})\betaOT.
  \end{split}
\end{align}
Hence \eqref{eq:F10}, can be written as
\begin{align}
    \betaHatT=\left(\Sigr\otimes\left(\vect{X}_{1,\cdot}^\top\vect{X}_{1,\cdot} + 
  \vect{\Sigma}_{0,\beta}^{-1}\right)^{-1}\right)\left(\vectoriz\left(\vect{X}_{1,\cdot}^\top (\vect{Y}-\YhatB)\Sigr^{-1}\right)+(\vect{\Sigma}_{r}^{-1}\otimes \vect{\Sigma}_{0,\beta}^{-1})\betaOT\right).\label{eq:I14}
\end{align}
Multiplying terms (using $\vect{C}=\left(\vect{X}_{1,\cdot}^\top\vect{X}_{1,\cdot} + 
  \vect{\Sigma}_{0,\beta}^{-1}\right)^{-1}$) we get
\begin{align}
  \begin{split}
  \left(\Sigr\otimes \vect{C}\right) \vectoriz\left(\vect{X}_{1,\cdot}^\top (\vect{Y}-\YhatB)\Sigr^{-1}\right)
  =&  \vectoriz\left(\vect{C}\vect{X}_{1,\cdot}^\top (\vect{Y}-\YhatB)\right)\label{eq:F14}
    \end{split}
  \end{align}
and 
\begin{align}
  \begin{split}
 \left( \Sigr\otimes\vect{C}\right)    (\vect{\Sigma}_{r}^{-1}\otimes \vect{\Sigma}_{0,\beta}^{-1})\betaOT
 =&\left(\vect{I}_{m}\otimes
  \left(\vect{C}
  \vect{\Sigma}_{0,\beta}^{-1}\right)\right)\betaOT\\
  =&\vectoriz \left(\vect{C}
  \vect{\Sigma}_{0,\beta}^{-1}\betaOT^m\right).\label{eq:F15}
  \end{split}
 \end{align}
Inserting \eqref{eq:F14}--\eqref{eq:F15} in \eqref{eq:I14} we get    
\begin{align}
    \betaHatT
  =&\vectoriz\left( \left(\vect{X}_{1,\cdot}^\top\vect{X}_{1,\cdot} + 
  \vect{\Sigma}_{0,\beta}^{-1}\right)^{-1}
  \left(\vect{X}_{1,\cdot}^\top (\vect{Y}-\YhatB)+
  \vect{\Sigma}_{0,\beta}^{-1}\betaOT^m\right)\right).
\end{align}

With $\vect{\Sigma}_{0,\beta}=
\frac{1-\lambda}{\lambda T}(\SighHatdT+\ST\SighHatdB\ST^\top)^{-1}$ and $\vect{X}_{1,\cdot}=\YhatT-\YhatB\ST^\top$, we get 
\begin{align}
    \betaHatT=&  \vectoriz\left[\left(\left(\YhatT-\YhatB\ST^\top\right)^\top\left(\YhatT-\YhatB\ST^\top\right) +   \frac{\lambda T }{1-\lambda}\left(\SighHatdT+ \ST\SighHatdB\ST^\top\right)\right)^{-1}\times\right.\nonumber\\
&\left.\left(  \left(\YhatT-\YhatB\ST^\top\right)^\top\left(\vect{Y}-\YhatB\right) +
  \frac{\lambda T}{1-\lambda}\left(\SighHatdT+\ST\SighHatdB\ST^\top\right)\left(  \betaOT^m\right)^\top\right)\right]\nonumber\\
    =&  \vectoriz\left[\left((1-\lambda)\left(\YhatT-\YhatB\ST^\top\right)^\top\left(\YhatT-\YhatB\ST^\top\right) +   \lambda T \left(\SighHatdT+\ST\SighHatdB\ST^\top\right)\right)^{-1}\times\right.\nonumber\\
&\left.\left( (1-\lambda) \left(\YhatT-\YhatB\ST^\top\right)^\top\left(\vect{Y}-\YhatB\right) +  \lambda T \left(\SighHatdT+\ST\SighHatdB\ST^\top\right)\left(  \betaOT^m\right)^\top\right)\right].
\end{align}
Hence, if we choose $\betaOT^m= \SighHatdB\ST^\top\left(\SighHatdT+ \ST\SighHatdB\ST^\top\right)^{-1} $, we get
\begin{align}
    \betaHatT^m=& \left((1-\lambda)\left(\YhatT-\YhatB\ST^\top\right)^\top\left(\YhatT-\YhatB\ST^\top\right) +   \lambda T \left(\SighHatdT+\ST\SighHatdB\ST^\top\right)\right)^{-1}\times\nonumber\\
&\left( (1-\lambda) (\YhatT-\YhatB\ST^\top)^\top(\vect{Y}-\YhatB) +
  \lambda T \ST\SighHatdB\right).
\end{align}
Finally, using the linear constraints $\betaB^m = \vect{I}_m-\ST^\top\betaT^m$ and \eqref{eq:F6}, we get
\begin{align}
  \begin{split}
  \left[\begin{matrix}
      \betaHatT^m\\
      \betaHatB^m
    \end{matrix}\right]=&
\left[ \begin{matrix}
  \vect{0}\\
  \vect{I}
  \end{matrix}\right]+
\left[ \begin{matrix}
  \vect{I}\\
  -\ST^\top
  \end{matrix}\right]      \betaHatT^m\\
=&\vect{P}^\top(\lambda).
\end{split}
\end{align}
This is \eqref{eq:35} and concludes the proof of Theorem \ref{the:shrink}.\qed 

\section{Proof of Corollary \ref{remark:MapSig}}\label{collo:MapSig}
The log-posterior density wrt.~$\Sigr$ is
\begin{align}
  \begin{split}
  l_{\text{MAP}}(\Sigr)\propto & -\frac{T}{2}\log|\Sigr| -
  \frac{1}{2}\left(\vect{y}-\YhatB-\vect{X}\betaT\right)^\top\left(\vect{I}_T\otimes\Sigr^{-1} \right)
  \left(\vect{y}-\YhatB-\vect{X}\betaT\right) -\\
 & \frac{1}{2}\left(\betaT-\betaOT\right)^\top\left(\vect{\Sigma}^{-1}_{r}\otimes\vect{\Sigma}^{-1}_{0,\beta}\right)\left(\betaT-\betaOT\right)-
  \frac{n-m}{2}\log|\vect{\Sigma}_{r}|.\label{eq:MAP_sig.r}
  \end{split}
\end{align}
The first two terms and the last term are treated in Appendix \ref{proof:main}. The only term in \eqref{eq:MAP_sig.r} that we have not treated in the previous is the derivative of the third term. Using \eqref{eq:MatrixNorm} we get 
\begin{align}
\begin{split}
\left(\betaT-\betaOT\right)^\top\left(\vect{\Sigma}^{-1}_{r}\otimes\vect{\Sigma}^{-1}_{0,\beta}\right)\left(\betaT-\betaOT\right)=&
\text{Tr}\left(\left(\betaT^m-\betaOT^m\right)^\top\vect{\Sigma}^{-1}_{0,\beta} \left(\betaT^m-\betaOT^m\right)\vect{\Sigma}^{-1}_{r}\right)\\
=&\text{Tr}\left(\left(\betaT^m-\betaOT^m\right)\vect{\Sigma}^{-1}_{r}\left(\betaT^m-\betaOT^m\right)^\top\vect{\Sigma}^{-1}_{0,\beta}\right),
\end{split}
\end{align}
and using \eqref{eq:derivInvTrace} we get 
\begin{align}
\frac{\partial \left(\betaT-\betaOT\right)^\top\left(\vect{\Sigma}^{-1}_{r}\otimes\vect{\Sigma}^{-1}_{0,\beta}\right)\left(\betaT-\betaOT\right)}{\partial\Sigr}=&
-\vect{\Sigma}^{-1}_{r} \left(\betaT^m-\betaOT^m\right)^\top \vect{\Sigma}^{-1}_{0,\beta}\left(\betaT^m-\betaOT^m\right) \vect{\Sigma}^{-1}_{r}.
\end{align}

The derivative of $l_{\text{MAP}}$ wrt.~$\Sigr$ is therefore
 \begin{align}
   \begin{split}
  \frac{\partial l_{\text{MAP}}(\Sigr)}{\partial \Sigr}=& -\frac{T}{2}\Sigr^{-1} +
   \frac{1}{2}\Sigr^{-1}\sum_{t=1}^T\vect{e}_t \vect{e}_t^\top\Sigr^{-1} +\\&   \frac{1}{2}\vect{\Sigma}^{-1}_{r}\left(\betaHatT^m-\betaOT^m\right)^\top\vect{\Sigma}^{-1}_{0,\beta}\left(\betaHatT^m-\betaOT^m\right)\vect{\Sigma}^{-1}_{r}-
  \frac{n-m}{2}\vect{\Sigma}_{r}^{-1}
  \end{split}
\end{align}
and, hence, the MAP estimate of $\Sigr$ is 
\begin{align}
  \begin{split}
  \SigrHat=&\frac{1}{T+n-m}\left(\sum_{t=1}^T\vect{e}_t \vect{e}_t^\top +  \left(\betaHatT^m-\betaOT^m\right)^\top\vect{\Sigma}^{-1}_{0,\beta}\left(\betaHatT^m-\betaOT^m\right)\right)\\
 =&  \frac{1}{T+n-m}\left(\sum_{t=1}^T\vect{e}_t \vect{e}_t^\top + \right.\\ & \left.
 \frac{\lambda T}{1-\lambda} \left(\betaHatT^m-\betaOT^m\right)^\top\left(\SighHatdT+\ST\SighHatdB\ST^\top\right)\left(\betaHatT^m-\betaOT^m\right)\right)\\
  =&  \frac{T/(1-\lambda)}{T+n-m}\left(\frac{1-\lambda}{T}\sum_{t=1}^T\vect{e}_t \vect{e}_t^\top + \right.\\ & \left.
 \lambda  \left(\betaHatT^m-\betaOT^m\right)^\top\left(\SighHatdT+\ST\SighHatdB\ST^\top\right)\left(\betaHatT^m-\betaOT^m\right)\right).
 \end{split}
\end{align}
Using (the proof of) Corollary \ref{remark:reconvar} (Appendix \ref{proof:reconvar}) and inserting $\betaOT$ we can rewrite this as
\begin{align}
  \begin{split}
  \SigrHat=&
   \frac{T/(1-\lambda)}{T+n-m}\left((1-\lambda)\vect{P}(\lambda)\SighHat\vect{P}^\top(\lambda) + 
 \lambda \left(\left(\betaHatT^m\right)^\top\left(\SighHatdT+\ST\SighHatdB\ST^\top\right)\betaHatT^m -\right.\right. \\ &\left.\left.
\left(\betaHatT^m\right)^\top\ST\SighHatdB-
\SighHatdB\ST^\top\betaHatT^m+
\SighHatdB\ST^\top\left(\SighHatdT+ \ST\SighHatdB\ST^\top\right)^{-1}\ST\SighHatdB\right)\right).
\end{split}
\end{align}
With $\SigsHat=(1-\lambda)\SighHat+\lambda\SighHatd$, we have $(1-\lambda)\vect{P}(\lambda)\SighHat\vect{P}^\top(\lambda)=\vect{P}(\lambda)\SigsHat\vect{P}^\top(\lambda)-\lambda\vect{P}(\lambda)\SighHatd\vect{P}^\top(\lambda)$, where
\begin{align}
  \begin{split}
  \vect{P}(\lambda)\vect{\Sigma}^d_h\vect{P}^\top(\lambda) =&
  \left[\begin{matrix} \betaHatT^m \\
      \vect{I}- \ST^\top\betaHatT^m
      \end{matrix}
      \right]^\top
  \left[\begin{matrix}
      \vect{\Sigma}^d_{h,T} & \vect{0}\\
      \vect{0} & \SighHatdB 
      \end{matrix}
      \right]
    \left[\begin{matrix} \betaHatT^m \\
      \vect{I}- \ST^\top\betaHatT^m
      \end{matrix}
      \right]\\
    =& \left(\betaHatT^m\right)^\top\vect{\Sigma}^d_{h,T}\betaHatT^m+ 
      \left(\vect{I}- \left(\betaHatT^m\right)^\top\ST\right) \SighHatdB \left(\vect{I}- 
      \ST^\top\betaHatT^m\right)\\
   =& \left(\betaHatT^m\right)^\top\vect{\Sigma}^d_{h,T}\betaHatT^m+
       \SighHatdB-\SighHatdB\ST^\top \betaHatT^m -\\
       & \left(\betaHatT^m\right)^\top\ST\SighHatdB+
   \left(\betaHatT^m\right)^\top\ST \SighHatdB  \ST^\top\betaHatT^m\\
   =& \left(\betaHatT^m\right)^\top\left(\vect{\Sigma}^d_{h,T}+\ST \SighHatdB  \ST^\top\right)\betaHatT^m+
       \SighHatdB-\\ &
       \SighHatdB\ST^\top \left(\betaHatT^m\right)^\top - (\betaHatT^m)^\top\ST\SighHatdB.
       \end{split}
\end{align}
It follows that the MAP estimate of $\Sigr$ can be written as  \citep[using][eq. (156) for the second equallity ]{petersen_Matrix2006}
\begin{align}
  \begin{split}
  \SigrHat=&
\frac{T/(1-\lambda)}{T+n-m}\left(\vect{P}(\lambda)\SigsHat\vect{P}^\top(\lambda)- 
    \lambda \vect{\hat{\Sigma}}^d_{h,B} + \right. \\ &\left.
    \lambda\SighHatdB\ST^\top(\SighHatdT+ \ST\SighHatdB\ST^\top)^{-1}\ST\SighHatdB
    \right) \\
       =& \frac{T/(1-\lambda)}{T+n-m}\left(\vect{P}(\lambda)\SigsHat\vect{P}^\top(\lambda)- 
    \lambda\left(\left(\SighHatdB\right)^{-1}+ \ST^\top\left(\SighHatdT\right)^{-1}\ST
    \right)^{-1}\right). \label{eq:proofSig.r}
    \end{split}
\end{align}
This is \eqref{eq:Sigma.rMap} and  concludes the proof of Corollary \ref{remark:MapSig}.\qed

\section{Proof of Theorem \ref{the:SigMap}}\label{proof:MAP-var}
If we assume that $\vect{\Sigma_r}\sim \mathcal{W}^{-1}(\vect{\Psi},v)$ and $v>m-1$, then the log-posterior density becomes
\begin{align}
  \begin{split}
  l_{\text{MAP}}(\Sigr)\propto & -
  \frac{1}{2}\left(\vect{y}-\YhatB-\vect{X}\betaHatT\right)^\top\left(\vect{I}_T\otimes\Sigr^{-1} \right)
  (\vect{y}-\YhatB-\vect{X}\betaHatT) -\\&\frac{1}{2}\left(\betaT-\betaOT\right)^\top\left(\vect{\Sigma}^{-1}_{r}\otimes\vect{\Sigma}^{-1}_{0,\beta}\right)(\betaT-\betaOT)-\\ 
   & \frac{n-m}{2}\log|\vect{\Sigma}_{r}| - 
  \frac{T}{2}\log|\Sigr|-\frac{v+m+1}{2}|\Sigr|-
  \frac{1}{2}\text{Tr}\left(\vect{\Psi}\Sigr^{-1}\right)\label{eq:proofMapVar},
  \end{split}
\end{align}
 and the derivative wrt.~$\Sigr$ is 
 \begin{align}
   \begin{split}
  \frac{\partial l_{\text{MAP}}(\Sigr)}{\partial \Sigr}=& -\frac{T+n+v+1}{2}\Sigr^{-1} +
   \frac{1}{2}\Sigr^{-1}\sum_{t=1}^T\vect{e}_t \vect{e}_t^\top\Sigr^{-1} +\\&   \frac{1}{2}\vect{\Sigma}^{-1}_{r}\left(\betaHatT^m-\betaOT^m\right)\vect{\Sigma}^{-1}_{0,\beta}\left(\betaHatT^m-\betaOT^m\right)^\top\vect{\Sigma}^{-1}_{r}-
  \frac{1}{2}\Sigr^{-1}\vect{\Psi}\Sigr^{-1}.
 \end{split}
\end{align}
Consequently, using \eqref{eq:proofSig.r} the MAP estimate of $\Sigr$ is 
\begin{align}
  \begin{split}
  \SigrHatMAP=&\frac{\left(\sum_{t=1}^T\vect{e}_t \vect{e}_t^\top +  \left(\betaHatT^m-\betaOT^m\right)\vect{\Sigma}^{-1}_{0,\beta}\left(\betaHatT^m-
  \betaOT^m\right)^\top+\vect{\Psi}\right)}{T+n+v+1}\\
 =&  \frac{T/(1-\lambda)}{T+n+v+1}\left(\vect{P}(\lambda)\SigsHat\vect{P}^\top(\lambda)- 
    \right. \\ &\left.
    \lambda\left(\left(\SighHatdB\right)^{-1}+ \ST^\top\left(\SighHatdT\right)^{-1}\ST
    \right)^{-1}+\frac{1-\lambda}{T}\vect{\Psi}\right).
    \end{split}
\end{align}
If we set $\vect{\Psi}=\frac{T\lambda}{1-\lambda}\left(\left(\SighHatdB\right)^{-1}+ \ST^\top\left(\SighHatdT\right)^{-1}\ST
    \right)^{-1}$, then the estimate is
\begin{align}
  \vect{\hat{\Sigma}}_{r,MAP}=&\frac{T/(1-\lambda)}{T+n+v+1}\vect{P}(\lambda)\SigsHat\vect{P}^\top(\lambda).
  \end{align}
If, in addition, we set $v=\frac{\lambda}{1-\lambda}T-(n+1)$, then the estimate becomes
\begin{align}
    \vect{\hat{\Sigma}}_{r,MAP}=&\vect{P}(\lambda)\SigsHat\vect{P}^\top(\lambda).
\end{align}
This is \eqref{eq:Sigma.rMap} and concludes the proof of Theorem \ref{the:SigMap}. \qed
\section{Proof of Corollary \ref{remark:RemlMap}}\label{proof:RemlMap}
When adding a REML correction term (the Hessian wrt.~$\vect{\beta}$ of the MAP-objective) to \eqref{eq:proofMapVar}, we get (using $\eB=\vect{y}-\yhatB$)
\begin{align}
  \begin{split}
  l_{\text{MAP,RE}}(\Sigr)\propto & -
  \frac{1}{2}(\eB-\vect{X}\betaT)^\top\left(\vect{I}_T\otimes\Sigr^{-1} \right)
  (\eB-\vect{X}\betaT) -\\
   &\frac{1}{2}(\betaT-\betaOT)^\top\left(\vect{\Sigma}^{-1}_{r}\otimes\vect{\Sigma}^{-1}_{0,\beta}\right)(\betaT-\betaOT)-\\ 
   & \frac{T+n+v+1}{2}\log|\vect{\Sigma}_{r}| - 
  \frac{1}{2}\text{Tr}(\vect{\Psi}\Sigr^{-1})+\\ 
  &\frac{1}{2}\log|\vect{X}^\top(\vect{I}_T\otimes\Sigr^{-1} )\vect{X}+\vect{\Sigma}^{-1}_{r}\otimes\vect{\Sigma}^{-1}_{0,\beta}|\label{eq:proofMapReVar}. 
  \end{split}
\end{align}

We assume that $\vect{x}_{i,t}=\vect{x}_{j,t}\in \mathbb{R}^{n-m}$. In this case the determinant in the REML correction term is (see \eqref{eq:XSigX})
\begin{align}
  |\vect{X}^\top(\vect{I}_T\otimes\Sigr^{-1} )\vect{X}+\vect{\Sigma}^{-1}_{r}\otimes\vect{\Sigma}^{-1}_{0,\beta}|=
  |\Sigr^{-1}|^{n-m}|\vect{X}_{1,\cdot}^\top\vect{X}_{1,\cdot}+\vect{\Sigma}^{-1}_{0,\beta}|^m 
\end{align}
and, hence,
\begin{align}
  \begin{split}
  l_{\text{MAP,RE}}(\Sigr)\propto & -
  \frac{1}{2}(\eB-\vect{X}\betaHatT)^\top\left(\vect{I}_T\otimes\Sigr^{-1} \right)
  (\eB-\vect{X}\betaHatT) -\\
   &\frac{1}{2}(\betaT-\betaOT)^\top\left(\vect{\Sigma}^{-1}_{r}\otimes\vect{\Sigma}^{-1}_{0,\beta}\right)(\betaT-\betaOT)-\\ 
   & \frac{T+m+v+1}{2}\log|\vect{\Sigma}_{r}| -  \frac{1}{2}tr(\vect{\Psi}\Sigr^{-1}). 
 \end{split}
  \end{align}
This gives the estimate (see Section \ref{proof:MAP-var})
\begin{align}
  \begin{split}
  \SigrHatsREML=&\frac{T/(1-\lambda)}{T+m+v+1}\left(\vect{P}(\lambda)\SigsHat\vect{P}^\top(\lambda)- 
    \right. \\ 
    &\left.
    \lambda\left(\left(\SighHatdB\right)^{-1}+ \ST^\top\left(\SighHatdT\right)^{-1}\ST
    \right)^{-1}+\frac{1-\lambda}{T}\vect{\Psi}\right).
    \end{split}
\end{align}
Inserting the values of $v$ and $\vect{\Psi}$ we get 
\begin{align}
  \begin{split}
  \SigrHatsREML=&\frac{T/(1-\lambda)}{T+m+v+1}\vect{P}(\lambda)\SigsHat\vect{P}^\top(\lambda)\\
  =&\frac{T}{T-(n-m)(1-\lambda)}\vect{P}(\lambda)\SigsHat\vect{P}^\top(\lambda).
    \end{split}
\end{align}
This is \eqref{eq:Sigma.rMap} and concludes the proof of Corollary \ref{remark:RemlMap}.\qed 
\revThree{\section{Singular $\Sigma_h$ in simulation example in Section \ref{sec:simStu}}\label{app:singularSig}}
Below, proof of the singularity claim in Section \ref{sec:simStu} is given. It should be stressed that these refer to the setup of the simulation example and are not general results that transfer to other base forecast models or data-generating processes. The results in Sections \ref{app:singularSig1} and \ref{app:singularSig2} use the forecasting model and the resulting summation matrix, while Section \ref{app:singularSig3} also uses the specific data-generating process used for simulation. 

\subsection{Sample variance--covariance}\label{app:singularSig1}
From a sample covariance point of view, if we consider the bottom-level observation matrix $\boldsymbol{Y}$, the base forecast error matrix can be written as
  \begin{align}
    \boldsymbol{Y}    \boldsymbol{S}^{\top}-\boldsymbol{\hat{Y}}=    \boldsymbol{Y}    \boldsymbol{S}^{\top}-\boldsymbol{Y}_{-T}\boldsymbol{S}^{\top}\boldsymbol{\Phi}
  \end{align}
  where $\boldsymbol{Y}_{-T}$ is the data matrix shifted one place, i.e.,
  \begin{align}
\boldsymbol{Y}=    \left[\begin{matrix}
        \boldsymbol{y}_{T}^\top\\
        \vdots\\
        \boldsymbol{y}_{1}^\top\\
  \end{matrix}\right];\quad
\boldsymbol{Y}_{-T}=    \left[\begin{matrix}
        \boldsymbol{y}_{T-1}^\top\\
        \vdots\\
        \boldsymbol{y}_{0}^\top\\
      \end{matrix}\right].
  \end{align}
  If $\boldsymbol{y}_{0}$ is not observed, it could be replaced by the unconditional expectation under the model. Next, consider the extended data matrix
  \begin{align}
\boldsymbol{Y}_{\text{ext}}=    \left[\begin{matrix}
        \boldsymbol{y}_{T}^{\top} &\boldsymbol{y}_{T-1}^\top \\
        \vdots & \vdots\\
       \boldsymbol{y}_{1}^{\top}& \boldsymbol{y}_{0}^\top\\
      \end{matrix}\right]\in \mathbb{R}^{T\times 2m}.    
  \end{align}
  The base forecast errors can be written as
  \begin{align}
    \boldsymbol{Y}    \boldsymbol{S}^{\top}-\boldsymbol{\hat{Y}}= \boldsymbol{Y}_{\text{ext}}
    \left[\begin{matrix}
        \boldsymbol{S}^\top \\ -\boldsymbol{S}^\top\boldsymbol{\Phi}
      \end{matrix}\right], 
  \end{align}
  hence, the ML estimate of $\boldsymbol{\hat{\Sigma}}_h$ is 
  \begin{align}
    \begin{split}
    \boldsymbol{\hat{\Sigma}}_h=&\frac{1}{T}(\boldsymbol{Y}    \boldsymbol{S}^{\top}-\boldsymbol{\hat{Y}})^\top(\boldsymbol{Y}    \boldsymbol{S}^{\top}-\boldsymbol{\hat{Y}})\\
   =& \frac{1}{T}    \left[\begin{matrix}
        \boldsymbol{S} & -\boldsymbol{\Phi}\boldsymbol{S}
      \end{matrix}\right]
        \boldsymbol{Y}_{\text{ext}}^\top\boldsymbol{Y}_{\text{ext}}
        \left[\begin{matrix}
        \boldsymbol{S}^\top \\ -\boldsymbol{\Phi}\boldsymbol{S}^\top
      \end{matrix}\right],
      \end{split}
  \end{align}
  which will be singular for $n>2m$. 
  
\subsection{LM point of view}\label{app:singularSig2}
  From the methodology presented in this paper, we need the base forecast matrix divided into top and bottom level,
    \begin{align}
      \hat{\vect{Y}}=\vect{Y}_{-T}\vect{S}^\top\vect{\Phi}=
      \vect{Y}_{-T}[\ST^\top\quad \vect{I}]\left[\begin{matrix}
          \vect{\Phi}_{\text{top}} &\vect{0}\\
          \vect{0} & \vect{\Phi}_{\text{bot}}
          \end{matrix}\right]
          =\left[\begin{matrix}
          \vect{\hat{Y}}_{\text{top}} & \vect{\hat{Y}}_{\text{bot}}
          \end{matrix}\right]
          .
    \end{align}
    The design matrix can therefore be written as (Theorem 2.2)
    \begin{align}
      \begin{split}
      \vect{X}=&\vect{Y}_{-T}\ST^\top\vect{\Phi}_{\text{top}}-\vect{Y}_{-T}\vect{\Phi}_{\text{bot}}\ST^\top\\
      =&      \vect{Y}_{-T}(\ST^\top\vect{\Phi}_{\text{top}}-\vect{\Phi}_{\text{bot}}\ST^\top).
        \end{split}
    \end{align}
    Hence, the rank of $\vect{X}$ cannot be greater than the rank of $\vect{Y}$, which in our case is $m=4$, and if the number of rows in $\ST$ is greater than $m=4$, then $\vect{X}^\top\vect{X}$ will be singular. 

\subsection{Second-order moment representation}\label{app:singularSig3}
With $\boldsymbol{y}_t$ and $\boldsymbol{\hat{y}}_{t|t-1}$ as in Section \ref{sec:simStu}, we have
  \begin{align}
    \begin{split}
    \boldsymbol{y}_{\text{full},t}=&\boldsymbol{S}\boldsymbol{y}_t\\
    \boldsymbol{\hat{y}}_{t}=&\boldsymbol{\Phi}\boldsymbol{S}\boldsymbol{y}_{t-1},
      \end{split}
    \end{align}
  where (in our case) we have  $\boldsymbol{\Phi}=diag([\phi_n,...\phi_1])$, with $\phi_i$ some estimate (we used the lag 1 auto-correlation of the $i$'th element of $\boldsymbol{y}_{\text{full},t}$ but could also be some known constants). In this setup, we have
  \begin{align}
    \begin{split}
    \boldsymbol{y}_{\text{full},t}-\boldsymbol{\hat{y}}_{t}=&\boldsymbol{S}(\boldsymbol{A}\boldsymbol{y}_{t-1}+\boldsymbol{\epsilon}_t)- \boldsymbol{\Phi}\boldsymbol{S}\boldsymbol{y}_{t-1}\\
    =&(\boldsymbol{S}\boldsymbol{A}-\boldsymbol{\Phi}\boldsymbol{S})\boldsymbol{y}_{t-1}+\boldsymbol{S}\boldsymbol{A}\boldsymbol{\epsilon}_t\\
    =&\left[\begin{matrix}
        \boldsymbol{S}\boldsymbol{A}-\boldsymbol{\Phi}\boldsymbol{S} & \boldsymbol{S}\boldsymbol{A}
      \end{matrix}\right]
    \left[\begin{matrix}
        \boldsymbol{y}_{t-1}\\
        \boldsymbol{\epsilon}_t
      \end{matrix}\right],
        \end{split}
  \end{align}
  with    $ \left[\begin{matrix}        \boldsymbol{y}_{t-1} &         \boldsymbol{\epsilon}_t  \end{matrix}\right]^\top\in \mathbb{R}^{2m}$. The base forecast error variance--covariance matrix $\boldsymbol{\Sigma}_h=V[\boldsymbol{y}_{\text{full},t}-\boldsymbol{\hat{y}}_{t|t-1}]$ will be singular if $n>2m$.

\clearpage
\newpage

\section{Additional plots for the case study}\label{app:case}

\begin{knitrout}
\definecolor{shadecolor}{rgb}{0.969, 0.969, 0.969}\color{fgcolor}\begin{figure}[p]

  {\centering
    \includegraphics[width=\maxwidth]{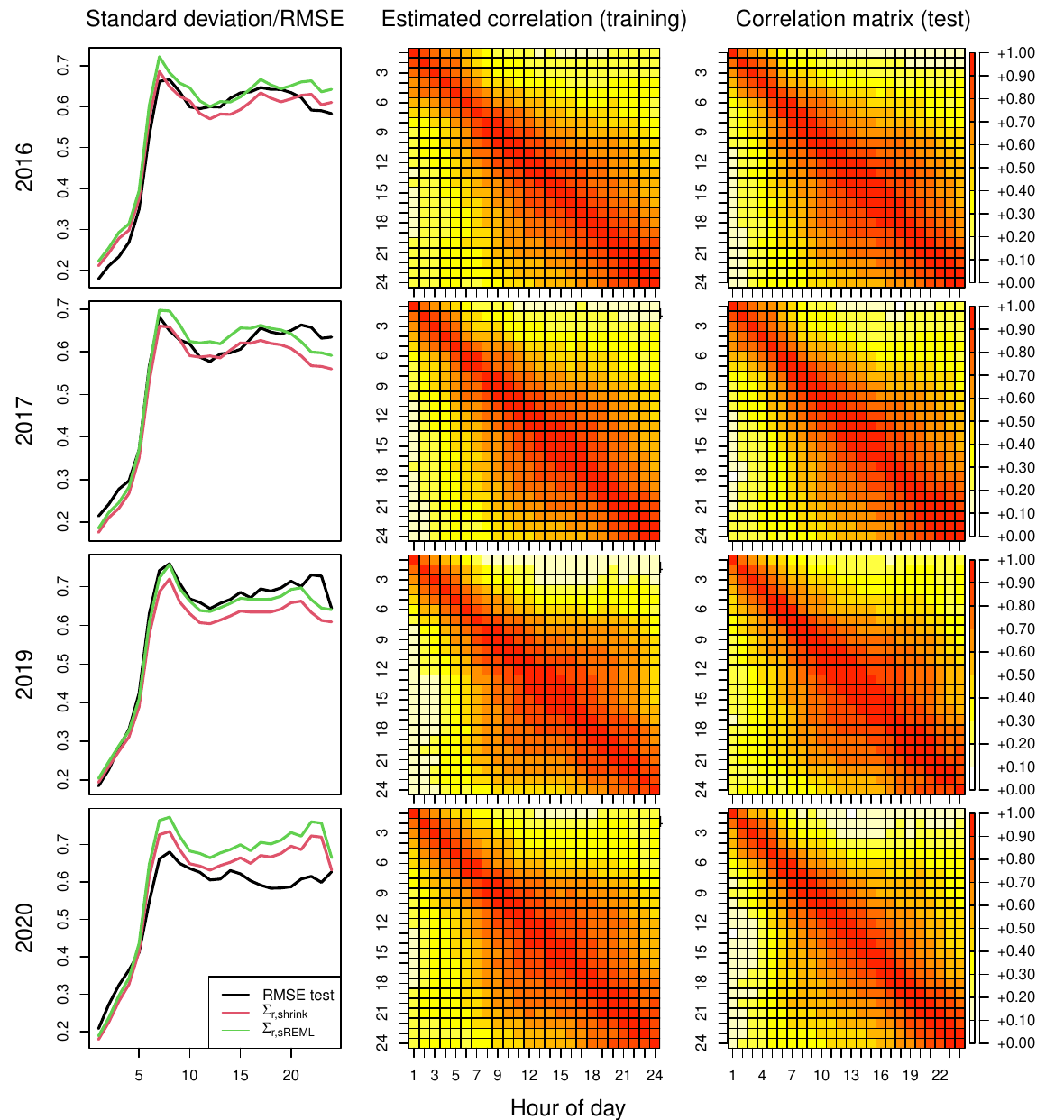} 

}

\caption[Standard deviation and correlation matrix corresponding to $\Sigr$ for different periods]{Standard deviation and correlation matrix corresponding to $\Sigr$ for different test periods, for further explanation see Figure \ref{fig:SweSigr}.}\label{fig:SweSigrSE1}
\end{figure}

\end{knitrout}

\begin{landscape}
  
\begin{knitrout}
\definecolor{shadecolor}{rgb}{0.969, 0.969, 0.969}\color{fgcolor}\begin{figure}[p]

  {\centering
    \includegraphics[width=1\linewidth,height=0.75\textheight]{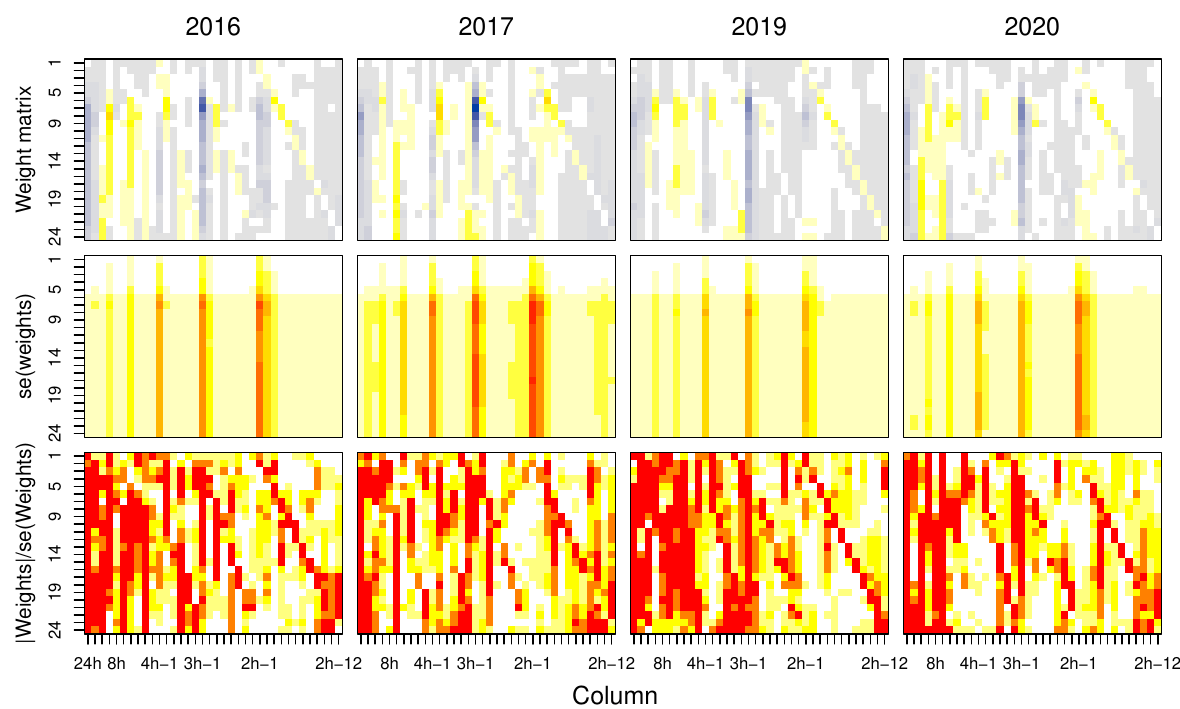} 

}

\caption{Volume-weighted weight matrices, associated standard errors, and absolute value of Wald test statics. Blue colors represent negative values and red/yellow represent positive values. Color coding is the same as in Figure \ref{fig:PlotWeights}. The years in the heading of the figures refer to the year used as test set.}\label{fig:PlotWeightsSE1_4}
\end{figure}

\end{knitrout}

 \end{landscape}

\clearpage

\begin{knitrout}
\definecolor{shadecolor}{rgb}{0.969, 0.969, 0.969}\color{fgcolor}\begin{figure}[p]
{\centering \includegraphics[width=\maxwidth]{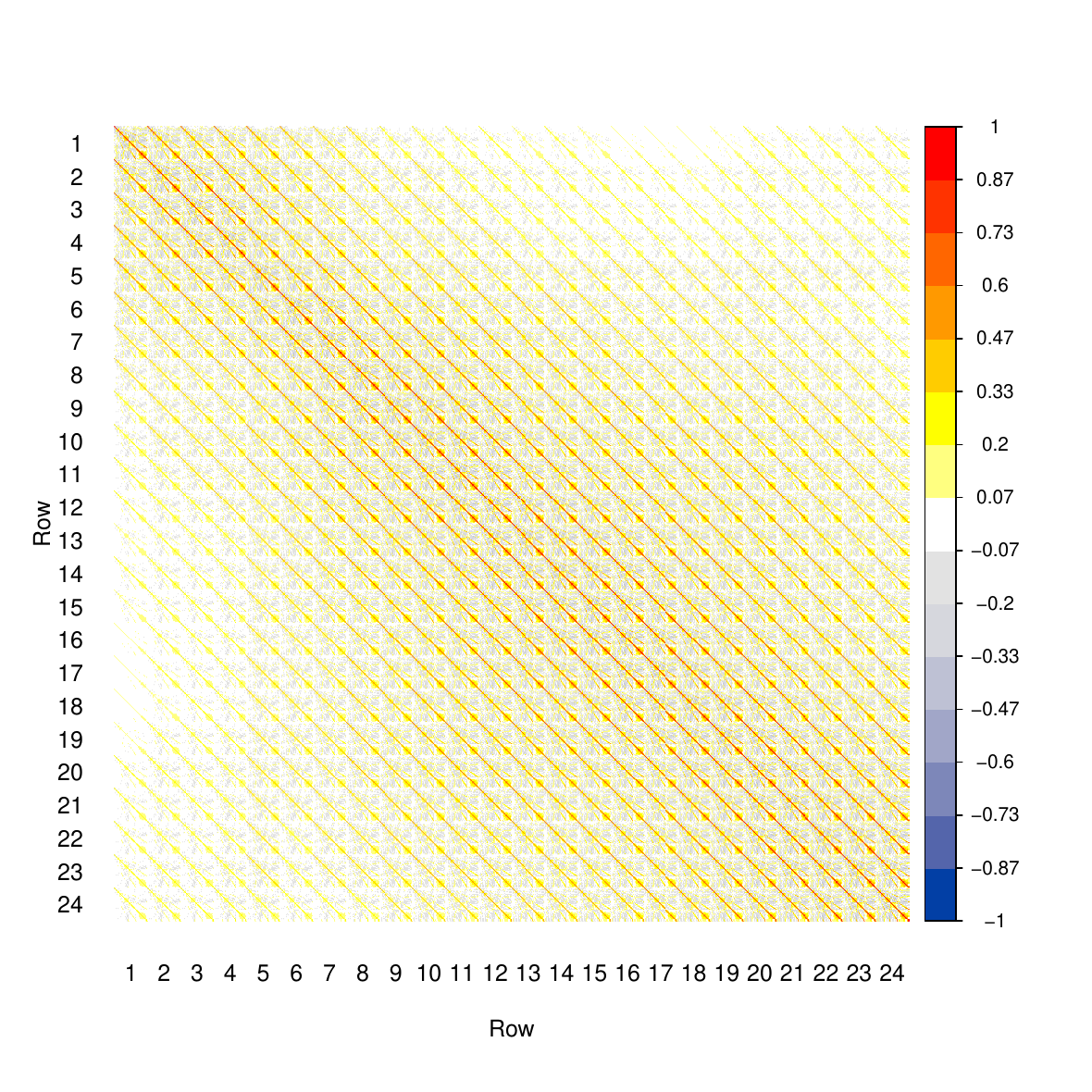} 
}

\caption[Parameter correlation for a selection of all rows in the weight matrix]{Parameter correlation of all parameters in the weight matrix when using 2018 as test period. The correlation patterns within rows (diagonal blocks) are constant while the between rows correlation matrices (off-diagonal blocks) display decreasing correlation with distance so that, e.g., row 1 and row 24 are weakly correlated while the correlation between row $n$ and row $n-1$ is strong. The within-row correlation is more clear in Figures \ref{fig:SweCorrPlot} and \ref{fig:SweCorrPlotSE1Row1}.} \label{fig:SweCorrPlotFull}
\end{figure}

\end{knitrout}

\clearpage
\newpage

\clearpage
\newpage

\begin{knitrout}
\definecolor{shadecolor}{rgb}{0.969, 0.969, 0.969}\color{fgcolor}\begin{figure}[p]

  {\centering
    \includegraphics[width=\maxwidth]{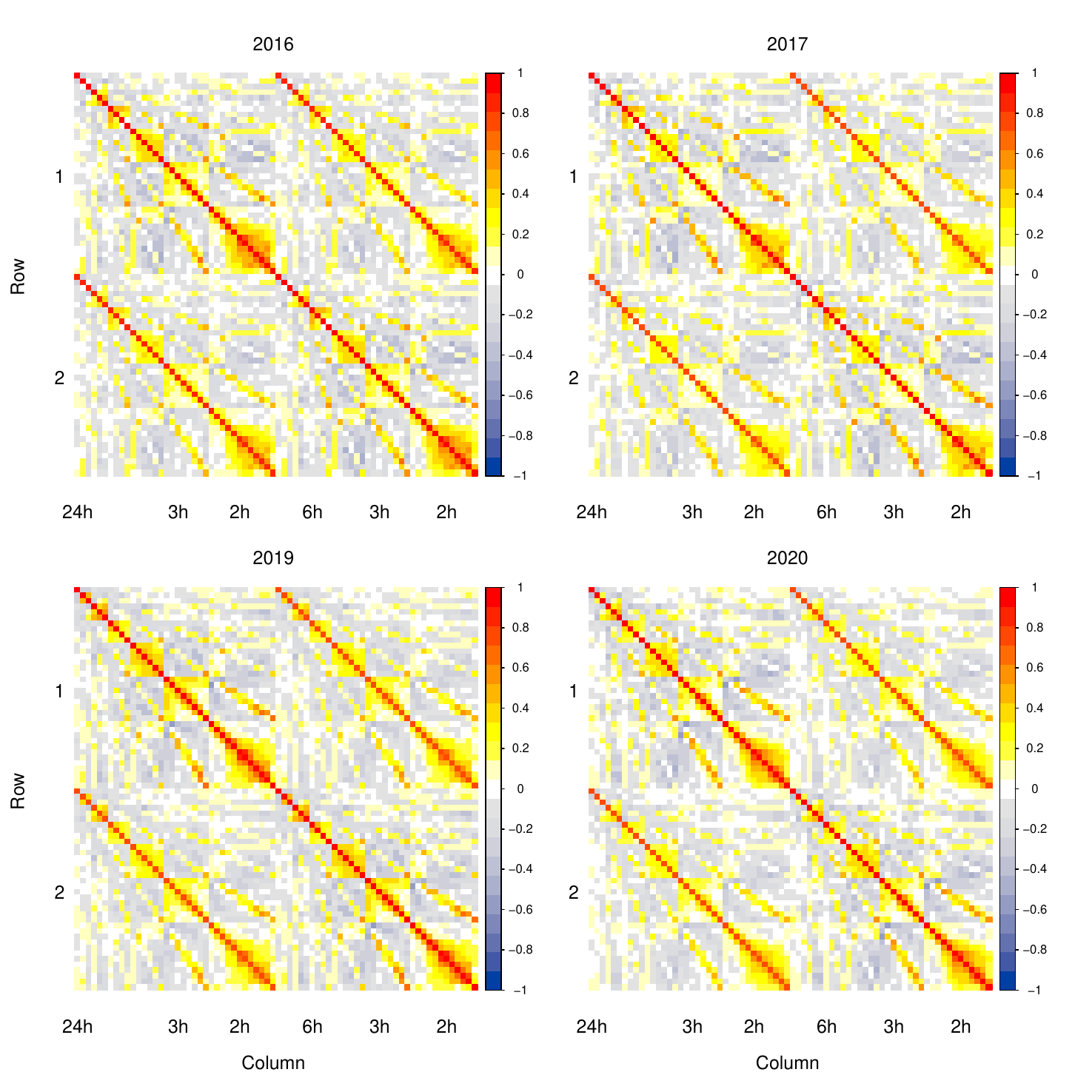}
}
\caption[Parameter correlation for rows 1--2 of the weight matrix ($\PT$)]{Parameter correlation for rows 1--2 of the weight matrix ($\PT$) different test periods. The correlation between levels of the forecasts (e.g., 2h and 3h.) is generally weak as is the correlation between forecast far away in time (e.g., 2h-1 and 2h-12 for row 1). In contrast, the correlation between the same weight in different rows is generally strong.}\label{fig:SweCorrPlotSE1Row1}
\end{figure}

\end{knitrout}

\clearpage
\newpage

\section{Variance separation tables}\label{sec:projTabs}

\begin{table}[h]
\centering
\begingroup\small
\begin{tabular}{l|r|rr|r|rrr}
  & $\frac{||\vect{y}_I-\vect{\hat{y}}_I||^2}{T}$ & $\frac{||\vect{y}_I-\vect{\tilde{y}}_I||^2}{T}$ & $\frac{||\vect{\tilde{y}}_I-\vect{\hat{y}}_I||^2}{T}$ & $\frac{ \frac{\sum\vect{SS}_{mod}}{n-m}}{\frac{\sum \vect{SSE}}{T-n+m}}$ & $\frac{||\vect{y}_I-\vect{\tilde{y}}_I^{\lambda}||^2}{T}$ & $\frac{||\vect{\tilde{y}}_I^{\lambda}-\vect{\hat{y}}_I||^2}{T}$ & $\frac{2(\vect{\tilde{e}}_I^{\lambda})^T\vect{{\tilde{\hat{e}}}}_I^{\lambda}}{T}$ \\ 
  \hline
24h & 212.44 & 87.685 & 124.753 & 13.00 & 104.593 & 96.076 & 11.769 \\ 
  12h & 93.91 & 54.356 & 39.553 & 7.56 & 64.477 & 25.544 & 3.889 \\ 
  8h & 60.42 & 39.984 & 20.434 & 5.26 & 47.425 & 10.262 & 2.731 \\ 
  6h & 48.04 & 32.438 & 15.606 & 5.02 & 38.429 & 7.383 & 2.232 \\ 
  4h & 33.01 & 22.650 & 10.357 & 4.93 & 26.802 & 4.680 & 1.525 \\ 
  3h & 25.46 & 17.437 & 8.026 & 4.97 & 20.628 & 3.579 & 1.256 \\ 
  2h & 17.48 & 11.985 & 5.495 & 4.91 & 14.170 & 2.430 & 0.881 \\ 
  1h & 9.51 & 6.209 & 3.299 & 5.44 & 7.332 & 1.726 & 0.451 \\ 
   \hline
Total & 500.27 & 272.745 & 227.524 & 5.42 & 323.855 & 151.679 & 24.735 \\ 
  \end{tabular}
\endgroup
\caption{
\label{tab:projTabSE1} Variance separation \eqref{eq:proj2} on the training set based on the forecasts for 2016 (i.e., training period is 2017). 
For explanation see Table \ref{tab:projTabSE}.} 
\end{table}

\begin{table}[h]
\centering
\begingroup\small
\begin{tabular}{l|r|rr|r|rrr}
  & $\frac{||\vect{y}_I-\vect{\hat{y}}_I||^2}{T}$ & $\frac{||\vect{y}_I-\vect{\tilde{y}}_I||^2}{T}$ & $\frac{||\vect{\tilde{y}}_I-\vect{\hat{y}}_I||^2}{T}$ & $\frac{ \frac{\sum\vect{SS}_{mod}}{n-m}}{\frac{\sum \vect{SSE}}{T-n+m}}$ & $\frac{||\vect{y}_I-\vect{\tilde{y}}_I^{\lambda}||^2}{T}$ & $\frac{||\vect{\tilde{y}}_I^{\lambda}-\vect{\hat{y}}_I||^2}{T}$ & $\frac{2(\vect{\tilde{e}}_I^{\lambda})^T\vect{{\tilde{\hat{e}}}}_I^{\lambda}}{T}$ \\ 
  \hline
24h & 174.69 & 86.278 & 88.407 & 8.65 & 101.419 & 65.510 & 7.756 \\ 
  12h & 85.04 & 52.746 & 32.291 & 5.62 & 61.567 & 20.448 & 3.022 \\ 
  8h & 54.44 & 38.470 & 15.968 & 3.68 & 45.042 & 7.545 & 1.850 \\ 
  6h & 42.78 & 31.439 & 11.339 & 3.20 & 36.638 & 4.485 & 1.654 \\ 
  4h & 29.38 & 22.129 & 7.255 & 2.93 & 25.734 & 2.584 & 1.067 \\ 
  3h & 22.74 & 17.055 & 5.690 & 2.99 & 19.825 & 2.001 & 0.919 \\ 
  2h & 15.99 & 11.741 & 4.252 & 3.37 & 13.650 & 1.674 & 0.668 \\ 
  1h & 9.16 & 6.091 & 3.072 & 4.97 & 7.069 & 1.758 & 0.336 \\ 
   \hline
Total & 434.22 & 265.947 & 168.275 & 3.56 & 310.945 & 106.005 & 17.273 \\ 
  \end{tabular}
\endgroup
\caption{
\label{tab:projTabSE2} Variance separation \eqref{eq:proj2} on the training set based on the forecasts for 2017 (i.e., training period is 2016). 
For explanation see Table \ref{tab:projTabSE}.
} 
\end{table}

\begin{table}[]
\centering
\begingroup\small
\begin{tabular}{l|r|rr|r|rrr}
  & $\frac{||\vect{y}_I-\vect{\hat{y}}_I||^2}{T}$ & $\frac{||\vect{y}_I-\vect{\tilde{y}}_I||^2}{T}$ & $\frac{||\vect{\tilde{y}}_I-\vect{\hat{y}}_I||^2}{T}$ & $\frac{ \frac{\sum\vect{SS}_{mod}}{n-m}}{\frac{\sum \vect{SSE}}{T-n+m}}$ & $\frac{||\vect{y}_I-\vect{\tilde{y}}_I^{\lambda}||^2}{T}$ & $\frac{||\vect{\tilde{y}}_I^{\lambda}-\vect{\hat{y}}_I||^2}{T}$ & $\frac{2(\vect{\tilde{e}}_I^{\lambda})^T\vect{{\tilde{\hat{e}}}}_I^{\lambda}}{T}$ \\ 
  \hline
24h & 209.5 & 89.04 & 120.44 & 12.36 & 102.97 & 95.62 & 10.89 \\ 
  12h & 93.3 & 55.43 & 37.85 & 7.19 & 63.66 & 25.97 & 3.65 \\ 
  8h & 59.6 & 40.54 & 19.06 & 4.78 & 46.90 & 10.31 & 2.38 \\ 
  6h & 47.4 & 32.90 & 14.54 & 4.56 & 38.03 & 7.43 & 1.98 \\ 
  4h & 32.8 & 22.97 & 9.80 & 4.58 & 26.53 & 4.86 & 1.38 \\ 
  3h & 25.3 & 17.67 & 7.60 & 4.60 & 20.42 & 3.68 & 1.16 \\ 
  2h & 17.3 & 12.14 & 5.15 & 4.49 & 14.03 & 2.45 & 0.80 \\ 
  1h & 9.3 & 6.28 & 3.03 & 4.89 & 7.27 & 1.65 & 0.41 \\ 
   \hline
Total & 494.4 & 276.98 & 217.47 & 5.01 & 319.81 & 151.97 & 22.66 \\ 
  \end{tabular}
\endgroup
\caption{
\label{tab:projTabSE3} Variance separation \eqref{eq:proj2} on the training set based on the forecasts for 2019 (i.e., training period is 2018).
For explanation see Table \ref{tab:projTabSE}.} 
\end{table}

\begin{table}[h]
\centering
\begingroup\small
\begin{tabular}{l|r|rr|r|rrr}
  & $\frac{||\vect{y}_I-\vect{\hat{y}}_I||^2}{T}$ & $\frac{||\vect{y}_I-\vect{\tilde{y}}_I||^2}{T}$ & $\frac{||\vect{\tilde{y}}_I-\vect{\hat{y}}_I||^2}{T}$ & $\frac{ \frac{\sum\vect{SS}_{mod}}{n-m}}{\frac{\sum \vect{SSE}}{T-n+m}}$ & $\frac{||\vect{y}_I-\vect{\tilde{y}}_I^{\lambda}||^2}{T}$ & $\frac{||\vect{\tilde{y}}_I^{\lambda}-\vect{\hat{y}}_I||^2}{T}$ & $\frac{2(\vect{\tilde{e}}_I^{\lambda})^T\vect{{\tilde{\hat{e}}}}_I^{\lambda}}{T}$ \\ 
  \hline
24h & 226.08 & 110.771 & 115.304 & 9.51 & 131.084 & 84.812 & 10.180 \\ 
  12h & 103.32 & 66.543 & 36.780 & 5.50 & 78.241 & 22.145 & 2.937 \\ 
  8h & 68.31 & 49.301 & 19.007 & 3.81 & 58.116 & 8.193 & 1.998 \\ 
  6h & 54.90 & 39.365 & 15.534 & 3.89 & 46.411 & 6.463 & 2.026 \\ 
  4h & 38.78 & 27.531 & 11.248 & 4.18 & 32.479 & 4.971 & 1.328 \\ 
  3h & 30.39 & 21.252 & 9.141 & 4.50 & 25.038 & 4.174 & 1.181 \\ 
  2h & 19.88 & 14.546 & 5.332 & 3.63 & 17.144 & 1.967 & 0.766 \\ 
  1h & 10.51 & 7.541 & 2.973 & 3.89 & 8.883 & 1.244 & 0.387 \\ 
   \hline
Total & 552.17 & 336.850 & 215.318 & 3.95 & 397.395 & 133.969 & 20.804 \\ 
  \end{tabular}
\endgroup
\caption{
\label{tab:projTabSE4} Variance separation \eqref{eq:proj2} on the training set based on the forecasts for 2020 (i.e., training period is 2019).
For explanation see Table \ref{tab:projTabSE}.} 
\end{table}

\end{document}